\documentclass[10pt,journal]{IEEEtran}

\usepackage{amsmath,amsthm}
\usepackage{color}
\usepackage{algorithm}
\usepackage{tikz}
\usepackage{graphicx}
\usepackage{mathtools}
\usepackage{hyperref}
\usepackage{amsfonts}
\usetikzlibrary{patterns}
\usepackage[noadjust]{cite}
\usepackage{filecontents}
\usepackage{nicefrac}

\ifCLASSOPTIONcompsoc
\usepackage[caption=false,font=normalsize,labelfont=sf,textfont=sf]{subfig}
\else
\usepackage[caption=false,font=footnotesize]{subfig}
\fi

\title{\Huge \bf Formal Synthesis of Control Strategies for Positive Monotone Systems}

\author{Sadra Sadraddini,~\IEEEmembership{Student Member,~IEEE}
 and Calin Belta,~\IEEEmembership{Fellow,~IEEE} %
\thanks{The authors are with the Department of Mechanical Engineering, Boston University, Boston, MA 02215 \texttt{sadra,cbelta@bu.edu}.
\newline This work was partially supported by the NSF under grants CPS-1446151 and CMMI-1400167.}
}

\theoremstyle{definition}
\newtheorem{problem}{Problem}
\theoremstyle{remark}
\newtheorem{define}{Definition}
\newtheorem{lemma}{Lemma}
\newtheorem{example}{Example}
\newtheorem{remark}{Remark}
\newtheorem{assumption}{Assumption}
\newtheorem{theorem}{Theorem}
\newtheorem{corollary}{Corollary}
\newtheorem{proposition}{Proposition}

\newcommand{\bolds}[1]{ {\bf{{#1}}}}

\newcommand{\user}[1]{{#1}}
\newcommand{\userfinal}[1]{{\color{black}{#1}}}

\newcommand{\revone}[1]{{#1}}
\newcommand{\revtwo}[1]{{#1}}
\newcommand{\revthree}[1]{{#1}}

\DeclareMathOperator*{\argmin}{arg\,min}
\DeclareMathOperator*{\rem}{rem}

\begin{document}

\maketitle

\begin{abstract}
We design controllers from formal specifications for positive discrete-time monotone systems that are subject to bounded disturbances. Such systems are widely used to model the dynamics of transportation and biological networks. The specifications are described using signal temporal logic (STL), which can express a broad range of temporal properties. We formulate the problem as a mixed-integer linear program (MILP) and show that under the assumptions made in this paper, which are not restrictive for traffic applications, the existence of open-loop control policies is sufficient and almost necessary to ensure the satisfaction of STL formulas. We establish a relation between satisfaction of STL formulas in infinite time and set-invariance theories and provide an efficient method to compute robust control invariant sets in high dimensions. We also develop a robust model predictive framework to plan controls optimally while ensuring the satisfaction of the specification. Illustrative examples and a traffic management case study are included.
\end{abstract}


\section{Introduction}


\IEEEPARstart{I}{n} recent years, there has been a growing interest in using formal methods for specification, verification,  and synthesis in control theory. Temporal logics \cite{baier2008principles} provide a rich, expressive framework for describing a broad range of properties such as safety, liveness, and reactivity. In formal synthesis, the goal is to control a dynamical system from such a specification. For example, in an urban traffic network, a synthesis problem can be to generate traffic light control policies that ensure gridlock avoidance and fast enough traffic through a certain road, for all times. 


Control synthesis for linear and piecewise affine systems from linear temporal logic (LTL) specifications was studied in \cite{tabuada2006linear,kloetzer2008fully,Yordanov2012}. The automata-based approach used in these works requires constructing finite abstractions that (bi)simulate the original system. Approximate finite bisimulation quotients for nonlinear systems were investigated in \cite{pola2009symbolic,zamani2012symbolic}. The main limitations of finite abstraction approaches are the large computational burden of discretization in high dimensions and conservativeness when exact bisimulations are impossible or difficult to construct. As an alternative approach, LTL optimization-based control of mixed-logical dynamical (MLD) systems \cite{Bemporad1999} using mixed-integer programs was introduced in \cite{karaman,Wolff2014}, and was recently extended to model predictive control (MPC) from signal temporal logic (STL) specifications in \cite{raman2014model,raman2015reactive,sadraddini2015robust}. However, these approaches are unable to guarantee infinite-time safety and the results are fragile in the presence of non-deterministic disturbances. 
 

In some applications, the structural properties of the system and the specification can be exploited to consider alternative approaches to formal control synthesis. We are interested in systems in which the evolution of the state exhibits a type of \emph{order preserving law} known as \emph{monotonicity}, which is common in models of transportation, biological, and economic systems \cite{may2007theoretical,como2015throughput,coogan2015mixed,EricS.KimMuratArcak2016}. Such systems are also \emph{positive} in the sense that the state components are always non-negative. Control of positive systems have been widely studied in the literature \cite{haddad2010nonnegative,rantzer2011distributed,de2001stabilization}. Positive linear systems are always monotone \cite{rantzer2012distributed}. 

In this paper, we study optimal STL control of discrete-time positive monotone systems (i.e., systems with state partial order on the positive orthant) with bounded disturbances. STL allows designating time intervals for temporal operators, which makes it suitable for describing requirements with deadlines. Moreover, STL is equipped with quantitative semantics, which provides a measure to quantify how strongly the specification is satisfied/violated. The quantitative semantics of STL can also be used as cost 
for maximization in an optimal control setting. The STL specifications in this paper are restricted to a particular form that favors smaller values for the state components. We assume that there exists a maximal disturbance element that characterizes a type of upper-bound for the evolution of the system. These assumptions are specifically motivated by the dynamics of traffic networks, where the disturbances represent the volume of exogenous vehicles entering the network and the maximal disturbance characterizes the rush hour exogenous flow.  Our optimal control study is focused on STL formulae with infinite-time safety/persistence properties, which is relevant to optimal and correct traffic control in the sense that the vehicular flow is always free of congestion while the associated delay is minimized. 


The key contributions of this paper are as follows. First, for finite-time semantics, we prove that the existence of open-loop control policies is necessary and sufficient for maintaining STL correctness. 
For the correctness of infinite-time semantics, we show that the existence of open-loop control sequences is sufficient and almost necessary, in a sense that is made clear in the paper. Implementing open-loop control policies is very simple since online state measurements are not required, which can prove useful in applications where the state is difficult to access. We use a robust MPC approach to optimal control. The main contribution of our MPC framework is guaranteed recursive feasibility, a property that was not established in prior STL MPC works \cite{raman2014model,raman2015reactive,sadraddini2015robust}. We show via a case study that our method is applicable to systems with relatively high dimensions.

This remainder of the paper is organized as follows. We introduce the necessary notation and background on STL in Sec. \ref{sec:prelim}. The problems are formulated in Sec. \ref{sec:problem}. The technical details for control synthesis from finite and infinite-time specifications are given in Sec. \ref{sec:finite} and Sec. \ref{sec:infinite}, respectively. The robust MPC framework is explained in Sec. \ref{sec:mpc}. Finally, we introduce a traffic network model and explain its monotonicity properties in Sec. \ref{sec:traffic}, where a case study is also presented.

\subsection*{Related Work}
This paper is an extension of the conference version \cite{sadraddini2016safety}, where we studied safety control of positive monotone systems. Here, we significantly enrich the range of specifications to STL, provide complete proofs, and include optimal control. 

Monotone dynamical systems have been extensively investigated  in the mathematics literature \cite{hirsch1985systems,hirsch2005monotone,smith2008monotone}.  
Early studies mainly focused on stability properties and characterization of limit sets for autonomous, deterministic continuous-time systems \cite{hadeler1983quasimonotone,hirsch1985systems}. The results do not generally  hold for discrete-time systems, as discussed in \cite{hirsch2005monotone}. In particular, attractive periodic orbits are proven to be non-existent for continuous-time autonomous systems \cite{hadeler1983quasimonotone}, but may exist for discrete-time autonomous systems. Here we present a similar result for controlled systems, where we show that a type of attractive periodic orbit exists for certain control policies.

Angeli and Sontag \cite{angeli2003monotone} extended the notion of monotonicity to deterministic continuous-time control systems and provided results on interconnections of these systems. However, they assumed monotonicity with respect to both state and controls. We do not require monotonicity with respect to controls, which enables us to consider a broader class of systems. In particular, we do not require controls to belong to a partially ordered set. 

Switching policies for exponential stabilization of switched positive linear systems were studied in  \cite{blanchini2012co,fornasini2012stability}. Stabilization is closely related to set-invariance, which is thoroughly studied in this paper. Apart from richer specifications, we can handle more complex systems. We consider hybrid systems in which the mode is either determined directly by the control input or indirectly by the state \user{(e.g., signalized traffic networks).}

Recently, there has been some interest on formal verification and synthesis for monotone systems. Safety control of cooperative systems was investigated in  \cite{hafner2011computational,ghaemi2014control,meyer2016robust}. However, these work, like \cite{angeli2003monotone}, assumed monotonicity with respect to the control inputs as well. Computational benefits gained from monotonicity for reachability analysis of hybrid systems were highlighted in \cite{ramdani2010computing}. More recently, the authors in \cite{coogan2015efficient} provided an efficient method to compute finite abstractions for mixed-monotone systems (a more general class than monotone systems). \user{The authors in \cite{kim2017symbolic} exploited monotonicity to compute finite-state abstractions that are used for compositional LTL control.} 
While the approaches in \cite{coogan2015efficient,kim2017symbolic} can consider systems and specifications beyond the assumptions in this paper, they still require state-space discretization, which is a severe limitation in high dimensions. Moreover, they are conservative since the finite abstractions are often not bisimilar with the original system - whereas our approach provides a notion of (almost) completeness. \user{Finally, as opposed to the all mentioned works, our framework is amenable to optimal temporal logic control.}

\section{Preliminaries}
\label{sec:prelim}
\subsection{Notation}
\label{sec:notation}
For two integers $a,b$, we use $\rem(a,b)$ to denote the remainder of division of $a$ by $b$. Given a set $\mathcal{S}$ and a positive integer $K$, we use the shorthand notation $\mathcal{S}^K$ for $\prod_{i=1}^K \mathcal{S}$. A signal is defined as an infinite sequence $\bolds{s}=s_0s_1\cdots$, where $s_k \in \mathcal{S}$, $k \in \mathbb{N}$.  
\user{Given $s_1,s_2,\cdots,s_K\in \mathcal{S}$, the repetitive infinite-sequence $s_1s_2 \cdots s_K s_1 s_2 \cdots s_K\cdots$ is denoted by $(s_1 s_2 \cdots s_K)^\omega$.} The set of all signals that can be generated from $\mathcal{S}$ is denoted by $\mathcal{S}^\omega$. We use $\bolds{s}[k]={s}_k{s}_{k+1}\cdots$ and $\bolds{s}[{k_1:k_2}]={s}_{k_1}{s}_{k_1+1}\cdots{s}_{k_2}$,  $k_1 < k_2$, to denote specific portions of $\bolds{s}$.
A \emph{real} signal is $\bolds{\bolds{r}}={r}_0{r}_1{r}_2\cdots$, where ${r}_k \in \mathbb{R}^n, \forall k \in \mathbb{N}$. A vector of all ones in $\mathbb{R}^n$ is denoted by $1_n$. \revone{We use the notation $\bolds{1}_{n}[0:K]:= {1_n \cdots 1_n }$, where $1_n$ is repeated $K+1$ times}.
The positive closed orthant of the $n$-dimensional Euclidian space is denoted by $\mathbb{R}^n_+ := \left \{ x \in \mathbb{R}^n| x_{[i]} \ge 0, i=1,\cdots,n \right \}$, where $x=(x_{[1]},x_{[2]},\cdots,x_{[n]})^T$. For $a,b \in \mathbb{R}^n$, the non-strict partial order relation $\preceq$ is defined as:
$a \preceq b \Leftrightarrow b-a \in \mathbb{R}_+^n.$
\begin{define}[\cite{davey2002introduction}]
A set $\mathcal{X} \subset \mathbb{R}^n_+$ is a \emph{lower-set} if $\forall x \in \mathcal{X},L(x) \subseteq \mathcal{X}$, where $L(x):= \left\{ y \in \mathbb{R}_+^n \big| y \preceq x \right\}.
$
\end{define}
It is straightforward to verify that if $\mathcal{X}_1,\mathcal{X}_2$ are lower-sets, then $\mathcal{X}_1 \cup \mathcal{X}_2$ and $\mathcal{X}_1 \cap \mathcal{X}_2$ are also lower-sets.
 We extend the usage of notation $\preceq$ to equal-length real signals. For two real signals $\bolds{r}, \bolds{r}$, we denote $\bolds{r}'[t_1':t_2'] \preceq \bolds{r}[t_1:t_2]$, $t_2-t_1=t'_2-t'_1$, if $r'_{t'_1+k} \preceq r_{t_1+k}, k=0,1,\cdots,t_2-t_1$. \user{
Moreover, if $\bolds{r}, \bolds{r}' \in ({\mathbb{R}_+^n})^\omega$, we are also allowed to write  $\bolds{r}'[t_1':t_2'] \in L(\bolds{r}[t_1:t_2])$.}

\subsection{Signal Temporal Logic (STL)}

\label{sec:NF-STL}
In this paper, STL \cite{maler_stl} formulas are defined over discrete-time real signals. 
The syntax of negation-free STL is:
\begin{equation}
\label{eq:STL_syntax}
\varphi:= ~\pi~ | ~ \varphi_1 \wedge \varphi_2 ~ |~ \varphi_1 \vee \varphi_2 ~|~ \varphi_1 {\bf U}_{I} \varphi_2 ~|~ \bolds{F}_I \varphi ~|~ \bolds{G}_I \varphi,
\end{equation}  
where $\pi=(p({r}) \le c)$ is a predicate on ${r} \in \mathbb{R}^n$, $p: \mathbb{R}^n \rightarrow \mathbb{R}$, $c \in \mathbb{R}$; $\wedge$ and $\vee$ are Boolean connectives for conjunction and disjunction, respectively; ${\bf U}_{I}$, $\bolds{F}_I $, $\bolds{G}_I$ are the timed \emph{until}, \emph{eventually} and \emph{always} operators, respectively, and $I=[t_1,t_2]$ is a time interval,  $t_1,t_2 \in \mathbb{N}\cup \{\infty\}, t_2\ge t_1$. When $t_1=t_2$, we use the shorthand notation $\left\{t_1\right\}:=[t_1,t_1]$. Exclusion of negation does not restrict expressivity of temporal properties. It can be easily shown that any temporal logic formula can be brought into \emph{negation normal form} (where all negation operators apply to the predicates) \cite{ouaknine2008some,sadraddini2015robust}. 
We deliberately omit negation from STL syntax for laying out properties that are later exploited in the paper. For simplicity, in the rest of the paper, we will refer to negation-free STL simply as STL. 
The semantics of STL is inductively defined as:
\begin{equation}
\label{equ:semantics}
\begin{array}{lll}
\bolds{r}[t] \models \pi & \Leftrightarrow & p({r}_t) \le c, 
\\
\bolds{r}[t] \models  \varphi_1 \vee \varphi_2 & \Leftrightarrow & \bolds{r}[t] \models \varphi_1 ~\vee~\bolds{r}[t]  \models \varphi_2,
\\
\bolds{r}[t]  \models \varphi_1 \wedge \varphi_2 & \Leftrightarrow & \bolds{r}[t] \models \varphi_1 ~\wedge~\bolds{r}[t]  \models \varphi_2,
\\
\bolds{r}[t]  \models  \varphi_1~{\bf U}_{I}~ \varphi_2 & \Leftrightarrow & \exists t^\prime \in t+I 
~{s.t}~  \bolds{r}[{t^\prime}] \models \varphi_2 \\ &&\wedge~ \forall t^{\prime\prime} \in [t,t^\prime],   \bolds{r}[{t^{\prime\prime}}] \models \varphi_1,
\\
\bolds{r}[t]  \models  {\bf F}_I \varphi & \Leftrightarrow & \exists  t^\prime \in t+I ~s.t.~ \bolds{r}[t'] \models \varphi, \\
\bolds{r}[t]   \models  {\bf G}_I \varphi & \Leftrightarrow & \forall  t^\prime \in t+I ~s.t.~ \bolds{r}[t'] \models \varphi,
\end{array}
\end{equation} 
where $\models$ is read as \emph{satisfies}. The \emph{language} of $\varphi$ is the set of all signals such that $\bolds{r}[0] \models \varphi$. The \emph{horizon} of an STL formula $\varphi$, denoted by $h^\varphi$, is defined as the time required to decide the satisfaction of $\varphi$, which is recursively computed as \cite{dokhanchi}:
\begin{equation}
\label{eq:horizon}
\begin{array}{rl}
h^\pi=&0, \\
h^{\varphi_1 \wedge \varphi_2}=h^{\varphi_1 \vee \varphi_2}=&\max(h^{\varphi_1},h^{\varphi_2}), \\
h^{{\bf F}_{[t_1,t_2]}\varphi}=h^{{\bf G}_{[t_1,t_2]}\varphi}=& t_2+h^{\varphi},\\
h^{\varphi_1 {\bf U}_{[t_1,t_2]} \varphi_2}=& t_2+\max(h^{\varphi_1},h^{\varphi_2}). \\ 
\end{array}
\end{equation}
\begin{define}
An STL formula $\varphi$ is \emph{bounded} if $h^\varphi<\infty$. 
\end{define}
\revtwo{
\begin{define}[\cite{ouaknine2006safety}]
A \emph{safety STL} formula is an STL formula in which all ``until" and ``eventually" intervals are bounded.
\label{define:safety}
\end{define}
}
The satisfaction of $\varphi$ by $\bolds{r}[t]$ is decided only by $\bolds{r}[t:t+h^\varphi]$ and the rest of the signal values are irrelevant. Therefore, instead of $\bolds{r}[t]\models \varphi$, we occasionally write $\bolds{r}[t:t+h^\varphi]\models \varphi$ with the same meaning. 
The \emph{STL robustness score} $\rho(\bolds{r},\varphi,t) \in \mathbb{R}$ is a measure indicating how strongly $\varphi$ is satisfied by $\bolds{r}[t]$, which is recursively computed as \cite{maler_stl}:
\begin{equation}
\label{equ:quant}
\begin{array}{lll}
\rho(\bolds{r},\pi,t) & = & c-p({r_t}) , 
\\
\rho(\bolds{r},\varphi_1 \vee \varphi_2,t)   & = &\max(\rho(\bolds{r},\varphi_1,t),\rho(\bolds{r},\varphi_2,t) ),
\\
\rho(\bolds{r},\varphi_1 \wedge \varphi_2,t)   & = &\min(\rho(\bolds{r},\varphi_1,t),\rho(\bolds{r},\varphi_2,t) ),
\\
\rho(\bolds{r},\varphi_1~{\bf U}_{I}~ \varphi_2,t) & = & \underset{t^\prime \in t+I} \max \big (  
 \min (\rho(\bolds{r},\varphi_2,t'), \\ &&~~~~~~~ \underset{t'' \in [t,t']} \min \rho(\bolds{r},\varphi_1,t''))\big),
\\
\rho(\bolds{r},{\bf F}_{I}~\varphi,t) & = & \underset{t' \in t+I}\max ~ \rho(\bolds{r},\varphi,t'), \\
\rho(\bolds{r},{\bf G}_{I}~\varphi,t) & = & \underset{t' \in t+I}\min~  \rho(\bolds{r},\varphi,t'). \\
\end{array}
\end{equation} 
Positive (respectively, negative) robustness indicates satisfaction (respectively, violation) of the formula.

\begin{example}
\label{example:semantics}
Consider signal $\bolds{r} \in \mathbb{R}^\omega$, where ${r}_k=k, k \in \mathbb{N}$, and $\pi=({r}^2 \le 10)$. We have $\rho(\bolds{r}, \bolds{G}_{[0,3]} \pi,0)=\min(10-0,10-1,10-4,10-9)=1$ (satisfaction) and $\rho(\bolds{r}, \bolds{F}_{[4,6]} \pi,0)=\max(10-16,10-25,10-36)=-6$ (violation). 
\endproof
\end{example}

\begin{remark}
There are minor differences between the original STL introduced in \cite{maler_stl} and the one used in this paper. In \cite{maler_stl}, STL was  developed as an extension of metric interval temporal logic (MITL) \cite{ouaknine2006safety} for real-valued continuous-time signals. Here, without any loss of generality, we apply STL to discrete-time signals. Our STL is based on metric temporal logic (MTL) (similar to \cite{dokhanchi}). Thus, we allow the intervals of temporal operators to be singletons (punctual) or unbounded. It is worth to note that any STL formula in this paper can be translated into an LTL formula by appropriately replacing the time intervals of temporal operators with LTL ``next" operator. However, the LTL representation of STL formulas can be very inefficient. We prefer STL for convenience of specifying requirements for systems with real-valued states. We also exploit the STL quantitative semantics. 
 
\end{remark}

\section{Problem Statement and Approach}
\label{sec:problem}
We consider discrete-time systems of the following form:
\begin{equation}
x_{t+1}=f(x_t,u_t,w_t),
\label{eq:system}
\end{equation}
where $x_t \in \mathcal{X}$ is the state, $\mathcal{X} \subset \mathbb{R}_+^n$, $u_t \in \mathcal{U}$ is the control input, $\mathcal{U}= \mathbb{R}^{m_r} \times \{0,1\}^{m_b}$, and $w_t \in \mathcal{W}$ is the disturbance (adversarial input) at time $t$, $t\in \mathbb{N}$, $\mathcal{W}= \mathbb{R}^{q_r} \times \{0,1\}^{q_b}$. The sets $\mathcal{U}$ and $\mathcal{W}$ may include real and binary values. For instance, the set of controls in the traffic model developed in Sec. \ref{sec:traffic} includes binary values for decisions on traffic lights and real values for ramp meters. These types of systems are positive as all state components are non-negative. We also assume that $\mathcal{X}$ is bounded.
\begin{define}
\label{define:cooperative}
System \eqref{eq:system} is monotone (with partial order on $\mathbb{R}^n_+$) if for all $x,x'\in \mathcal{X}$,  $x' \preceq x$, we have $f(x',u,w) \preceq f(x,u,w), \forall u \in \mathcal{U}, \forall w \in \mathcal{W}$. 
\end{define}
The systems considered in this paper are positive and monotone with partial order on $\mathbb{R}^n_+$. For the remainder of the paper, we simply refer to systems in Definition \ref{define:cooperative} as monotone \footnote{The term \emph{cooperative} in dynamics systems theory is used specifically to refer to systems that are monotone with partial order defined on the positive orthant. We avoid using this term here as it might generate confusion with the similar terminology used for multi-agent control systems.}. 
 Although the results of this paper are valid for any general $f:\mathcal{X} \times \mathcal{U} \times \mathcal{W} \rightarrow \mathcal{X}$, we focus on systems that can be written in the form of mixed-logical dynamical (MLD) systems \cite{Bemporad1999}, which are defined in Sec. \ref{sec:finite}. It is well known that a wide range of systems involving discontinuities (hybrid systems), such as piecewise affine systems, can be transformed into MLDs \cite{heemels2001equivalence}.

\begin{assumption}
\label{assume:w}
There exist $w^* \in \mathcal{W}$ such that
\begin{equation}
\forall x\in \mathcal{X}, \forall u\in \mathcal{U}, ~f(x,u,w) \preceq f(x,u,w^*), \forall w\in \mathcal{W}.
\end{equation}
\end{assumption}
We denote $f(x,u,w^*)$ by $f^*(x,u)$ and refer to $f^*$ as the \emph{maximal system}.  
As it will be further explained in this paper, the behavior of monotone  system \eqref{eq:system} is mainly characterized by its maximal $f^*$. Assumption \ref{assume:w} is restrictive but holds for many compartmental systems where the disturbances are additive and the components are independent. Therefore, the maximal system corresponds to the situation that every component takes its most extreme value. 
We also note that if Assumption \ref{assume:w} is removed, overestimating $f$ by some $f^*$ such that $f(x,u,w) \preceq f^*(x,u), \forall w \in \mathcal{W}$, is always possible for a bounded $f$. By overestimating $f$ the control synthesis methods of this paper remain correct, but become conservative. 


We describe the desired system behavior using specifications written as STL formulas over a finite set of predicates. We assume that each predicate {$\pi$} is in the following form:
\begin{equation}
\revone{
\pi= \left ( a_\pi^T x \le b_\pi  \right ),}
\label{eq:predicate}
\end{equation} 
where $a_\pi \in \mathbb{R}^n_+$, $b_\pi \in \mathbb{R}_+$. It is straightforward to verify that the closed half-space defined by \eqref{eq:predicate} is a lower-set in $\mathbb{R}^n_+$.  
By restricting the predicates into the form \eqref{eq:predicate}, we ensure that a predicate remains true if the values of state components are decreased \revtwo{(Note that this is true for any lower set. We require linearity in order to decrease the computational complexity.)}. This restriction is motivated by monotonicity. For example, in a traffic network, the state is the vector representation of vehicular densities in different segments of the network. The satisfaction of a ``sensible" traffic specification has to be preserved if the vehicular densities are not increased all over the network. Otherwise, the specification encourages large densities and congestion.




\begin{define}
A control policy $\mu:=\bigcup_{t \in \mathbb{N}} \mu_t$ is a set of functions $\mu_t:\mathcal{X}^{t+1} \rightarrow \mathcal{U}$, where $$u_t=\mu_t(x_0,x_1,\cdots,x_t).$$
\end{define} 
An \emph{open-loop} control policy takes the simpler form $u_t=\mu_t(x_0)$, i.e., the decision on the sequence of control inputs is made using only the initial state $x_0$. On the other hand, in a (history dependent) \emph{feedback} control policy, $u_t=\mu_t(x_0,x_1,\cdots,x_t)$, the controller implementation requires real-time access to the state and its history.

An infinite sequence of admissible disturbances is $\bolds{w}=w_0w_1\cdots$, where $w_k \in \mathcal{W}$, $k \in \mathbb{N}$. Following the notation introduced in Sec. \ref{sec:notation}, the set of all infinite-length sequences of admissible disturbances is denoted by $\mathcal{W}^\omega$. Given an initial condition $x_0$, a control policy $\mu$ and $\bolds{w} \in \mathcal{W}^\omega$, the \emph{run} of the system is defined as the following signal: 
$$\bolds{x}=\bolds{x}(x_0,\mu,\userfinal{\bolds{w}}):=x_0 x_1 x_2 \cdots,$$
where $x_{t+1}=f(x_t,u_t,w_t), \forall t \in \mathbb{N}$. 
Now we formulate the problems studied in this paper. In all problems, we assume a monotone system \eqref{eq:system} is given, Assumption \ref{assume:w} holds, and all the predicates are in the form of \eqref{eq:predicate}.

\begin{problem}[Bounded STL Control]
\label{prob:bounded}
Given a bounded STL formula $\varphi$, find a set of initial conditions $\mathcal{X}_0 \subset \mathcal{X}$ and a control policy $\mu$ such that 
\begin{equation*}
\bolds{x}(x_0,\mu,\bolds{w})[0] \models \varphi, \forall \bolds{w} \in \mathcal{W}^\omega, \forall x_0 \in \mathcal{X}_0.
\end{equation*}
\end{problem}
\vspace{5pt}

As mentioned in the previous section, the satisfaction of $\varphi$ solely depends on $\bolds{x}[0:h^\varphi]$, where $h^\varphi$ is obtained from \eqref{eq:horizon}. The horizon $h^\varphi$ can be viewed as the time when the specification ends. In many engineering applications, the system is required to uphold certain behaviors for all times. Therefore, guaranteeing infinite-time safety properties is important. \revtwo{
We formulate \emph{bounded-global} STL formulas in the form of 
\begin{equation}
\label{eq:global}
\varphi_b \wedge {\bf G}_{[\Delta,\infty]} \varphi_g,
\end{equation}
where $\varphi_b, \varphi_g$ are bounded STL formulas, ${\bf G}_{[\Delta,\infty]}$ stands for unbounded temporal ``always"- as defined in Sec. \ref{sec:NF-STL}, and $\Delta \ge h^{\varphi_b}$ is a positive integer. Formula \eqref{eq:global} states that first, $\varphi_b$ is satisfied by the signal from time 0 to $\Delta$, and, afterwards, $\varphi_g$ holds for all times.

\begin{problem}[Bounded-global STL Control]
\label{prob:global}
Given bounded STL formulas $\varphi_b,\varphi_g$, $\Delta \in [h^{\varphi_b},\infty)$, find a set of initial conditions $\mathcal{X}_0 \subset \mathcal{X}$ and a control policy $\mu$ such that
\begin{equation}
\label{eq_bounded_global}
\bolds{x}(x_0,\mu,\bolds{w})[0] \models \varphi_b \wedge \bolds{G}_{[\Delta,\infty]}\varphi_g, \forall \bolds{w} \in \mathcal{W}^\omega, \forall x_0 \in \mathcal{X}_0.
\end{equation}
\end{problem}
\vspace{5pt}

As a special case, we allow $\varphi_b$ to be logical truth so Problem \ref{prob:global} reduces to \emph{global STL} control problem of satisfying ${\bf G}_{[\Delta,\infty]} \varphi_g$. Note that if $\varphi_g$ is replaced by logical truth, Problem \ref{prob:global} reduces to  Problem \ref{prob:bounded}. We have distinguished Problem \ref{prob:bounded} and Problem \ref{prob:global} as we use different approaches to solve them.

It can be shown that (see Appendix) a large subset of safety STL formulas - as in Definition \ref{define:safety} - can be written as 
$
\bigvee_{i=1}^{n_\phi} \phi_i,
$
where each $\phi_i, i=1,\cdots,{n_\phi}$, is a bounded-global formula. Therefore, the framework for solutions to Problem \ref{prob:global} can also be used for safety STL control as it leads to $n_\phi$ instances of Problem \ref{prob:global}, where a solution to any of the instances is also a solution to the original safety STL control problem. The drawback to this approach is that $n_\phi$ can be very large.}

\begin{remark} 
We avoid separate problem formulations for STL formulas containing unbounded ``eventually" or ``until" operators as their unbounded intervals can be safely under-approximated by bounded intervals. However, bounded under-approximation is not sound for the unbounded ``always" operator. A safety formula can be satisfied (respectively, violated) with infinite-length (respectively, finite-length) signals \cite{ouaknine2006safety}. 
\end{remark}

In the presence of disturbances, feedback controllers obviously outperform open-loop controllers. We show that the existence of open-loop control policies for guaranteeing the STL correctness of monotone systems in Problem 1 (respectively, Problem 2) is sufficient and (respectively, {almost}) necessary. The online knowledge of state is not necessary for STL correctness. But it can be exploited for planning controls optimally. 
While our framework can accommodate optimal control versions of Problem \ref{prob:bounded} and Problem \ref{prob:global}, the focus of this paper is on robust optimal control problem for global STL formulas - of form $\bolds{G}_{[0,\infty)}\varphi$, where $\varphi$ is a bounded formula. These type of problems are of practical interest for optimal traffic management (as discussed in Sec. \ref{sec:traffic}). 

We use a model predictive control (MPC) approach, which is a popular, powerful approach to optimal control of constrained systems. Given a planning horizon of length $H$ \footnote{The MPC horizon $H$ should not be confused with the STL horizon $h^\varphi$. }, a sequence of control actions 
 starting from time $t$ is denoted by $u_t^H:=u_{0|t} u_{1|t} \cdots  u_{H-1|t}.$ Given $u^H_t$ and $x_t$, we denote the predicted $H$-step system response by 
$$x_t^H(x_t,u_t^H,w_t^H):= x_{1|t} x_{2|t} \cdots x_{H|t},$$
 where 
$
x_{k+1|t}=f(x_{k|t},u_{k|t},w_{k|t}), k=0,1,\cdots,H-1,
$
 $x_{0|t}=x_t$ and $w_t^H:=w_{0|t} w_{1|t} \cdots w_{H-1|t}$. At each time, $u^H_t$ is found such that it optimizes a cost function $J\left(x_t^H,u^H_t\right)$, $J:\mathcal{X}^H \times \mathcal{U}^H \rightarrow \mathbb{R}$, subject to system constraints. When $u_t^H$ is computed, only the first control action $u_{0|t}$ is applied to the system and given the next state, the optimization problem is resolved for $u_{t+1}^H$. Thus, the implementation is closed-loop. 

\begin{problem}[Robust STL MPC] 
\label{prob:optimal}
Given a bounded STL formula $\varphi$, an initial condition $x_0$, a planning horizon $H$ and a cost function $J\left(x_t^H,u^H_t\right)$, \revtwo{find a control policy such that  $u_t=\mu(x_0,\cdots,x_{t})=u^{\text{opt}}_{0|t}$, where $u_t^{H,\text{opt}}:=u^\text{opt}_{0|t} \cdots u^\text{opt}_{H-1|t}$, and $u^{H,\text{opt}}$ is the following minimizer:
\begin{equation}
\label{eq:predictive}
\begin{array}{cl}
\displaystyle \argmin_{u_t^H}&  \underset{\userfinal{w^H_t}} \max~  J\left(x_t^H(x_t,u_t^H,w_t^H), u_t^H \right), \\
  \text{s.t.}&  \bolds{x}(x_0,\mu,\bolds{w})[0] \models \bolds{G}_{[0,\infty]}\varphi, \forall \bolds{w} \in \mathcal{W}^\omega,\\
  & x_{k+1}=f(x_k,u_k,w_k), \forall k \in \mathbb{N}. \\
\end{array}
\end{equation}
}
\end{problem}
\vspace{3pt}

The primary challenge of robust STL MPC is guaranteeing the satisfaction of the global STL formula while the controls are planned in a receding horizon manner (see the constraints in \eqref{eq:predictive}). Our approach takes the advantage of the results from Problem \ref{prob:global} to design appropriate terminal sets for the MPC algorithm such that the generated runs are guaranteed to satisfy the global STL specification while the online control decisions are computed (sub)optimally. Due to the temporal logic constraints, our MPC setup differs from the conventional one. The details are explained in Sec. \ref{sec:mpc}.

For computational purposes, we assume that $J$ is a piecewise affine function of the state and controls. Moreover, the cost functions in our applications are non-decreasing with respect to the state in the sense that $x'_{k|t} \preceq x_{k|t}, k=1,2,\cdots,H \Rightarrow J(\userfinal{x_t^{\prime H}},u^H_t) \preceq J(x_t^{H},u^H_t), \forall u^H_t \in  \mathcal{U}^H$. As it will become clear later in the paper, we will exploit this property to simplify the worst-case optimization problem in \eqref{eq:predictive} to an optimization problem for the maximal system. 

As mentioned earlier, a natural objective is maximizing STL robustness score. {It follows from the linearity of the predicates in \eqref{eq:predicate} and $\max$ and $\min$ operators in \eqref{equ:quant} that STL robustness score is a piecewise affine function of finite-length signals.} We can also consider optimizing a weighted combination of STL robustness score and a given cost function. We use this cost formulation for traffic application in Sec. \ref{sec:traffic}.

\section{Finite Horizon Semantics}
\label{sec:finite}
In this section, we explain the solution to Problem \ref{prob:bounded}. First, we exploit monotonicity to characterize the properties of the solutions. Next, we explain how to synthesize controls using a mixed integer linear programming (MILP) solver.

\begin{lemma} 
\label{lemma:finite_run} 
Consider runs $\bolds{x}$ and $\bolds{x}'$ and an STL formula $\varphi$. If for some $t,t'$, we have \revtwo{$ \bolds{x}'[t':t'+h^\varphi] \preceq  \bolds{x}[t:t+h^\varphi]$}, then $\bolds{x}[t] \models \varphi$ implies $\bolds{x}'[{t'}] \models \varphi$. 
\end{lemma}
\begin{IEEEproof}
Since all predicates denote lower-sets in the form of \eqref{eq:predicate}, we have $x'_{t'} \preceq x_{t} \Rightarrow a_\pi^T x'_{t'} \le a_\pi^T x_t$, $\bolds{x}[t] \models \pi \Rightarrow \bolds{x}'[t] \models \pi$. Thus, all predicates that were true by the valuations in $\bolds{x}$ remain true for $\bolds{x}'$. The negation-free semantics in \eqref{equ:semantics} implies that without falsifying any predicate, a formula can not be falsified. Therefore, $\bolds{x}[t] \models \varphi$ implies $\bolds{x}'[{t'}] \models \varphi$
\end{IEEEproof}

The \emph{largest set of admissible initial conditions} is defined as:
$$ \mathcal{X}_0^{\max}:=\left \{ x_0 \in \mathcal{X} {\Big |} \exists \mu ~\text{s.t.}~ \bolds{x}(x_0,\mu,\bolds{w}) \models \varphi, \forall \bolds{w} \in \mathcal{W}^\omega \right \}.$$
The set $ \mathcal{X}_0^{\max}$ is a union of polyhedra. Finding the half-space representation of all polyhedral sets in $\mathcal{X}_0^{\max}$ may not be possible for high dimensions. Therefore, we find a half-space representation for a subset of $\mathcal{X}_0^{\max}$. The following result states how to check whether $x_0 \in \mathcal{X}_0^{\max}$.

\begin{theorem} 
\label{theorem:iff}
We have $x_0 \in \mathcal{X}_0^{\max}$ if and only if there exists an open-loop control sequence $$u_{0}^{ol,x_0} u_1^{ol,x_0} \cdots u_{h^\varphi-1}^{ol,x_0}$$ 
such that $\bolds{x}^{ol,x_0}[{0:h^\varphi}] \models \varphi$, where $\bolds{x}^{ol,x_0}[{0:h^\varphi}]=x_0^{ol,x_0}x_1^{ol}\cdots x_{h^\varphi}^{ol,x_0}$, and $x_{k+1}^{ol,x_0}=f^*(x_k^{ol,x_0},u_k^{ol,x_0}), k=0,\cdots,h^\varphi-1, x_0^{ol,x_0}=x_0$.
\end{theorem}
\begin{IEEEproof}
(\emph{Necessity})
Satisfaction of $\varphi$ with $\bolds{w} \in \mathcal{W}^\omega$ requires at least one satisfying run for the maximal system, hence a corresponding control sequence exists. Denote it by $u_0^{ol,x_0} u_1^{ol,x_0} \cdots,u_{h^\varphi-1}^{ol,x_0}$.  
(\emph{Sufficiency})
Consider \user{any} run generated by the original system $x_{k+1}=f(x_k,u_k^{ol,x_0},w_k)$. We prove that $x_k \preceq x_k^{ol,x_0}~, k=0,1,\cdots,h^\varphi$, by induction over $k$. The base case $x_0 \preceq x_0^{ol,x_0}$ is trivial $(x_0=x_0^{ol,x_0})$. The inductive step is verified from monotonicity: $x_{k+1}=f(x_k,u_k^{ol,x_0},w_k) \preceq f^*(x_0^k,u_k^k) = x^{ol,x_0}_{k+1}$. Therefore, $\bolds{x}[0:h^\varphi] \preceq \bolds{x}^{ol,x_0}[0:h^\varphi]$, $\forall \bolds{w}[0:h^\varphi\user{-1}] \in \mathcal{W}^{h^\varphi}$. It follows from Lemma \ref{lemma:finite_run} that $\bolds{x}[0:h^\varphi] \models \varphi, \forall \bolds{w}[0:h^\varphi\user{-1}] \in \mathcal{W}^{h^\varphi}$. 
\end{IEEEproof}
\begin{corollary}
\label{corollary:lower}
The set $\mathcal{X}_0^{\max}$ is a lower-set.
\end{corollary}
\begin{IEEEproof}
Consider any $x_0' \in L(x_0), x_0 \in \mathcal{X}_0^{\max}$. Let $x'_{k+1}= f(x'_k,u_k^{ol,x_0},w_k), k=0,1,\cdots,h^\varphi-1$. It follows from monotonicity that $x_k' \preceq x_k^{ol,x_0}, k=0,1,\cdots,h^\varphi$, $\forall \bolds{w}[0:h^\varphi\user{-1}] \in \mathcal{W}^{h^\varphi}$. By the virtue of Lemma \ref{lemma:finite_run}, $\bolds{x}'[0:h^\varphi] \preceq \bolds{x}^{x_0,ol}[0:h^\varphi]$. Therefore, we have $\forall x_0 \in \mathcal{X}_0^{\max}, x'_0 \in L(x_0) \Rightarrow x'_0 \in  \mathcal{X}_0^{\max}$, which indicates $\mathcal{X}_0^{\max} $ is a lower-set.  
\end{IEEEproof}

\begin{corollary}
If $x_0 \in \mathcal{X}_0^{\max}$ and $\mu^{ol}$ is the following open-loop control policy
$$\mu^{ol}_t(x_0)=u^{ol,x_0}_t, t=0,1,\cdots,h^\varphi-1,$$ 
then \revone{$\bolds{x}(x'_0,\mu,\bolds{w})[0:h^\varphi] \models \varphi, \forall \bolds{w} \in \mathcal{W}^{h^\varphi}, \forall x'_0 \in L(x_0)$}.
\end{corollary}
\begin{IEEEproof}
Follows from the proof of Corollary \ref{corollary:lower}.
\end{IEEEproof}
Now that we have established the properties of the solutions to Problem \ref{prob:bounded}, we explain how to compute the admissible initial conditions and their corresponding open-loop control sequences. The approach is based on formulating the conditions in Theorem \ref{theorem:iff} as a set of constraints that can be incorporated into a feasibility solver. We convert all the constraints into a set of mixed-integer linear constraints and use off-the-shelf MILP solvers to check for feasibility. Converting logical properties into mixed-integer constraints is a common procedure which was employed for MLD systems in \cite{Bemporad1999}. The authors in \cite{karaman} and \cite{raman2014model} extended this technique to a framework for time bounded model checking of temporal logic formulas. A  variation of this method is explained here. 

First, the STL formula is recursively translated into a set of mixed-integer constraints. For each predicate $\pi=(a_\pi^T x \le b_\pi)$, as in \eqref{eq:predicate}, we define a binary variable $z^\pi_k \in \{0,1\}$ such that 1 (respectively, 0) stands for true (respectively, false). The relation between  $z^\pi$, robustness $\rho$, and $x$ is encoded as:
\begin{subequations}
\label{eq:z_predicate}
\begin{equation}
a_\pi^Tx - M (1-z^\pi) + \rho \le b_\pi,
\end{equation}
\begin{equation}
a_\pi^Tx + M z^\pi + \rho \ge b_\pi.
\end{equation}
\end{subequations}
The constant $M$ is a sufficiently large number such that $ M\ge \max \{a_\pi^T K, b_\pi\}$, where $K \in \mathbb{R}^n_+ $ is the upper bound for the state values, $x_k \preceq K, k =0,1,\cdots,h^\varphi$. 
In practice, $M$ is chosen sufficiently large such that the constraint $x \preceq K$ is never active. Note that the largest value of $\rho$ for which $z^\pi=1$ is $b_\pi-a^T_\pi x$, which is equal to the robustness of $\pi$. 

Now we encode the truth table relations. For instance, we desire to capture  $1 \wedge 0=0$ and $1 \vee 0=1$ using mixed-integer linear equations. Disjunction and conjunction connectives are encoded as the following  constraints:
\begin{subequations}
\label{eq:connectives}
\begin{equation}
\label{eq:conjunction}
z= \bigwedge_{i=1}^{n_z} z_i ~ \Rightarrow ~  z \le z_i, i=1,\cdots,n_z,
\end{equation}
\begin{equation}
\label{eq:disjunction}
z= \bigvee_{i=1}^{n_z} z_i ~ \Rightarrow ~  z \le  \sum_{i=1}^{n_z} z_i,
\end{equation}
\end{subequations}
where $z \in [0,1]$ is declared as a continuous variables. However, it only can take binary values as \user{evident from \eqref{eq:connectives}}. 
Similarly, define $z^\varphi_k \in [0,1]$ as the variable indicating whether $\bolds{x}[k] \models \varphi$. An STL formula is recursively translated as:

\begin{equation}
\label{eq:encoding}
\begin{array}{rl} 
\varphi= \bigwedge_{i=1}^{n_\varphi} \varphi_i  \Rightarrow &  z^\varphi_k= \bigwedge_{i=1}^{n_\varphi} z^{\varphi_i}_k, \\
\varphi= \bigvee_{i=1}^{n_\varphi} \varphi_i  \Rightarrow &  z^\varphi_k= \bigvee_{i=1}^{n_\varphi} z^{\varphi_i}_k, \\
\varphi = \bolds{G}_I \psi  \Rightarrow &  z^\varphi_k = \bigwedge_{k^\prime \in I} z^\psi_{k^\prime},\\
\varphi = \bolds{F}_I \psi  \Rightarrow  & z^\varphi_k = \bigvee_{k^\prime \in I} z^\psi_{k^\prime}, \\
\varphi = \psi_1 {\bf U}_I \psi_2  \Rightarrow &  z^\varphi_k =  \bigvee_{k^\prime \in I} \left ( z^{\psi_2}_{k'} \wedge \bigwedge_{k'' \in [k,k^\prime]} z^{\psi_1}_{k''} \right).
\end{array}
\end{equation}
Finally, we add the following constraints:
\begin{equation}
\label{eq:z_constraint}
z^\varphi_0=1,~ \rho \ge 0.
\end{equation} 
\begin{proposition}
\label{prop:feasibility_constraints}
The set of constraints in \eqref{eq:z_predicate},\eqref{eq:connectives},\eqref{eq:encoding},\eqref{eq:z_constraint} has the following properties: 
\begin{itemize}
\item[i)] we have $\bolds{x}[0] \models \varphi$ if the set of constraints is feasible;
\item[ii)] we have $\bolds{x}[0] \not \models \varphi$ if the set of constraints is infeasible;
\item[iii)] the largest $\rho$ such that the set of constraints, while ``$\rho \ge 0$" is removed from \eqref{eq:z_constraint}, is feasible is equal to $\rho(\bolds{x},\varphi,0)$. 
\end{itemize} 
\end{proposition}
\begin{proof}
i) We provide the proof for \eqref{eq:connectives}, as the case for more complex STL formulas are followed in a recursive manner from \eqref{eq:encoding}. If $z=1$, we have from \eqref{eq:conjunction} that $z_i=1, i=1,\cdots,n_z$, which correctly encodes conjunctions.  \revone{Similarly, $z=1$ in \eqref{eq:disjunction} indicates that not all $z_i,i=0,\cdots,n_z$ can be zero, or, $\exists i \in \{1,\cdots,n_z \}$ such that $z_i=1$, which correctly encodes disjunctions. ii) Infeasibility can be recursively traced back into \eqref{eq:connectives}. For both \eqref{eq:conjunction} and \eqref{eq:disjunction}, if $z=1$ is infeasible, it indicates that $z_i=0, i=1,\cdots,n_z$. } iii) We also prove this statement for \eqref{eq:connectives} as it is the base of recursion for general STL formulas. Let $z_i=(a^T_{\pi_i} x + \rho \le b_{\pi_i}), i=1,\cdots,n_z$. 
Consider \eqref{eq:conjunction} and the following optimization problem:
\begin{equation*}
\begin{array}{rl}
\rho^{\max}= \text{argmax} & \rho, \\
\text{s.t.} & a^T_{\pi_i} x + \rho \le b_{\pi_i}, i=1,\cdots,n_z, 
\end{array}
\end{equation*}
where its solution is $\underset{i=1,\cdots, n_z}\min(b_{\pi_i}-a^T_{\pi_i} x)$, which is identical to the quantitative semantics for conjunction (see \eqref{equ:quant}). Similarly, consider \eqref{eq:disjunction} and the following optimization problem:   
\begin{equation*}
\begin{array}{rl}
\rho^{\max}= \text{argmax} & \rho, \\
\text{s.t.} & \exists i \in \{1,\cdots,n_z\}, a^T_{\pi_i} x + \rho \le b_{\pi_i}, 
\end{array}
\end{equation*} 
where the solution is $\underset{i=1,\cdots, n_z}\max(b_{\pi_i}-a^T_{\pi_i} x)$, which is identical to the quantitative semantics for disjunction.    
\end{proof}

Our integer formulation for Boolean connectives slightly differs from the formulation in \cite{karaman}, \cite{raman2014model}, where lower bound constraints for the $z$'s are required. For example, for translating $z= \bigwedge_{i=1}^{n_z}z_i$, it is required to add $z\ge \sum_{i=1}^{n_z} z_i - n_z + 1$ to impose a lower bound for $z$. However, these additional constraints become necessary only when the negation operator is present in the STL formula. Hence, they are removed in our formulation. This reduces the constraint redundancy and degeneracy of the problem. By doing so, we observed computation speed gains (up to reducing the computation time by 50\%) in our case studies. Moreover, we encode quantitative semantics in a different way than \cite{raman2014model}, where a separate STL robustness-based encoding is developed which introduces additional integers. 
\user{Due to property ``iii" in Proposition \ref{prop:feasibility_constraints}, our encoding does not require additional integers to capture robustness hence it is computationally more efficient.}

\begin{define}
System \eqref{eq:system} is in MLD form \cite{Bemporad1999} if written as:
\begin{subequations}
\label{eq:mld}
\begin{equation}
x_{t+1} = A x_t + B_u u_t + B_w w_t + D_\delta \delta_t + D_r r_t,
\end{equation}
\begin{equation}
E_\delta \delta_t + E_r r_t \preceq E_x x_t  +E_u u_t + E_w w_t + e,
\label{eq_nonlinear}
\end{equation}
\end{subequations}
where $\delta_t \in \{0,1\}^{n_\delta}$ and $r_t \in \mathbb{R}^{n_r}$ are auxiliary variables and $A,B_u,B_w,D_\delta, D_r, E_\delta, E_r, E_x, E_u, E_w, e$ are appropriately defined constant matrices such that \eqref{eq:mld} is well-posed in the sense that given $x_t,u_t,w_t$, the feasible set for $x_{t+1}$ is a single point  equal to $f(x_t,u_t,w_t)$. \revthree{Introducing auxiliary variables and enforcing \eqref{eq_nonlinear} can capture nonlinear $f$ \cite{Bemporad1999}}.  
\end{define}
The system equations are brought into mixed-integer linear constraints by transforming system \eqref{eq:system} into its MLD form. 
As mentioned earlier, any piecewise affine system can be transformed into an MLD. In the case studies of this paper, the construction of \eqref{eq:mld} from a piecewise affine \eqref{eq:system} is not explained as the procedure is well documented in \cite{heemels2001equivalence}.  

Finally, the set of constraints in Theorem \ref{theorem:iff} can be cast as:
\begin{equation}
\label{equ:bounded_cosntarints}
\left \{ \begin{array}{ll}
x_0^{ol,x_0}=x_0, & \text{Initial condition;} \\
x_{k+1}^{ol,x_0}=f^*(x_{k}^{ol,x_0},u_k^{ol,x_0}), & \text{System constraints;}\\
z^{\pi}_k = (a_\pi^T x_k^{ol,x_0} \le b_\pi), & \text{Predicates; } \\
z^\varphi_0=1, \rho \ge 0, & \text{STL satisfaction.} \\
\end{array}
\right.
\end{equation}

Checking the satisfaction of the set of  constraints in \eqref{equ:bounded_cosntarints} can be formulated as a MILP feasibility problem, which is handled using powerful off-the-shelf solvers. For a fixed initial condition $x_0$, the feasibility of the MILP indicates whether $x_0 \in \mathcal{X}_o^{\max}$.  An explicit representation of $\mathcal{X}_o^{\max}$ requires variable elimination from \eqref{equ:bounded_cosntarints}, which is computationally intractable for a large MILP. 
Alternatively, we can set $x_0$ as a free variable while maximizing a cost function (e.g. norm of $x_0$) such that a large $L(x_0)$ is obtained. Another natural candidate is maximizing $\rho(\bolds{x}^{ol,x_0},\varphi,0)$. It is worth to note that by finding a set of distinct initial conditions and taking the union of all $L(x_0)$, we are able to find a representation for an under-approximation of $\mathcal{X}_o^{\max}$. 

MILPs are NP-complete. The complexity of solving \eqref{equ:bounded_cosntarints} grows exponentially with respect to the number of binary variables and polynomially with respect to the number of continuos variables. The number of binary variables in our framework is $\mathcal{O} \left ( h^\varphi({n_\pi}+m_b+q_b+n_\delta) \right)$ - \revone{$n_\pi$ is the number of predicates} - and the number of continuous variables is $\mathcal{O} \left( h^\varphi(n+m_r+q_r+n_r) \right)$.  
In other words, the exponential growth builds upon the intricacy of the specification and the number of modes demonstrated by the hybrid nature of the system. However, the complexity is polynomial with respect to the dimension of the state.      

\begin{example}
\label{example:NF-STL}
Consider the following switched system:
\begin{equation*}
x^+=e^{A_u\tau}x + A_u^{-1}(I-e^{-A_u\tau})w,
\end{equation*}
where $x=(x_{[1]},x_{[2]})^T \in \mathbb{R}_+^2$, $u\in \mathcal{U}$ is the control input (switch), $\mathcal{U}=\{1,2\}$, and
\begin{equation*}
A_1=\left ( \begin{array}{cc} 1 & 1 \\ 1 & -5 \end{array} \right )
, A_2= \left ( \begin{array}{cc} -8 & 1 \\ 1 & 2 \end{array} \right ).
\end{equation*}
The (additive) disturbance $w$ is bounded to $L(w^*)$, where $w^*=(1.5,1)^T$ and $\tau=0.1$. This system is the discrete-time version of $\dot{x}=A_u x +w$ with sample time $\tau$. Both matrices are {Metzler} (all off-diagonal terms are non-negative hence all the elements of its exponential are positive)  and non-Hurwitz hence constant input results in unbounded trajectories.
The system is desired to satisfy the  following STL formula:
\begin{equation*}
\varphi= \bigvee_{T=0}^{10} \left(\bolds{F}_{[0,T]} p_1 \wedge \bolds{F}_{\{T\}} p_2\right),
\end{equation*}
where $p_1= \left((x_{[1]} \le 1) \wedge (x_{[2]} \le 5)\right)$ and $p_2= \left((x_{[1]} \le 5) \wedge (x_{[2]} \le 1)\right)$. In plain English, $\varphi$ states that ``within 10 time units, the trajectory visits the box characterized by $p_1$ first and then the box corresponding to $p_2$" (see Fig. \ref{fig:switch_NF_STL}).
We transformed this system into its MLD form \eqref{eq:mld}. We formulated the constraints in \eqref{equ:bounded_cosntarints}  as a MILP and set the cost function to maximize $\left\| x_0 \right\|_\infty$ and used the Gurobi \footnote{\texttt{www.gurobi.com}} MILP solver. The solution was obtained in less than 0.05 seconds on a 3GHz Dual Core MacBook Pro. 
We obtained $x_0=(2.82~2.82)^T$ and the following open-loop control sequence: $1~2~1~2~2~1~1~1~1~1$.   By applying this control sequence, we sampled a trajectory of the original system  $f$ with values of $w$ drawn from a uniform distribution over $L(w^*)$. Both the trajectories of $f$ and $f^*$ satisfy the specification. The results are shown in Fig. \ref{fig:switch_NF_STL}. 
\begin{figure}[t]
\centering
\includegraphics[width=0.23\textwidth]{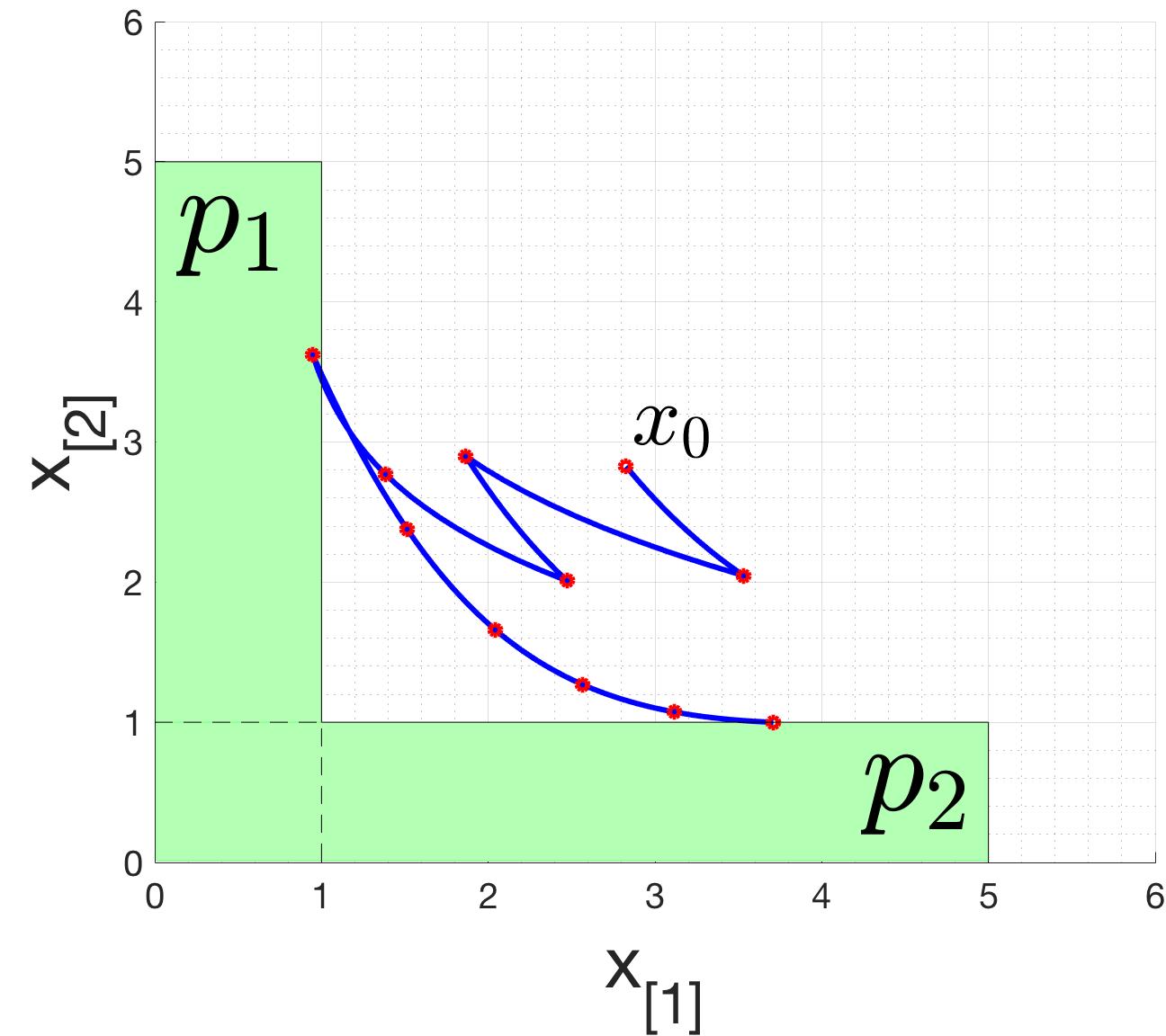}
\includegraphics[width=0.23\textwidth]{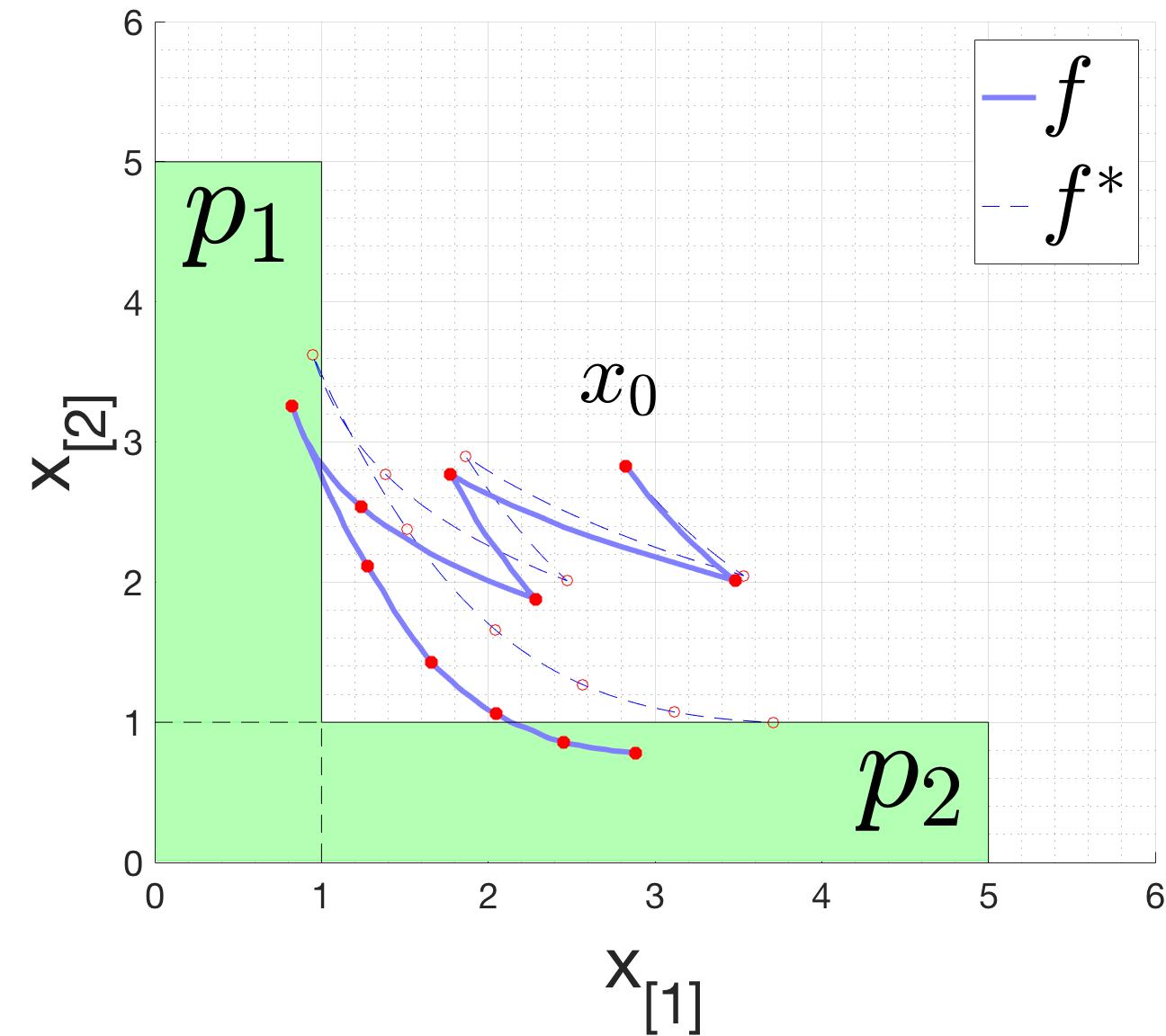}
\caption{Example \ref{example:NF-STL}: [Left] The trajectory of the maximal system $f^*$ which satisfies the specification.  [Right] For the same controls, the trajectory of the original system $f$ with $w$ drawn from an uniform distribution over $L(w^*)$.  }
\label{fig:switch_NF_STL}
\end{figure}
            
\end{example}

\section{Infinite Horizon Semantics}
\label{sec:infinite}

In this section, we provide a solution to Problem \ref{prob:global}. We show that the infinite-time property in \eqref{eq:global} can be guaranteed using repetitive control sequences. First, we consider global specifications and extend the results from our previous work \cite{sadraddini2016safety} in Sec. \ref{sec:s_global}.
\revtwo{Next, we show how to find controls for bounded-global STL formulas (Problem \ref{prob:global}) in Sec. \ref{sec:phi}. Solution completeness is discussed in Sec. \ref{sec:necessary}.}

\user{
\subsection{Global formulas: s-sequences and \userfinal{inductive} invariance}
\label{sec:s_global}

Consider the global specification $\bolds{G}_{[0,\infty]} \varphi$, where $\varphi$ is a bounded formula. We introduce some additional notation. 

\revone{
\begin{define}
Given a bounded STL formula $\varphi$ over predicates in the form \eqref{eq:predicate}, the \emph{language realization set} (LRS) \cite{sadraddini2016feasibility} is:
\begin{equation}
\mathcal{L}^\varphi:=\left \{x_0x_1\cdots x_{h^\varphi} \in \mathcal{X}^{h^\varphi} \big | x_0x_1\cdots x_{h^\varphi} \models \varphi        \right \}.
\end{equation}
\end{define}
}
\begin{proposition}
The set $\mathcal{L}^\varphi$ is a lower-set.
\end{proposition}
\begin{IEEEproof}
For all $x_0x_1\cdots x_{h^\varphi} \in \mathcal{L}^\varphi$ and $x'_0x'_1\cdots x'_{h^\varphi} \preceq x_0x_1\cdots x_{h^\varphi}$, it follows from Lemma \ref{lemma:finite_run} that $x'_0x'_1\cdots x'_{h^\varphi} \models \varphi$. Thus, $x'_0x'_1\cdots x'_{h^\varphi} \in LRS(\varphi)$ hence $LRS(\varphi)$ is a lower set. 
\end{IEEEproof}

It follows from the semantics of global operator in \eqref{equ:semantics} that $\bolds{x} \models \bolds{G}_{[0,\infty]} \varphi$ is equivalent to $\bolds{x}[t:t+h^\varphi] \in \mathcal{L}^\varphi, \forall t \in \mathbb{N}.$ 

\begin{define} 
A set $\userfinal{\Omega_{\mathcal{L}^\varphi}} \subseteq \mathcal{L}^\varphi$ is a \emph{robust control invariant} (RCI) set if: 
\begin{equation}
\begin{array}{l}
\forall~ x_0x_1\cdots x_{h^\varphi} \in \Omega_{\mathcal{L}^\varphi}, \exists u \in \mathcal{U}, \text{ s.t. } \\ x_1 x_2 \cdots x_{h^\varphi}  f(x_{h^\varphi},u,w) ~\in \Omega_{\mathcal{L}^\varphi}, \forall w \in \mathcal{W}.
\end{array}
\end{equation}
\end{define}
Satisfaction of $\bolds{G}_{[0,\infty]} \varphi$ is accomplished by finding a RCI set in $\mathcal{L}^\varphi$. 
Note that unlike traditional definitions of RCI sets (e.g., \cite{Blanchini:1999aa}), where the set is defined in the state-space $\mathcal{X}$, our RCI set is defined in an augmented form of the state-space $\mathcal{X}^{h^\varphi}$. The language realization set can also be interpreted as the ``safe" set in $(h^\varphi+1)$-length trajectory space. } 
The \emph{maximal} RCI set inside $\mathcal{L}^\varphi$ provides a complete solution to the set-invariance problem. The computation of maximal RCI set requires implementing an iterative fixed-point algorithm which is computationally intensive for MLD systems and non-convex sets (see \cite{kerrigan2001,rakovic2004computation} for discussion). We use monotonicity to provide an alternative approach. The following result is a more general version of the one in \cite{sadraddini2016safety}.

\begin{theorem}
\label{theorem_repetitive}
Given a bounded formula $\varphi$, if there exists $\bolds{x}^s[0:h^\varphi] \in \mathcal{L}^\varphi$, and a sequence of controls: $u^{s}_0,\cdots,u^s_{T-1}$ - where $T$ is a positive integer determining the length of the sequence - such that:
\begin{enumerate}
\item $\bolds{x}^s[k:k+h^\varphi] \in \mathcal{L}^\varphi, k=0,1,\cdots,T$, where $x^s_{h^\varphi+k+1}=f^*(x^s_{h^\varphi+k},u^s_k)$,
\item $\bolds{x}^s[T:T+h^\varphi] \preceq \bolds{x}^s[0:h^\varphi]$,  
\end{enumerate}
then the following set is a RCI set in $\mathcal{L}^\varphi$:
\begin{equation}
\label{eq:robust_set_phi_omega}
\Omega_{\mathcal{L}^\varphi}:= \bigcup_{k=0}^{T-1} L(\bolds{x}^{s}[k:k+h^\varphi]).
\end{equation}
\end{theorem}

\begin{IEEEproof}
For any $x'_0x'_1\cdots x'_{h^\varphi} \in \userfinal{\Omega_{\mathcal{L}^\varphi}}$, there exists $i \in \{0,1\cdots,T-1\}$ such that $x'_0x'_1\cdots x'_{h^\varphi} \in L(\bolds{x}^{s}[i:i+h^\varphi])$. 
On one hand, we have $x_{i+1}^{s} \cdots x_{i+h^\varphi}^{s}f^*(x_{i+h^\varphi}^{s},u_{i}^{s}) \in \Omega_{\mathcal{L}^\varphi}$. On the other hand, we have $x'_1 \preceq x_{i+1}^{s},\cdots, x'_{h^\varphi} \preceq x_{i+h^\varphi}^{s}$. By applying $u_{i}^{s}$, monotonicity implies 
\begin{equation*}
\begin{array}{ll}
&
f(x'_{h^\varphi},u_{i}^{s},w) \preceq f^*(x_{i+h^\varphi}^{s},u_{i}^{s}) = x_{i+1+h^\varphi}^{s}, \forall w \in \mathcal{W} \\
 \Rightarrow &  x'_1 x'_2 \cdots x'_{h^\varphi}  f(x'_{h^\varphi},u_{i}^{s},w) \in
\\ 
& L\left(x_{i+1}^{s} x_{i+2}^{s} \cdots x_{i+1+h^\varphi}^{s})\right),\forall w \in \mathcal{W}.
\end{array}
\end{equation*}
And the proof is complete from the fact that $x_{i+1}^{s} \cdots x_{i+1+h^\varphi}^{s} \in \Omega_{\mathcal{L}^\varphi}$ for all $i \in \{0,1\cdots,T-1\}$.  
\end{IEEEproof}
\begin{corollary}
\label{corollary_repetitive}
Let the conditions in Theorem \ref{theorem_repetitive} hold and $\bolds{x}[t_0:t_0+h^\varphi] \in L(\bolds{x}^s[0:h^\varphi])$ for some $t_0 \in \mathbb{N}$. Consider the following control sequence starting from time $t_0+h^\varphi$:
\begin{equation}
\label{eq_s}
\bolds{u}^s:=\left(u^s_0 u^s_1 \cdots u^s_{T-1}\right)^\omega,
\end{equation}
i.e., $u^s_t=u^s_{\rem(t-t_0-h^\varphi,T)}, t\ge t_0+h^\varphi$. Let $x_{k+1}=f^*(x_k,u_k), k=t_0+h^\varphi,\cdots$. Then we have $\bolds{x}[t:t+h^\varphi] \in \Omega_{\mathcal{L}^\varphi}, \forall t\ge t_0$.
\end{corollary}
\begin{proof}
We prove by induction that $\bolds{x}[t:t+h^\varphi] \in L(\bolds{x}^s[\rem(t-t_0,T):\rem(t-t_0\userfinal{,T})+h^\varphi]), \forall t\ge t_0$. The base case for $t=t_0$ is true. In order to prove the inductive step $\bolds{x}[t+1:t+1+h^\varphi] \in L(\bolds{x}^s[\rem(t+1-t_0,T):\rem(t+1-t_0,T)+h^\varphi])$, we need to prove that $x_{t+k+1} \preceq x^s_{\rem(t+1-t_0,T)+k}, k=0,\cdots,h^\varphi$, for which we need to only prove the case for $k=h^\varphi$ as previous inequalities are already assumed by induction.   
We show $x_{t+h^\varphi+1} \preceq x^s_{\rem(t+1-t_0,T)+h^\varphi}$ through monotonicity and the induction assumption that $x_{t+h^\varphi} \preceq x^s_{\rem(t-t_0,T)+h^\varphi}$:
\begin{equation*}
\begin{array}{ll}
x_{t+h^\varphi+1}& =f^*( x_{t+h^\varphi},u^s_{\rem(t-t_0,T)} )
\\
 & \preceq f^*( x^s_{\rem(t-t_0,T)+h^\varphi} ,u^s_{\rem(t-t_0,T)}) 
 \\
 & = x^s_{\rem(t-t_0,T)+1+h^\varphi} \preceq x^s_{\rem(t+1-t_0,T)+h^\varphi}. 
\end{array}
\end{equation*}
Note that $x^s_{T+h^\varphi} \preceq x^s_{h^\varphi}$. The ``$\preceq$" in the last line can be replaced by ``$=$" when $\rem(t-t_0,T)+1 \neq T$. 
\end{proof}

We refer to the repetitive sequence of controls in \eqref{eq_s} as an \emph{s-sequence}. An s-sequence is an invariance inducing open-loop control policy.   
Once the latest $h^\varphi+1$-length of system state are brought into $\Omega_{\mathcal{L}^\varphi}$, an s-sequence keeps the $h^\varphi+1$-length trajectory of the system in $\Omega_{\mathcal{L}^\varphi}$ for all subsequent times.

The computation of an s-sequence requires solving an MILP for $\bolds{x}^{s}[h^\varphi:T] \models \bolds{G}_{[0,T]} \varphi$ (an instance of Problem \ref{prob:bounded}) with an additional set of constraints in  $\bolds{x}^{s}[0:h^{\varphi}] \models \varphi$ (again, an instance of Problem \ref{prob:bounded}, but without the dynamical constraints. In other words, $\bolds{x}^{s}[0:h^{\varphi}]$ does not need to be a trajectory of the maximal system), and $\bolds{x}^{s}[T:T+h^{\varphi}] \preceq \bolds{x}^{s}[0:h^{\varphi}]$ (linear constraints). We are usually interested in the {shortest} s-sequence since its computation requires the smallest MILP. Algorithmically, we start from $T=1$ and implement $T\gets T+1$ until the MILP formulating the conditions in Theorem \ref{theorem_repetitive} becomes feasible and an s-sequence is found. As it will be implied from results in Sec. \ref{sec:necessary}, existence of an s-sequence is almost necessary for existence of a RCI set.

\begin{example}
\label{example:s}
Consider the system in Example \ref{example:NF-STL}. We wish to keep the trajectory in the set characterized by $p_1 \vee p_2$, i.e., $\mathcal{S}=L\left((1,5)^T) \right) \cup L\left((5,1)^T) \right)$. Note that this set is non-convex. We set the cost function to maximize $\left\|x_0\right\|_1$. The shortest s-sequence has $T=5$ and is: $(2~1~2~1~1)^\omega$. The resulting trajectory satisfying the definition of s-sequence is shown in Fig. \ref{fig:s} (a). The corresponding robust control invariant set is shown in Fig. \ref{fig:s}. (b) (cyan region), which is characterized by the $x^s_0,x^s_1,\cdots,x^s_4$ (red dots) that lie inside $\mathcal{S}$ (green region). Note that the $[0,2]\times[0,2]$ portion of the coordinates in Fig. \ref{fig:switch_NF_STL} is shown here for a clearer representation of the details. 
\end{example}

\subsection{Bounded-global specifications: $\phi$-sequences}
\label{sec:phi}
Now we consider general bounded-global formulas - as in Problem \ref{prob:global} - and generalize the paradigm used for s-sequences. We provide the key result of this section.

\begin{figure}[t]
\centering
\includegraphics[width=0.23\textwidth]{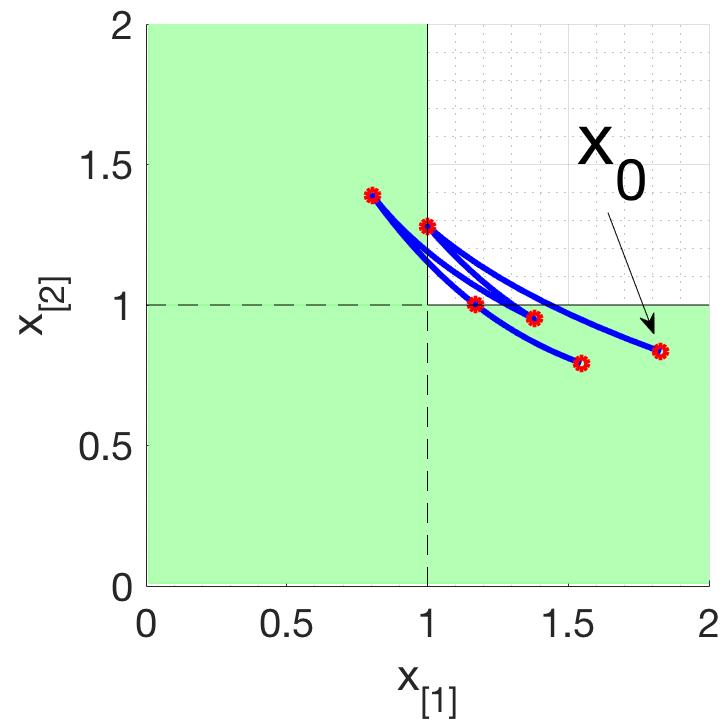}
\includegraphics[width=0.23\textwidth]{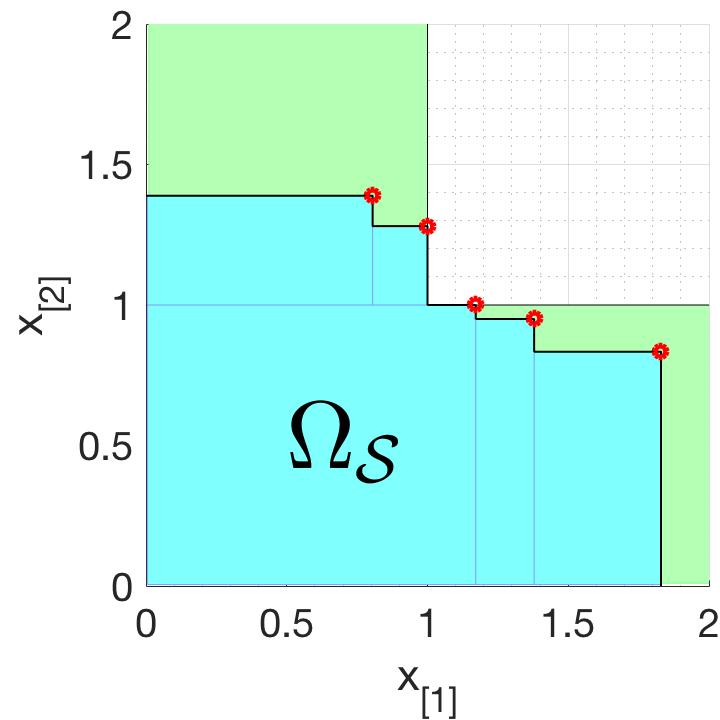}
\caption{Example \ref{example:s}: [Left] The trajectory satisfying the conditions of s-sequences. [Right] The corresponding robust control invariant set inside $\mathcal{S}$. }
\label{fig:s}
\end{figure}

\begin{theorem}
\label{theorem:phi-sequence}
Given a bounded-global STL formula $\phi=\varphi_b \wedge \bolds{G}_{[\Delta,\infty]} \varphi_g$, an initial condition $x_0$, a control sequence $u_0^{\phi} \cdots u_{\Delta+T+h^{\varphi_g}-1}^{\phi}$, where $T$ is a positive integer, and a non-negative integer $T_0 < T$, let the following conditions hold:
\begin{enumerate}
\item $\bolds{x}^{\phi}[0:\Delta+T+h^{\varphi_g}] \models \varphi_b \wedge \bolds{G}_{[\Delta,\Delta+T]} \varphi_g$,
\item $\bolds{x}^{\phi}[\Delta+T:\Delta+T+h^{\varphi_g}] \preceq \bolds{x}^{\phi}[\Delta+T_0:\Delta+T_0+h^{\varphi_g}]$; 
\end{enumerate}
where $x_{k+1}^{\phi}=f^*(x_k^{\phi},u_k^{\phi}), k \in [0,\Delta+T+h^\varphi_g-1]$, $x_0^{\phi}=x_0$. Let $\mu^{ol}$ be the open-loop  control policy  corresponding  to the following control sequence:
\begin{equation}
\label{eq:phi}
\bolds{u}^{\phi}:=u_{0}^{\phi}\cdots u_{\Delta+T_0+h^{\varphi_g}-1}^{\phi}\left(u_{\Delta+T_0+h^{\varphi_g}}^{\phi}\cdots u_{\Delta+T+h^{\varphi_g}-1}^{\phi}\right)^\omega,
\end{equation}
Then $ 
\bolds{x}(x'_0,\mu^{ol},\bolds{w}) \models \phi, \forall \bolds{w} \in \mathcal{W}^\omega, \forall x'_0 \in L(x_0).$ Moreover, the following set is a RCI set in $\mathcal{L}^{\varphi_g}$:
\begin{equation}
\label{eq:robust_set_phi}
\Omega_{\mathcal{L}^{\varphi_g}}:= \bigcup_{i=0}^{T-T_0-1} L(\bolds{x}^{\phi}[\Delta+T_0+i:\Delta+T_0+h^{\varphi_g}+i]).
\end{equation}
\end{theorem}
\begin{IEEEproof}
We need to prove that $\bolds{x}(x^\phi_0,\mu^{ol},\bolds{w}^*) \models \phi$, where $\bolds{w}^*=(w^*)^\omega$. 
The fact that $\bolds{x}(x'_0,\mu^{ol},\bolds{w})[0] \models \phi, \forall x'_0 \in L(x_0), \forall \bolds{w}\in \mathcal{W}^\omega $, follows from monotonicity and Lemma \ref{lemma:finite_run}. 
The fact that $\Omega_{\mathcal{L}^{\varphi_g}}$ is a RCI set \userfinal{follows from Theorem \ref{theorem_repetitive} as \eqref{eq:robust_set_phi_omega} is obtained from replacing $\Delta=T_0=0$ in \eqref{eq:robust_set_phi} }. It follows that $\left(u_{\Delta+T_0+h^{\varphi_g}}^{\phi}\cdots u_{\Delta+T+h^{\varphi_g}-1}^{\phi}\right)^\omega$ is an s-sequence. 
For all $t\ge \Delta + T + h^{\varphi_g}$, let
\begin{equation}
\label{eq_x_phi}
x^\phi_{t+1}=f^*(x_t^\phi,u^\phi_{\Delta+T_0+h^{\varphi_g}+\rem(t-\Delta-T_0-h^{\varphi_g},T-T_0)}).
\end{equation}
Using Corollary \ref{corollary_repetitive}, we have 
$\bolds{x}^\phi[k+\Delta+T_0:k+\Delta+T_0+h^{\varphi_g}] \in \mathcal{L}^{\varphi_g}, \forall k \in \mathbb{N}$, and the proof is complete.

\end{IEEEproof}

We refer to the sequence of controls in \eqref{eq:phi} as a $\phi$-\emph{sequence}. The computation of a $\phi$-sequence requires solving an MILP for $\bolds{x}^{\phi}[0:\Delta+T+h^{\varphi_g}] \models \varphi_b \wedge \bolds{G}_{[\Delta,\Delta+T]} \varphi_g$ (an instance of Problem \ref{prob:bounded}) with an additional set of constraints in  $\bolds{x}^{\phi}[\Delta+T:\Delta+T+h^{\varphi_g}] \preceq \bolds{x}^{\phi}[\Delta+T_0:\Delta+T_0+h^{\varphi_g}]$ (linear constraints). Thus, similar to s-sequecnes, the computation of a $\phi$-sequence is based on feasibility checking of a MILP. We have two parameters $T$ and $T_0 < T$ to search over. We start from $T=1$ and implement $T\gets T+1$, while checking for all $T_0 < T$, until the corresponding MILP gets feasible. In Sec. \ref{sec:necessary}, we discuss the necessity of existence of a feasible solution for some $T_0,T$.

Another interpretation of a $\phi$-sequence is a sequence that consists of \userfinal{an initialization} segment of length $\Delta+h^{\varphi_g}$ to bring the latest $h^{\varphi_g}$ states of the system into $\Omega_{\mathcal{L}^{\varphi_g}} \subseteq \mathcal{L}^{\varphi_g}$ and a repetitive segment of length $T$ to stay in $\Omega_{\mathcal{L}^{\varphi_g}}$. The repetitive segment is an s-sequence. 
Since control inputs eventually becoming periodic, the long-term behavior is expected to demonstrate periodicity, which leads to the following result based on Theorem \ref{theorem:phi-sequence}.

\begin{corollary}
\label{theorem:orbit}
The $\omega$-limit set of the run given by \eqref{eq_x_phi} is non-empty and corresponds to the following periodical orbit:
\begin{equation}
\label{eq:orbit}
\left(x_{0}^{\phi,\infty} x_{1}^{\phi,\infty} \cdots x_{T-T_0-1}^{\phi,\infty}\right)^\omega,
\end{equation}
where
$ x_{k}^{\phi,\infty}:= \lim_{c \rightarrow \infty} x_{k+\Delta+T_0+c(T-T_0)}^{\phi}, k=0,\cdots,T-T_0-1.$

\end{corollary} 
\begin{IEEEproof}
We show that $x^\phi_{t} \preceq x^\phi_{t+T-T_0}, \forall t \ge \Delta+T_0$. Similar to the proof of Corollary \ref{corollary_repetitive}, we use induction. The base case for $t=\Delta+T_0$ is already in the second condition in Theorem \ref{theorem:phi-sequence}. The inductive step is proven as follows:
\begin{equation*}
x^\phi_{t+1+T-T_0} =f^*( x^\phi_{t+T-T_0},u_{T-T_0+t} ) \preceq f^*( x^\phi_{t},u_{t})=  x^\phi_{t+1},
\end{equation*}
where from \eqref{eq:phi} we have 
$$u_{t+T-T_0}=u_t=u^\phi_{t+\Delta+T_0+h^{\varphi_g}+\rem(t-\Delta-T_0-h^{\varphi_g},T-T_0)}.$$ Thus, each component of the sequence $x^\phi_{\Delta+T_0+k}x^\phi_{\Delta+T+k} x^\phi_{\Delta+2T-T_0+k} \cdots$, $k=0,\cdots,T-T_0$, is monotonically decreasing. 
\emph{Monotone convergence theorem} \cite{sutherland1975introduction} explains that a lower-bounded monotonically decreasing sequence converges (in this case, all values are lower-bounded by zero). Thus, $\lim_{c\rightarrow \infty} x_{\Delta+T_0+k+c(T-T_0)}^{\phi}, k=0,\cdots,T-T_0,$ exists and the proof is complete. 
\end{IEEEproof}

\begin{example}
\label{example:SSTL}
Consider the system in Example \ref{example:NF-STL}. We wish to satisfy 
\begin{equation*}
\phi= \bolds{F}_{[0,5]} p_1~ \wedge~ \bolds{G}_{[5,\infty)} \left(
\bolds{F}_{[0,6]} p_1~ \wedge~ \bolds{F}_{[0,6]} p_2 \right).
\end{equation*}
The specification is in form in \eqref{eq_bounded_global} with $h^{\varphi_b}=\Delta=5,h^\varphi_g=6$. 
This specification requires that $p_1$ is visited at least once until $t=5$ and, afterwards, $p_1$ and $p_2$ are persistently visited while the maximum time between two subsequent visits is not greater than $6$. We find a $\phi$-sequence solving a MILP for $T=7, T_0=0,$ while maximizing $\left\|x_0 \right \|_1$. The obtained $\phi$-sequence is 
$\bolds{u}^\phi=2~2~2~2~1~2~1~1~1~1~2~(2~ 2~ 1~ 1 ~1 ~1 ~2)^\omega$ for $x_0=(12.4,0)^T$. The first \userfinal{$\Delta+T+h^{\varphi_g}+1=5+7+6+1=19$} time points of the trajectory of the maximal system $f^*$ satisfying the conditions in Theorem \ref{theorem:phi-sequence} are shown in  Fig. \ref{fig:SSTL} [Left]. A sample trajectory of $f$ with values of $\bolds{w}$ chosen uniformly from $\mathcal{W}$ is also shown. Both trajectories satisfy $\phi$. The limit-set of $f^*$, which is a 7-periodical orbit, is shown in Fig. \ref{fig:SSTL} [Right].  
\begin{figure}[t]
\centering
\includegraphics[height=0.20\textwidth]{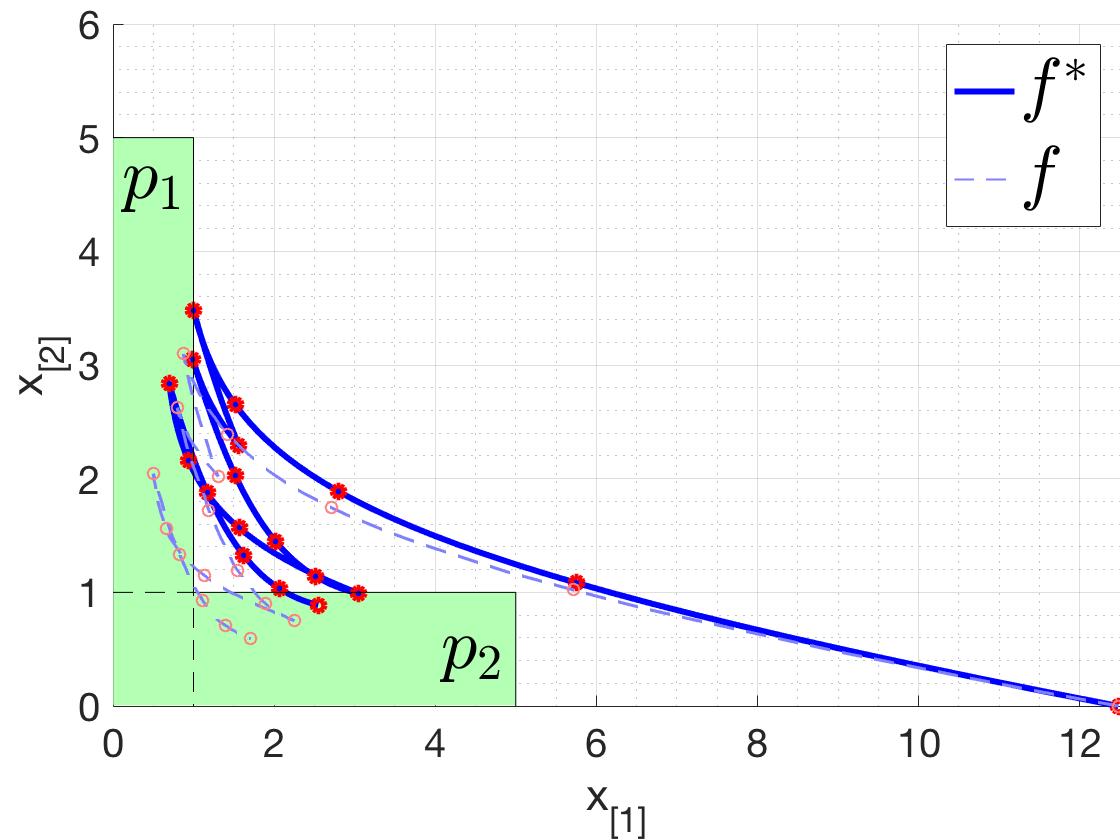}
\includegraphics[height=0.20\textwidth]{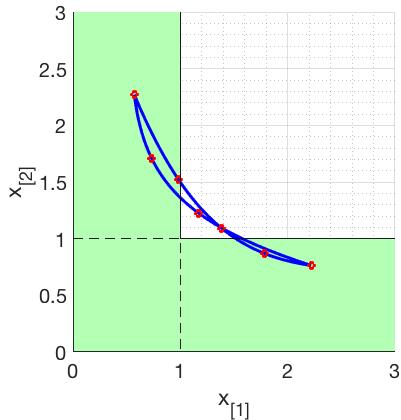}
\label{fig:SSTL}
\caption{Example \ref{example:SSTL}: [Left] The first 19 points of the trajectory of the maximal system $f^*$ that satisfy the conditions in Theorem \ref{theorem:phi-sequence}. A sample trajectory of $f$ is also shown. [Right] The $\omega$-limit set (red dots) of $f^*$ is a 7-periodic orbit.}
\end{figure}
\end{example}

\user{
\subsection{Necessity of Open-loop Strategies}
\label{sec:necessary}
We showed that if there exists an initial condition and a finite length control sequence such that the statements in Theorem \ref{theorem_repetitive} hold, an open-loop control sequence is sufficient for satisfying of a bounded-global formula, as was formulated in Problem \ref{prob:global}. In this section, we address the necessity conditions.}
\revthree{ 
We show that the existence of open-loop control strategies for satisfying a bounded-global specifications is \emph{almost necessary} in the sense that if a $\phi$-sequence is not found using Theorem \ref{theorem:phi-sequence} for large values of $T$, then it is almost certain that no correct control policy (including feedback policies) exists, or, if exists any, it is \emph{fragile} in the sense that a slight increase in the effect of the disturbances makes the policy invalid. We characterize the necessity conditions based on hypothetical perturbations in the disturbance set.}

\begin{theorem}
\label{theorem:necessary}
Suppose system \eqref{eq:system} is strongly monotone with respect to the maximal disturbance in the sense that for all $\epsilon>0$, there exists a perturbed disturbance set $\mathcal{W}_p$ with maximal disturbance $w^*_p$ such that 
\begin{equation}
\label{eq:SMA}
\forall x \in \mathcal{X}, \forall u \in \mathcal{U}, f(x,u,w^*) + {1}_n \epsilon \preceq f(x,u,w_p^*).
\end{equation}
Consider the bounded-global formula $\phi=\varphi_b \wedge \bolds{G}_{[\Delta,\infty]}\varphi$. Given $\epsilon>0$, the disturbance set is altered to $\mathcal{W}_p$ such that \eqref{eq:SMA} holds. If there exists a control policy $\mu$ and an initial condition $x_0$ such that $\bolds{x}(x_0,\mu,\bolds{w}_p) \models \phi, \forall \bolds{w}_p \in \mathcal{W}_p^\omega$, then there exists at least one open-loop control policy $\mu^{ol}$ in the form of a $\phi$-sequence in \eqref{eq:phi} for the original system such that
\begin{equation} 
T \le \nicefrac{A}{\epsilon^{n(h^{\varphi_g}+1)}},
\end{equation}
where $A$ is a constant depending on $\mathcal{L}^{\varphi_g}$. 
\end{theorem}
\begin{IEEEproof}
Given a bounded set $\mathcal{C} \subset \mathbb{R}^{n(h^{\varphi}+1)}$, we define the diameter $d(\mathcal{C}):= \inf \{ d |s_1\preceq s_2+d{1}_{n(h^{\varphi}+1)}, \forall s_1,s_2 \in \mathcal{C} \}$ (e.g., the diameter of an axis-aligned hyper-box is equal to the length of its largest side). 
Consider a partition of $\mathcal{L}^{\varphi}$ by a finite number of cells, where the diameter of each cell is less than $\epsilon$. The maximum number of cells required for such a partition is $\nicefrac{A}{\epsilon^{n(h^{\varphi_g}+1)}}$, where $A$ is a constant dependent on the shape and volume of $\mathcal{L}^{\varphi_g}$. A conservative upper bound on $A$ can be given as follows. \revone{Define $a^*\in \mathbb{R}_+$ as
$$
\argmin_{\nicefrac{a}{\epsilon} \in \mathbb{N}} \Big\{\bolds{x}[0:h^{\varphi_g}] \preceq a \bolds{1}_{n}[0:h^{\varphi_g}], \forall \bolds{x}[0:h^{\varphi_g}] \in \mathcal{L}^{\varphi_g} \Big\}.
$$ 
}
Since $\mathcal{L}^{\varphi_g}$ is bounded and closed, $a^*$ exists. We have $\mathcal{L}^{\varphi_g} \subseteq L(a^* 1_{n(h^{\varphi_g}+1)})$. Let $A$ be ${a^*}^{n(h^{\varphi_g}+1)}$ - the volume of $L(a^* 1_{n(h^{\varphi_g}+1)})$, which is a hyper-box. 
Partition $L(a^* 1_{n(h^{\varphi_g}+1)})$ into $N:=\nicefrac{A}{\epsilon^{n(h^{\varphi_g}+1)}}$ number of equally sized cubic cells with side length of $\epsilon$. Such a partition also partitions $\mathcal{L}^{\varphi_g}$ to at most $N$ number of cells where the diameter of each cell is not greater than $\epsilon$. 

Since there exists $\mu$ such that $\bolds{x}(x_0,\mu,\bolds{w}_p) \models \phi, \forall \bolds{w}_p \in \mathcal{W}_p^\omega$, then there exist at least one run satisfying $\phi$ for system $x_{k+1}=f(x_{k},u_k,w_p^*)$. Let $x_0,\cdots, x_{\Delta+h^{\varphi_g}+N}$ be the first $\Delta+h^{\varphi_g}+N+1$ time points of a such a run. We have $\bolds{x}[k:k+h^{\varphi_g}] \in \mathcal{L}^{\varphi_g}, k=\Delta,\cdots,\Delta+N$. Consider the sequence $\bolds{x}[\Delta:\Delta+h^{\varphi_g}] \bolds{x}[\Delta+1:\Delta+1+h^{\varphi_g}] \cdots \bolds{x}[\Delta+N:\Delta+N+h^{\varphi_g}]$. 
Consider a partition of $\mathcal{L}^{\varphi_g}$ with cells that for all cells the diameter is less than $\epsilon$. By the virtue of \emph{pigeonhole principle}, there exists a cell \userfinal{that} contains at least two time points $\bolds{x}[{k_1}:k_1+h^{\varphi_g}]$ and $\bolds{x}[{k_2}:k_2+h^{\varphi_g}]$, $\Delta \le k_1 \le k_2\le \Delta+N$. From the assumption on the diameter of the cells we have 
\begin{equation}
\label{eq_cell}
\bolds{x}[{k_2}:k_2+h^{\varphi_g}] \preceq \bolds{x}[{k_1}:k_1+h^{\varphi_g}] + \epsilon \bolds{1}_{n}[0:h^{\varphi_g}]. 
\end{equation} 

Now consider system $x'_{k+1}=f(x'_{k},u_k,w^*)$ - the original maximal system - with $x'_{k_1+h^{\varphi_g}}=x_{k_1+h^{\varphi_g}}$. We prove that 
\begin{equation}
\label{eq_prime}
x'_{k} + 1_n \epsilon \le x_k, \forall k> k_1+h^{\varphi_g}. 
\end{equation}
We use induction. The base case for $k=k_1+h^{\varphi_g}+1$ is verified using \eqref{eq:SMA}: 
\begin{equation*}
\begin{array}{rl}
x'_{k_1+1+h^{\varphi_g}} + 1_n \epsilon  = & f(x'_{k_1+h^{\varphi_g}},u_{k_1+h^{\varphi_g}},w^*) +1_n \epsilon \\
\le & f(x'_{k_1+h^{\varphi_g}},u_{k_1+h^{\varphi_g}},w_p^*) \\ 
=  & x_{k_1+1+h^{\varphi_g}}.
\end{array}
\end{equation*}
The inductive step is verified using monotonicity and \eqref{eq:SMA}: 
\begin{equation*}
\begin{array}{rl}
x'_{k+1+h^{\varphi_g}}+ 1_n \epsilon  = & f(x'_{k+h^{\varphi_g}},u_{k+h^{\varphi_g}},w^*) +1_n \epsilon \\
\le & f(x'_{k+h^{\varphi_g}},u_{k+h^{\varphi_g}},w_p^*)  \\
 \le & f(x_{k+h^{\varphi_g}},u_{k+h^{\varphi_g}},w_p^*) = x_{k+1+h^{\varphi_g}}.
\end{array}
\end{equation*}
It immediately follows from \eqref{eq_prime} that 
\begin{equation}
\label{eq_primesignal}
\bolds{x'}[k_2:k_2+h^{\varphi_g}] + \epsilon \bolds{1}_{n}[0:h^{\varphi_g}] \le \bolds{x}[k_2:k_2+h^{\varphi_g}]. 
\end{equation}
Since the lefthand of \eqref{eq_cell} is the righthand of \eqref{eq_primesignal}, we have:
\begin{equation}
\label{eq_primerci}
\bolds{x'}[k_2:k_2+h^{\varphi_g}] \le \bolds{x}[k_1:k_1+h^{\varphi_g}]. 
\end{equation}
This is reminiscent of the conditions in Theorem \ref{theorem_repetitive}. Now by defining $x'_k:=x_k, k=k_1,\cdots,k_1+h^{\varphi_g}-1$, we conclude that  
$$\Omega'_{\mathcal{L}^{\varphi_g}}:= \bigcup_{k=k_1}^{k_2-1} L(\bolds{x}'[k:k+h^\varphi])$$
is a RCI set for system with adversarial disturbance set $\mathcal{W}$ and $(u_{k_1}\cdots u_{k_2+h^\varphi_g-1})^\omega$ is an s-sequence.   
 
Now, once again, consider the original system $x'_{k+1}=f(x'_{k},u_k,w^*)$ with $x'_0=x_0$. Monotonicity implies $\bolds{x'}[0:k_1+h^{\varphi_g}] \le \bolds{x}[0:k_1+h^{\varphi_g}]$. Thus, by applying $u_0,\cdots,u_{k_1+h^{\varphi_g}-1}$ and using Lemma \ref{lemma:finite_run}, we have $\bolds{x'}[0:k_1+h^{\varphi_g}] \models \varphi_b \wedge \bolds{G}_{[\Delta,k_1]} \varphi_g$. Corollary \ref{corollary_repetitive} implies $x'[k_1+h^{\varphi_g}] \models \bolds{G}_{[0,\infty)} \varphi_g$ if $(u_{k_1},\cdots,u_{k_2-1})^\omega$ is applied starting from time $k_1$. Finally, monotonicity and  Lemma \ref{lemma:finite_run} immediately indicate that $\bolds{x}(x''_0,\mu^{ol},\bolds{w}) \models \phi, \forall x''_0 \in L(x_0), \forall \bolds{w} \in \mathcal{W}^\omega$, where $\mu^{ol}$ is the following open-loop control strategy producing the following control sequence:
$$
u_0\cdots u_{k_1-1} (u_{k_1}\cdots u_{k_2+h^\varphi_g-1})^\omega,
$$ 
which is in form of \eqref{eq:phi} with $k_1=\Delta+T_0+h^{\varphi_g}$ and $k_2=\Delta+T$. Since $k_2 \le \Delta+N$, we also have $T \le N, N=\nicefrac{A}{\epsilon^{n(h^{\varphi_g}+1)}}$, and the proof is complete.

\end{IEEEproof}

\begin{corollary}
\label{corollary_necessity}
Suppose that for all $T \le T^{\max}, T_0 < T$, there does not exist an initial condition and a control sequence such that the conditions in  Theorem \ref{theorem:phi-sequence} hold. Then there does not exist any solution to Problem \ref{prob:global} given that the maximal disturbance is $w_p^*$ such that \eqref{eq:SMA} holds with $\epsilon > \sqrt[n(h^{\varphi_g}+1)] {T^{\max}}$.  
\end{corollary}

The relation between the fragility in Theorem \ref{theorem:necessary} and the length of the $\phi$-sequence suggests that by performing the search for longer $\phi$-sequences (which are computationally more difficult), the bound for fragility becomes smaller, implying that a correct control policy (if exists) is close to the limits (i.e., robustness score is close to zero, or the constraints are barely satisfied in the case with maximal disturbance). In practice, the bounds in Theorem \ref{theorem:necessary} are very conservative and one may desire to find tighter bounds for specific applications.

\begin{example}
Consider Example \ref{example:s}. Suppose that there does not exist an s-sequence of length smaller than 144 with maximal disturbance $w^*$. The constant $A$ (area in this 2D case, see proof of Theorem \ref{theorem:necessary}) of region corresponding to $p_1 \vee p_2$ is $9$. Therefore, $\mathcal{S}$ can be partitioned into $144$ equally sized square cells with side length $0.25$. Note that we have $\epsilon^2\ge 9/T$. Since the disturbances are additive, it follows that if \user{$A_u^{-1}(I-e^{-A_u t})(w^*_p-w^*) > (0.25,0.25)^T, u=1,2$, then there does \emph{not} exist any control strategy $\mu$ and $x_0 \in \mathbb{R}_+^n$ such that $\bolds{x}(x_0,\mu,\bolds{w}_p) \models \bolds{G}_{[0,\infty]}(x \in \mathcal{S}),\forall \bolds{w}_p \in \mathcal{W}_p^\omega$.}
\end{example}

\section{Model Predictive Control}
\label{sec:mpc}
In this section, we provide a solution to Problem \ref{prob:optimal}. We assume full knowledge of the history of state. As mentioned in Sec. \ref{sec:problem}, the cost function $J$ is assumed to be non-decreasing with respect to the state values hence the system constraints are replaced with those of the maximal system. First, we explain the MPC setup for global STL formulas. Next, we  prove that the proposed framework is guaranteed to generate runs that satisfy the global STL specification \eqref{eq:global}.

Let $t\ge h^\varphi-1$. The case of $t<h^\varphi-1$ is explained later. 
Given planning horizon $H$, the states that are predictable at time $t$ using controls in $u^H_t$ are $x_{1|t},x_{2|t},\cdots,x_{H|t}$. Given predictions $x_{1|t},x_{2|t},\cdots,x_{H|t}$, we need to enforce $\bolds{x}[{t-h^\varphi+1,t+H}] \models \bolds{G}_{[0,H-1]} \varphi$ at time $t$. Notice that 
\begin{equation}
\user{\bolds{x}_{|t}[{t-h^\varphi+1,t+H}]:=x_{t-h^\varphi+1} \cdots x_t x_{1|t} \cdots x_{H|t},}
\end{equation}
i.e., the first $h^\varphi$ time points are actual values, the rest are predictions. Also, note that the values in $\bolds{x}[\tau:\tau+h^\varphi]$ are independent of the values in $x_t^H$ for $\tau \le t-h^\varphi$ and are not fully available for $\tau>t+H-h^\varphi$. Thus, $[t-h^\varphi+1,t+H-h^\varphi]$ is the time window for imposing constraints at time $t$ \cite{sadraddini2015robust}.



The MPC optimization problem is initially written as (we do not use it for control synthesis as explained shortly): 
\begin{equation}
\label{eq:optimization}
\begin{array}{cl}
\text{minimize} &  J\left(x_t^H, u_t^H \right), \\
  \text{s.t.} & x_{k+1|t}=f^*(x_{k|t},u_{k|t}), k=0,\cdots,H-1, \\
   & \bolds{x}_{|t}[t-h^\varphi+1,t+H] \models \user{ \bolds{G}_{[0,H-1]} \varphi}.
\end{array}
\end{equation}
The set of constraints in \eqref{eq:optimization} requires the knowledge of $x_{t-h^\varphi+1} x_{t-h^\varphi+2} \cdots x_t$. Thus, the proposed control policy requires a finite memory for the history of last $h^\varphi$ states. As it will be shown in Proposition \ref{prop_satisfaction}, persistent feasibility of the constraints in \eqref{eq:optimization} leads to fulfilling $\bolds{G}_{[0,\infty]} \varphi$.  However, persistent feasibility of the MPC setup in \eqref{eq:optimization} is not guaranteed. We address this issue for the remainder of this section.

\begin{define}
An MPC strategy is \emph{recursively feasible} if, for all $t \in \mathbb{N}$, the control at time $t$ is selected such that the MPC optimization problem at $t+1$ becomes feasible.     
\end{define}

Our goal is to modify \eqref{eq:optimization} such that it becomes recursively feasible. 
It is known that adding a (the maximal) RCI set acting as a terminal constraint is sufficient (and necessary) to guarantee recursive feasibility \cite{kerrigan2001feasible}. We add the terminal constraint $\bolds{x} [t+H-h^\varphi:t\revone{+H}] \in \Omega_{\mathcal{L}^\varphi}$ to \eqref{eq:optimization} to obtain:
\begin{equation}
\label{eq:recursive_optimization}
\begin{array}{cl}
u^{H,opt}_t  = &  \underset{u^H_t \in \mathcal{U}^H}{\argmin} ~ J\left(x_t^H, u_t^H \right), \\
 \text{s.t.} & x_{k+1|t}=f^*(x_{k|t},u_{k|t}), k=0,\cdots,H-1, \\
  & \bolds{x}_{|t}[t-h^\varphi+1,t+H] \models \user{\bolds{G}_{[0,H-1]} \varphi}, \\
  & \bolds{x}_{|t} [t+H-h^\varphi:t+H] \in \Omega_{\mathcal{L}^\varphi}.
\end{array}
\end{equation}

\revone{
\begin{proposition}
\label{prop_satisfaction}
Let $\mu_t(x_0,\cdots,x_t)=\mu_t(x_{t-h^\varphi+1},\cdots,x_t)=u^{H,opt}_{0|t}$, where $u^{H,opt}=u^{H,opt}_{0|t} \cdots u^{H,opt}_{H-1|t}$ is given by \eqref{eq:recursive_optimization}. If the optimization problem \eqref{eq:recursive_optimization} is feasible for all $t \ge h^\varphi-1$, then $\bolds{x}(x_0,\mu,\bolds{w})[0] \models \bolds{G}_{[0,\infty]} \varphi, \forall \bolds{w} \in \mathcal{W}^\omega$.  
\end{proposition}
\begin{proof}
We show that $\bolds{x}(x_0,\mu,\bolds{w})[0:k+h^\varphi] \models \bolds{G}_{[0,k]} \varphi, \forall \bolds{w} \in \mathcal{W}^*, \forall k \in \mathbb{N},$ using induction over $k$. 
Consider \eqref{eq:recursive_optimization} for $t=k+h^\varphi-1$ for any $k \in \mathbb{N}$. The second constraint in \eqref{eq:recursive_optimization} requires $\bolds{x}_{|t}[k,k+h^\varphi] \models \varphi$, or equivalently, $x_k \cdots x_{k+h^\varphi-1} x_{1|k+h^\varphi-1} \models \varphi$. By applying $u^{opt}_{0|t}$, monotonicity implies 
$x_{k+h^\varphi}=f(x_{k+h^\varphi-1},u^{opt}_{0|t},w) \preceq x_{1|k+h^\varphi-1} = f^*(x_{k+h^\varphi-1},u^{opt}_{0|t}), \forall w \in \mathcal{W}.$ From Lemma \ref{lemma:finite_run} we have $\bolds{x}[k:k+h^\varphi] \models \varphi$. Thus, we have shown $\bolds{x}[k:k+h^\varphi] \models \varphi, \forall k \in \mathbb{N}$, and the proof is complete.             
\end{proof}
}

\begin{proposition}
The MPC strategy corresponding to \eqref{eq:recursive_optimization} is recursively feasible. 
\end{proposition}
\begin{IEEEproof}
Suppose $u_t^H=u_{0|t} \cdots u_{H-1|t}$ and $x_t^H=x_{t+1|t} \cdots,x_{t+H-1|t}$ is a feasible solution for \eqref{eq:recursive_optimization} at time $t$. Since $\Omega_{\mathcal{L}^\varphi}$ is a RCI set, there exist $u^{r} \in \mathcal{U}$ such that $\bolds{x}_{|t}[{t+H+1}-h^\varphi:t+H+1]=x_{H-h^\varphi+1|t} x_{H-h^\varphi+2|t} \cdots x_{H|t}  f(x_{H|t},u^r,w)  \in \Omega_{\mathcal{L}^\varphi}, \forall w \in \mathcal{W}$. Suppose $u_{0|t}$ is applied to the system. We have $x_{t+1}=f(x_t,u_{0|t},w) \preceq f^*(x_t,u_{0|t})=x_{1|t}, \forall w \in \mathcal{W}$. 

Now, we prove that the optimization problem at time $t+1$ is feasible by showing that at least one feasible solution \userfinal{exists}. Let 
$u^H_{t+1}=u_{1|t} u_{2|t} \cdots,u_{H|t} u^r$. We already showed that $x_{t+1}=x_{0|t+1} \preceq x_{1|t}$. By induction and using monotonicity, it follows that $x_{k|t+1} \preceq x_{k+1|t}, k=,1,\cdots,H-2$.
Therefore, we have $x_{t-h^\varphi+2}\cdots x_{t+1} x_{1|t+1} \cdots x_{H-1|t+1} \preceq x_{t-h^\varphi+2}\cdots x_{1|t} x_{2|t} \cdots x_{H|t}$, which using Lemma \ref{lemma:finite_run} establishes $x_{t-h^\varphi+2}\cdots x_{t+1} x_{1|t+1} \cdots x_{H-1|t+1} \models \bolds{G}_{[0,H-1]} \varphi$. In order to complete the proof, it remains to show that $\bolds{x} [t+H+1-h^\varphi:t+H+1]= x_{H+1-h^\varphi|t} \cdots x_{H|t+1} \models \varphi$. This follows from invariance. Note that $x_{H|t+1}=f^*(x_{H|t},u^r)$. Therefore $x_{H+1-h^\varphi|t}\cdots  \cdots x_{H|t+1} \in \Omega_{\mathcal{L}^\varphi}$, and since $\Omega_{\mathcal{L}^\varphi} \in \mathcal{L}^\varphi$, we have $x_{H+1-h^\varphi|t}\cdots  \cdots x_{H|t+1} \models \varphi$, and the proof is complete. 
\end{IEEEproof}

The MPC optimization problem is also converted into a MILP problem. It is computationally easier to solve  the optimization problem in \eqref{eq:recursive_optimization} by solving $T$ MILPs:
\begin{equation}
\label{eq:MPC_recursive}
\begin{array}{cl}
u^{opt,H}_t = & \underset{u^H_t \in \mathcal{U}^H,i=0,\cdots,T-1}{\argmin} ~  J\left(x_t^H, u_t^H \right), \\
  \text{s.t.} & x_{k+1|t}=f^*(x_{k|t},u_{k|t}), k=0,\cdots,H-1, \\
  & \bolds{x}_{|t}[t-h^\varphi+1,t+H] \models \bolds{G}_{[0,H-1]} \varphi, \\
  & \bolds{x}_{|t} [t+H-h^\varphi:t+H] \in L(\bolds{x}^{\varphi,x_0}[i:i+h^\varphi]).
  \end{array}
\end{equation}
Note that all MILPs can be aggregated into a single large MILP in the expense of additional constraints for capturing non-convexities of the terminal condition. 

Finally, consider $t<h^\varphi$. In this case, we require $H\ge h^\varphi$ and replace the interval $[t-h^\varphi+1,t+H-h^\varphi]$ with $[0,t+H-h^\varphi]$  for $t<h^\varphi$ in \eqref{eq:MPC_recursive}. 
For applications where initialization is not important in long-term (like traffic management), a simpler approach is to initialize the MPC from $t=h^\varphi-1$ and assume all previous state values are zero (hence all the past predicates are evaluated as true).  

\begin{remark}
In our previous work on STL MPC of linear systems \cite{sadraddini2015robust}, we did not establish recursive feasibility. In order to recover from possible infeasibility issues, we proposed maximizing  the STL robustness score (a negative value) whenever the MPC optimization problem became infeasible. Although recursive feasibility is guaranteed here, un-modeled disturbances and initial conditions outside $\mathcal{X}_0^{\max}$ can lead to infeasibility. The formalism in \cite{sadraddini2015robust} can be used to recover from infeasibility with minimal violation of the specification.
\end{remark}

\section{Application to Traffic Management}
In this section, we explain how to apply our methods to traffic management. First, the model that we use for traffic networks is explained, which is similar to the one in \cite{coogan2015controlling} but freeways are also modeled. Next, the monotonicity properties of the model are discussed. We show that there exists a congestion-free set in the state-space in which the traffic dynamics is monotone. Finally, a case study on a mixed urban and freeway network is presented. 

\label{sec:traffic}
\subsection{Model}
The topology of the network is described by a directed graph $(\mathcal{V},\mathcal{L})$, where $\mathcal{V}$ is the set of nodes and $\mathcal{L}$ is the set of edges. Each $l \in \mathcal{L}$ represents a one-way traffic link from tail node $\tau(l) \in \mathcal{V} \cup \emptyset$ to head node $\eta(l) \in \mathcal{V}$, where $\tau(l)=\emptyset$ stands for links originating from outside of the network. We distinguish between three types of links based on their control actuations: 1) $\mathcal{L}_r$: road links actuated by traffic lights, 2) $\mathcal{L}_o$: freeway on-ramps actuated by ramp meters, 3) $\mathcal{L}_f$: freeway segments which are not directly controlled.  Freeway off-ramps are treated the same way as the roads. Uncontrolled roads are also treated the same as freeways. We have $\mathcal{L}_r \cup \mathcal{L}_o \cup \mathcal{L}_f=\mathcal{L}$. 
\begin{remark}
Some works, e.g. \cite{como2015throughput}, consider control over freeway links by varying speed limits, which adds to the control power  but requires the existence of such a control architecture within the infrastructure. We do not consider this type of control actuation in this paper but it can easily be incorporated into our model by modeling freeways links the same way as on-ramps, where the speed limit becomes analogous to the ramp meter input.    
\end{remark}
 
The number of vehicles on link $l$ at time $t$ is represented by $x_{[l],t} \in [0,c_l]$, which is assumed to be a continuous variable, and $c_l$ is the capacity of $l$. In other words, vehicular movements are treated as fluid-like flow in our model. 
The number of vehicles that are able to flow out of $l$ in one time step, if link $l$ is actuated, is:
\begin{equation}
\label{eq:potflow}
q_{[l],t}:= \min \left \{x_{[l],t}, \bar{q}_l, \underset{\{l'  | \tau(l')=\eta(l) \}} \min \frac{\alpha_{l:l'} }{\beta_{l:l'}} (c_{l'}-x_{[l'],t}) \right \},
\end{equation} 
where $\bar{q}_l$ is the maximum outflow of link $l$ in one time step, which is physically related to the speed of the vehicles. The last argument in the minimizer determines the minimum supply available in the downstream links of $l$, where $\alpha_{l:l'} \in [0,1]$ is the capacity ratio of link $l'$ available to vehicles arriving from  link $l$ (typically portion of the lanes), $\beta_{l:l'} \in [0,1]$ is the ratio of the vehicles in $l$ that flow into $l'$ (turning ratio). For simplicity, we assume capacity ratios and turning ratios are constants. System state is represented by $x \in \mathbb{R}^n_+: \{x_{[l]}\}_{l \in \mathcal{L}}$,  where $n$ is the number of the links in the network. The state space is
$ \mathcal{X}:=\prod_{l \in \mathcal{L}} [0,c_l].$

A schematic diagram illustrating the behavior of $q_{[l]}$ with respect to the state variables $x_{[l]},x_{[l']}$- which is known as the \emph{fundamental diagram} in the traffic literature \cite{geroliminis2008existence} - is shown in Fig. \ref{fig:free}. The link flow drops if one (or more) of its downstream links do not have enough capacity to accommodate the incoming flow. In this case (when the last argument in \eqref{eq:potflow} is the minimizer), we say the traffic flow is \emph{congested}. Otherwise, the traffic flow is \emph{free}. 
This motivates the following definition:     
\begin{define}
The \emph{congestion-free set}, denoted by $\Pi$, is defined as the following region in the state space:
\begin{equation}
\begin{array}{cl}
\Pi:=\Big \{ & x\in \mathcal{X} {\Big |}\min\{x_{[l]},\bar{q}_l\} \le \frac{\alpha_{l:l'} }{\beta_{l:l'}} (c_{l'}-x_{[l']}), \\ & \forall l,{l'} \in \mathcal{L}, \tau(l')=\eta(l)  \Big\}.  
\end{array}
\label{eq:free}
\end{equation}
\end{define}
\begin{proposition}
The congestion-free set is a lower-set.
\end{proposition}
\begin{IEEEproof}
Consider $x \in \Pi$ and any $x'\in L(x)$. For all $l,{l'} \in \mathcal{L}, \tau(l')=\eta(l)$, we have $\min\{x'_{l},\bar{q}_l\} \le \min\{x_{[l]},\bar{q}_l\} $ and  $ (c_{l'}-x_{[l]'}) \le (c_{l'}-x'_{l'})$. Therefore, $\min\{x'_{l},\bar{q}_l\} \le \frac{\alpha_{l:l'} }{\beta_{l:l'}} (c_{l'}-x'_{l'})$. Thus $x'\in \Pi$, which indicates $\Pi$ is a lower-set. 
\end{IEEEproof}
Note that $\Pi$ is, in general, non-convex. The predicate $(x \in \Pi)$ can be written as a Boolean logic formula over predicates in the form of \eqref{eq:predicate} as:
\begin{equation}
\begin{array}{l}
\displaystyle \bigwedge_{l,{l'} \in \mathcal{L}, \tau(l')=\eta(l)} \Big ( \left ( (x_{[l]} \le \bar{q}_l) \wedge (x_{[l]}+ \frac{\alpha_{l:l'} }{\beta_{l:l'}} x_{[l]'} \le \frac{\alpha_{l:l'} }{\beta_{l:l'}} c_{l'})  \right) \\ \vee 
 (q_{[l]}+ \frac{\alpha_{l:l'} }{\beta_{l:l'}} x_{[l]'} \le \frac{\alpha_{l:l'} }{\beta_{l:l'}} c_{l'})   \Big ).
\end{array}
\label{eq:free_predicate}
\end{equation}
Notice how the minimizer in \eqref{eq:free} is translated to a disjunction in \eqref{eq:free_predicate}.   

\begin{figure}[t]
\centering
\begin{tikzpicture}[xscale=1.2,yscale=1]
\draw  [fill=green!30]  (0,0) -- (1,1) -- (3,1) -- (3,0) -- cycle;
\draw  [fill=red!30]  (3,1) -- (4,0) -- (3,0) -- cycle;

\draw[color=black]  (0,0) -- (1,1) -- (3,1) -- (4,0) ;
\draw[dashed,color=black]  (3,1) -- (2.5,1.5);
\draw[dashed,color=black]  (0,1) -- (1,1);

\draw[->] (0,0) -- (0,1.5);
\draw[->] (0,0) -- (5,0);
\draw (5.6,0) node[] {$x_{[l]},x_{[l']}$};
\draw (0,1.7) node[] {$q_{[l]}$};
\draw (4,-0.2) node[] {${c}_{l'}$};

\draw (-0.3,1) node[] {$\bar{q}_l$};

\draw (3.6,1.5) node[] {$ \frac{\alpha_{l:l'} }{\beta_{l:l'}} (c_{l'}-x_{[l']})$};
\draw (1.8,0.5) node[] {free flow};
\draw (4.7,0.7) node[] {congested flow};
\draw[color=black,->]  (3.7,0.7) -- (3.3,0.3);
\end{tikzpicture}
\caption{The fundamental diagram. The flow out of link $l$ drops if the number of vehicles on the immediate downstream link $l'$ is close to its capacity. The  congestion is defined by this blocking behavior.}
\label{fig:free}
\end{figure}
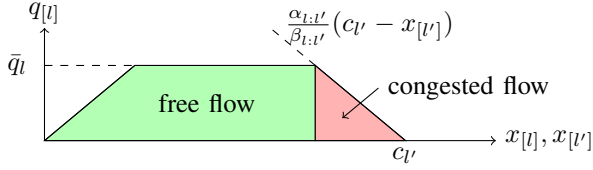

Now we explain the controls. The actuated flow of link $l$ at time $t$ is denoted by $\vec{q}_{[l],t}$, where we have the following relations:
\begin{equation}
\vec{q}_{[l],t}= \left \{ \begin{array}{ll} 
\displaystyle s_{[l],t} q_{[l],t}, & l \in \mathcal{L}_r, \\
\min \{ q_{[l],t}, r_{[l],t} \}, & l \in \mathcal{L}_o, \\
q_{[l],t}, & l \in \mathcal{L}_f, \\
\end{array}
\right.
\label{eq:actuated}
\end{equation}
where $s_{[l],t} \in \{0,1\}$ is the traffic light for link $l$, where $1$ (respectively, $0$) stands for green (respectively, red) light, and $r_{[l],t} \in \mathbb{R}_+$ is the ramp meter input for on-ramp $l$ at time $t$. Ramp meter input limits the number of vehicles that are allowed to enter the freeway in one time step. 
In order to disallow simultaneous green lights for links $l,l'$ (which are typically pair of links pointing toward a common intersection in perpendicular directions), we add the additional constraints  $s_{[l],t}+s_{[l]',t} \le 1$. 
In simple gridded networks, as in our case study network illustrated in Fig. \ref{fig:network}, it is more convenient to define phases for actuation in north-south or east-west directions that are unambiguously mapped to traffic lights for each individual link.
The evolution of the network is given by:
\begin{equation}
\label{eq:flow}
x_{[l],t+1}=x_{[l],t}-\vec{q}_{[l],t}+w_{[l],t}+\sum_{l', \eta(l')=\tau(l)} \beta_{l':l} \vec{q}_{[l'],t},
\end{equation} 
where $w_{[l],t}$ is the number of exogenous vehicles entering link $l$ at time $t$, which is viewed as the adversarial input.
The evolution relation above can be compacted into the form \eqref{eq:system}:
\begin{equation}
\label{eq:traffic}
x_{t+1}=f_{\text{traffic}}(x_t,u_t,w_t),
\end{equation} 
where $u_t$ and $w_t$ are the vector representations for control inputs (combination of traffic lights and ramp meters) and disturbances inputs, respectively.  
Note that $f_{\text{traffic}}$ represents a hybrid system which each mode is affine. The mode is determined by the control inputs and state (which determines the minimizer arguments).  
Some works consider nonlinear representations for the fundamental digram (Fig. \ref{fig:free}), but they still can be approximated using piecewise affine functions.

\subsection{Monotonicity}
\begin{theorem}
System \eqref{eq:traffic} is monotone in $\Pi$.
\label{theorem:monotonicty_of_traffic}  
\end{theorem}

\begin{IEEEproof}
Consider $x',x \in \Pi, x \preceq x'$. We show that $f_{\text{traffic}}(x,u,w) \preceq f_{\text{traffic}}(x',u,w), \forall w \in \mathcal{W}, \forall u \in \mathcal{U}$. Observe in \eqref{eq:flow} that we only need to verify is proving that $x_{[l]}-\vec{q}_{[l]}$ is a non-decreasing function of $x_{[l]}$ as all other terms are additive and non-decreasing with respect to $x$. Since $x,x' \in \Pi$, the last argument in \eqref{eq:potflow} is never the minimizer. Thus, for all $\l \in \mathcal{L}$, we have $x_{[l]}-\vec{q}_{[l]} \in \{0, x_{[l]}-r_{[l]},x_{[l]}-c_l,x_{[l]}\}$, depending on the mode of the system and actuations, which all are non-decreasing functions of $x_{[l]}$. Thus, $f_\text{traffic}$ is monotone in $\Pi$.   
\end{IEEEproof}
The primary objective in our traffic management approach is finding control policies such that the state is restricted to $\Pi$, which not only eliminates congestion, but also ensures that the system is monotone hence the methods of this paper become applicable. It is  worth to note that the traffic system becomes non-monotone when flow is congested in diverging junctions, as shown in \cite{Coogan2014}. This phenomena is attributed to the first-in-first-out (FIFO) nature of the model.    
By assuming fully non-FIFO models, system becomes monotone in the whole state space. For a more thorough discussion on physical aspects of monotonicity in traffic networks, see \cite{coogan2015mixed}.

The maximal system in \eqref{eq:traffic} corresponds to the scenario where each $w_l$ is equal to its maximum allowed value $w_l^*$. 


\subsection{Case Study}

\subsubsection*{Network}

Consider the network in Fig. \ref{fig:network}, which consists of urban roads (links 1-26, 27,29,31,33 and 49-53), freeway segments (links 35-48) and freeway on-ramps (links 28,30,32,34). The layout of the network illustrates a freeway passing by an urban area, which is common in many realistic traffic layouts. There are 14 intersections (nodes a-n) controlled by traffic lights. Each intersection has two modes of actuation: north-south (NS) and east-west (EW). There are four entries to the freeway (nodes o-r) that are regulated by ramp meters. We have $n=53$ and $\mathcal{U}=\mathbb{R}_+^4 \times \{0,1\}^{14}$. Vehicles arrive from links 1,6,11,15,19,23,35,42,49 and 52. The parameters of the network are shown in Table \ref{table:parameters}.

 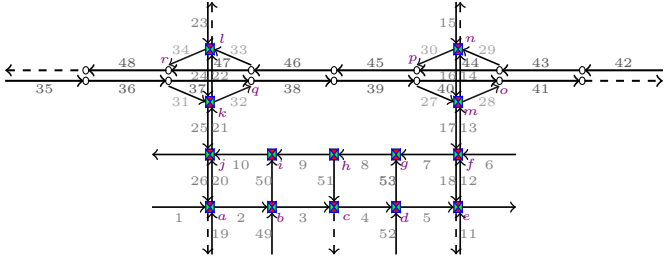
\begin{figure}[t]
 \vspace{10pt}
\centering
\tikzset{
  font={\fontsize{0.5pt}{2}\selectfont}}
\begin{tikzpicture}[xscale=0.55,yscale=0.7]
\draw [line width=0.22mm,->] (  -4.4 , 0 ) -- (  -3.1 , 0 );
\draw [line width=0.22mm,->] (  -2.9 , 0 ) -- (  -1.6 , 0 );
\draw [] (  -3.1 , -0.1  ) rectangle ( -2.9 , 0.1 );
\draw [blue, fill=red] ( -3.1 , 0.1 ) -- ( -2.9 , 0.1 ) -- ( -3 , 0 )-- cycle;
\draw [blue, fill=red] ( -3.1 , -0.1 ) -- ( -2.9 , -0.1 ) -- ( -3 , 0 )-- cycle;
\draw [blue, fill=green] ( -3.1 , -0.1 ) -- ( -3.1 , 0.1 ) -- ( -3 , 0 )-- cycle;
\draw [blue, fill=green] ( -2.9 , -0.1 ) -- ( -2.9 , 0.1 ) -- ( -3 , 0 )-- cycle;
\draw [line width=0.22mm,->] (  -1.4 , 0 ) -- (  -0.1 , 0 );
\draw [] (  -1.6 , -0.1  ) rectangle ( -1.4 , 0.1 );
\draw [blue, fill=red] ( -1.6 , 0.1 ) -- ( -1.4 , 0.1 ) -- ( -1.5 , 0 )-- cycle;
\draw [blue, fill=red] ( -1.6 , -0.1 ) -- ( -1.4 , -0.1 ) -- ( -1.5 , 0 )-- cycle;
\draw [blue, fill=green] ( -1.6 , -0.1 ) -- ( -1.6 , 0.1 ) -- ( -1.5 , 0 )-- cycle;
\draw [blue, fill=green] ( -1.4 , -0.1 ) -- ( -1.4 , 0.1 ) -- ( -1.5 , 0 )-- cycle;
\draw [line width=0.22mm,->] (  0.1 , 0 ) -- (  1.4 , 0 );
\draw [] (  -0.1 , -0.1  ) rectangle ( 0.1 , 0.1 );
\draw [blue, fill=red] ( -0.1 , 0.1 ) -- ( 0.1 , 0.1 ) -- ( 0 , 0 )-- cycle;
\draw [blue, fill=red] ( -0.1 , -0.1 ) -- ( 0.1 , -0.1 ) -- ( 0 , 0 )-- cycle;
\draw [blue, fill=green] ( -0.1 , -0.1 ) -- ( -0.1 , 0.1 ) -- ( 0 , 0 )-- cycle;
\draw [blue, fill=green] ( 0.1 , -0.1 ) -- ( 0.1 , 0.1 ) -- ( 0 , 0 )-- cycle;
\draw [line width=0.22mm,->] (  1.6 , 0 ) -- (  2.9 , 0 );
\draw [] (  1.4 , -0.1  ) rectangle ( 1.6 , 0.1 );
\draw [blue, fill=red] ( 1.4 , 0.1 ) -- ( 1.6 , 0.1 ) -- ( 1.5 , 0 )-- cycle;
\draw [blue, fill=red] ( 1.4 , -0.1 ) -- ( 1.6 , -0.1 ) -- ( 1.5 , 0 )-- cycle;
\draw [blue, fill=green] ( 1.4 , -0.1 ) -- ( 1.4 , 0.1 ) -- ( 1.5 , 0 )-- cycle;
\draw [blue, fill=green] ( 1.6 , -0.1 ) -- ( 1.6 , 0.1 ) -- ( 1.5 , 0 )-- cycle;
\draw [] (  2.9 , -0.1  ) rectangle ( 3.1 , 0.1 );
\draw [blue, fill=red] ( 2.9 , 0.1 ) -- ( 3.1 , 0.1 ) -- ( 3 , 0 )-- cycle;
\draw [blue, fill=red] ( 2.9 , -0.1 ) -- ( 3.1 , -0.1 ) -- ( 3 , 0 )-- cycle;
\draw [blue, fill=green] ( 2.9 , -0.1 ) -- ( 2.9 , 0.1 ) -- ( 3 , 0 )-- cycle;
\draw [blue, fill=green] ( 3.1 , -0.1 ) -- ( 3.1 , 0.1 ) -- ( 3 , 0 )-- cycle;
\draw [line width=0.22mm,->] (  3.1 , 0 ) -- (  4.4 , 0 );
\draw [line width=0.22mm,<-] (  -4.4 , 1 ) -- (  -3.1 , 1 );
\draw [line width=0.22mm,<-] (  -2.9 , 1 ) -- (  -1.6 , 1 );
\draw [] (  -3.1 , 0.9  ) rectangle ( -2.9 , 1.1 );
\draw [blue, fill=red] ( -3.1 , 1.1 ) -- ( -2.9 , 1.1 ) -- ( -3 , 1 )-- cycle;
\draw [blue, fill=red] ( -3.1 , 0.9 ) -- ( -2.9 , 0.9 ) -- ( -3 , 1 )-- cycle;
\draw [blue, fill=green] ( -3.1 , 0.9 ) -- ( -3.1 , 1.1 ) -- ( -3 , 1 )-- cycle;
\draw [blue, fill=green] ( -2.9 , 0.9 ) -- ( -2.9 , 1.1 ) -- ( -3 , 1 )-- cycle;
\draw [line width=0.22mm,<-] (  -1.4 , 1 ) -- (  -0.1 , 1 );
\draw [] (  -1.6 , 0.9  ) rectangle ( -1.4 , 1.1 );
\draw [blue, fill=red] ( -1.6 , 1.1 ) -- ( -1.4 , 1.1 ) -- ( -1.5 , 1 )-- cycle;
\draw [blue, fill=red] ( -1.6 , 0.9 ) -- ( -1.4 , 0.9 ) -- ( -1.5 , 1 )-- cycle;
\draw [blue, fill=green] ( -1.6 , 0.9 ) -- ( -1.6 , 1.1 ) -- ( -1.5 , 1 )-- cycle;
\draw [blue, fill=green] ( -1.4 , 0.9 ) -- ( -1.4 , 1.1 ) -- ( -1.5 , 1 )-- cycle;
\draw [line width=0.22mm,<-] (  0.1 , 1 ) -- (  1.4 , 1 );
\draw [] (  -0.1 , 0.9  ) rectangle ( 0.1 , 1.1 );
\draw [blue, fill=red] ( -0.1 , 1.1 ) -- ( 0.1 , 1.1 ) -- ( 0 , 1 )-- cycle;
\draw [blue, fill=red] ( -0.1 , 0.9 ) -- ( 0.1 , 0.9 ) -- ( 0 , 1 )-- cycle;
\draw [blue, fill=green] ( -0.1 , 0.9 ) -- ( -0.1 , 1.1 ) -- ( 0 , 1 )-- cycle;
\draw [blue, fill=green] ( 0.1 , 0.9 ) -- ( 0.1 , 1.1 ) -- ( 0 , 1 )-- cycle;
\draw [line width=0.22mm,<-] (  1.6 , 1 ) -- (  2.9 , 1 );
\draw [] (  1.4 , 0.9  ) rectangle ( 1.6 , 1.1 );
\draw [blue, fill=red] ( 1.4 , 1.1 ) -- ( 1.6 , 1.1 ) -- ( 1.5 , 1 )-- cycle;
\draw [blue, fill=red] ( 1.4 , 0.9 ) -- ( 1.6 , 0.9 ) -- ( 1.5 , 1 )-- cycle;
\draw [blue, fill=green] ( 1.4 , 0.9 ) -- ( 1.4 , 1.1 ) -- ( 1.5 , 1 )-- cycle;
\draw [blue, fill=green] ( 1.6 , 0.9 ) -- ( 1.6 , 1.1 ) -- ( 1.5 , 1 )-- cycle;
\draw [] (  2.9 , 0.9  ) rectangle ( 3.1 , 1.1 );
\draw [blue, fill=red] ( 2.9 , 1.1 ) -- ( 3.1 , 1.1 ) -- ( 3 , 1 )-- cycle;
\draw [blue, fill=red] ( 2.9 , 0.9 ) -- ( 3.1 , 0.9 ) -- ( 3 , 1 )-- cycle;
\draw [blue, fill=green] ( 2.9 , 0.9 ) -- ( 2.9 , 1.1 ) -- ( 3 , 1 )-- cycle;
\draw [blue, fill=green] ( 3.1 , 0.9 ) -- ( 3.1 , 1.1 ) -- ( 3 , 1 )-- cycle;
\draw [line width=0.22mm,<-] (  3.1 , 1 ) -- (  4.4 , 1 );
\draw [line width=0.22mm,->] (  -1.5 , 0.1 ) -- (  -1.5 , 0.9 );
\draw [line width=0.22mm,<-] (  0 , 0.1 ) -- (  0 , 0.9 );
\draw [line width=0.22mm,->] (  1.5 , 0.1 ) -- (  1.5 , 0.9 );
\draw [line width=0.22mm,->] (  -1.5 , -0.9 ) -- (  -1.5 , -0.1 );
\draw [line width=0.22mm,<-,dashed] (  0 , -0.9 ) -- (  0 , -0.1 );
\draw [line width=0.22mm,->] (  1.5 , -0.9 ) -- (  1.5 , -0.1 );
\draw [line width=0.22mm,<-,dashed] (  -3.05 , -0.9 ) -- (  -3.05 , -0.1 );
\draw [line width=0.22mm,->] (  -2.95 , -0.9 ) -- (  -2.95 , -0.1 );
\draw [line width=0.22mm,<-] (  -3.05 , 0.1 ) -- (  -3.05 , 0.9 );
\draw [line width=0.22mm,->] (  -2.95 , 0.1 ) -- (  -2.95 , 0.9 );
\draw [line width=0.22mm,<-] (  -3.05 , 1.1 ) -- (  -3.05 , 1.9 );
\draw [line width=0.22mm,->] (  -2.95 , 1.1 ) -- (  -2.95 , 1.9 );
\draw [line width=0.22mm,<-] (  -3.05 , 2.1 ) -- (  -3.05 , 2.9 );
\draw [line width=0.22mm,->] (  -2.95 , 2.1 ) -- (  -2.95 , 2.9 );
\draw [line width=0.22mm,<-] (  -3.05 , 3.1 ) -- (  -3.05 , 3.9 );
\draw [line width=0.22mm,->,dashed] (  -2.95 , 3.1 ) -- (  -2.95 , 3.9 );
\draw [line width=0.22mm,<-,dashed] (  2.95 , -0.9 ) -- (  2.95 , -0.1 );
\draw [line width=0.22mm,->] (  3.05 , -0.9 ) -- (  3.05 , -0.1 );
\draw [line width=0.22mm,<-] (  2.95 , 0.1 ) -- (  2.95 , 0.9 );
\draw [line width=0.22mm,->] (  3.05 , 0.1 ) -- (  3.05 , 0.9 );
\draw [line width=0.22mm,<-] (  2.95 , 1.1 ) -- (  2.95 , 1.9 );
\draw [line width=0.22mm,->] (  3.05 , 1.1 ) -- (  3.05 , 1.9 );
\draw [line width=0.22mm,<-] (  2.95 , 2.1 ) -- (  2.95 , 2.9 );
\draw [line width=0.22mm,->] (  3.05 , 2.1 ) -- (  3.05 , 2.9 );
\draw [line width=0.22mm,<-] (  2.95 , 3.1 ) -- (  2.95 , 3.9 );
\draw [line width=0.22mm,->,dashed] (  3.05 , 3.1 ) -- (  3.05 , 3.9 );
\draw [] (  -3.1 , 1.9  ) rectangle ( -2.9 , 2.1 );
\draw [blue, fill=red] ( -3.1 , 2.1 ) -- ( -2.9 , 2.1 ) -- ( -3 , 2 )-- cycle;
\draw [blue, fill=red] ( -3.1 , 1.9 ) -- ( -2.9 , 1.9 ) -- ( -3 , 2 )-- cycle;
\draw [blue, fill=green] ( -3.1 , 1.9 ) -- ( -3.1 , 2.1 ) -- ( -3 , 2 )-- cycle;
\draw [blue, fill=green] ( -2.9 , 1.9 ) -- ( -2.9 , 2.1 ) -- ( -3 , 2 )-- cycle;
\draw [] (  -3.1 , 2.9  ) rectangle ( -2.9 , 3.1 );
\draw [blue, fill=red] ( -3.1 , 3.1 ) -- ( -2.9 , 3.1 ) -- ( -3 , 3 )-- cycle;
\draw [blue, fill=red] ( -3.1 , 2.9 ) -- ( -2.9 , 2.9 ) -- ( -3 , 3 )-- cycle;
\draw [blue, fill=green] ( -3.1 , 2.9 ) -- ( -3.1 , 3.1 ) -- ( -3 , 3 )-- cycle;
\draw [blue, fill=green] ( -2.9 , 2.9 ) -- ( -2.9 , 3.1 ) -- ( -3 , 3 )-- cycle;
\draw [] (  2.9 , 1.9  ) rectangle ( 3.1 , 2.1 );
\draw [blue, fill=red] ( 2.9 , 2.1 ) -- ( 3.1 , 2.1 ) -- ( 3 , 2 )-- cycle;
\draw [blue, fill=red] ( 2.9 , 1.9 ) -- ( 3.1 , 1.9 ) -- ( 3 , 2 )-- cycle;
\draw [blue, fill=green] ( 2.9 , 1.9 ) -- ( 2.9 , 2.1 ) -- ( 3 , 2 )-- cycle;
\draw [blue, fill=green] ( 3.1 , 1.9 ) -- ( 3.1 , 2.1 ) -- ( 3 , 2 )-- cycle;
\draw [] (  2.9 , 2.9  ) rectangle ( 3.1 , 3.1 );
\draw [blue, fill=red] ( 2.9 , 3.1 ) -- ( 3.1 , 3.1 ) -- ( 3 , 3 )-- cycle;
\draw [blue, fill=red] ( 2.9 , 2.9 ) -- ( 3.1 , 2.9 ) -- ( 3 , 3 )-- cycle;
\draw [blue, fill=green] ( 2.9 , 2.9 ) -- ( 2.9 , 3.1 ) -- ( 3 , 3 )-- cycle;
\draw [blue, fill=green] ( 3.1 , 2.9 ) -- ( 3.1 , 3.1 ) -- ( 3 , 3 )-- cycle;
\draw ( -6 , 2.4 ) circle (  0.07 );
\draw ( -6 , 2.6 ) circle (  0.07 );
\draw [line width=0.27mm,->] (  -7.95 , 2.4 ) -- (  -6.05 , 2.4 );
\draw [line width=0.27mm,<-,dashed] (  -7.95 , 2.6 ) -- (  -6.05 , 2.6 );
\draw ( -4 , 2.4 ) circle (  0.07 );
\draw ( -4 , 2.6 ) circle (  0.07 );
\draw [line width=0.27mm,->] (  -5.95 , 2.4 ) -- (  -4.05 , 2.4 );
\draw [line width=0.27mm,<-] (  -5.95 , 2.6 ) -- (  -4.05 , 2.6 );
\draw ( -2 , 2.4 ) circle (  0.07 );
\draw ( -2 , 2.6 ) circle (  0.07 );
\draw [line width=0.27mm,->] (  -3.95 , 2.4 ) -- (  -2.05 , 2.4 );
\draw [line width=0.27mm,<-] (  -3.95 , 2.6 ) -- (  -2.05 , 2.6 );
\draw ( 0 , 2.4 ) circle (  0.07 );
\draw ( 0 , 2.6 ) circle (  0.07 );
\draw [line width=0.27mm,->] (  -1.95 , 2.4 ) -- (  -0.05 , 2.4 );
\draw [line width=0.27mm,<-] (  -1.95 , 2.6 ) -- (  -0.05 , 2.6 );
\draw ( 2 , 2.4 ) circle (  0.07 );
\draw ( 2 , 2.6 ) circle (  0.07 );
\draw [line width=0.27mm,->] (  0.05 , 2.4 ) -- (  1.95 , 2.4 );
\draw [line width=0.27mm,<-] (  0.05 , 2.6 ) -- (  1.95 , 2.6 );
\draw ( 4 , 2.4 ) circle (  0.07 );
\draw ( 4 , 2.6 ) circle (  0.07 );
\draw [line width=0.27mm,->] (  2.05 , 2.4 ) -- (  3.95 , 2.4 );
\draw [line width=0.27mm,<-] (  2.05 , 2.6 ) -- (  3.95 , 2.6 );
\draw ( 6 , 2.4 ) circle (  0.07 );
\draw ( 6 , 2.6 ) circle (  0.07 );
\draw [line width=0.27mm,->] (  4.05 , 2.4 ) -- (  5.95 , 2.4 );
\draw [line width=0.27mm,<-] (  4.05 , 2.6 ) -- (  5.95 , 2.6 );
\draw [line width=0.27mm,->,dashed] (  6.05 , 2.4 ) -- (  7.95 , 2.4 );
\draw [line width=0.27mm,<-] (  6.05 , 2.6 ) -- (  7.95 , 2.6 );

\draw [line width=0.215mm,->,black!90] (  -2.9 , 2.0 ) -- (  -2.0 , 2.3 );
\draw [line width=0.215mm,<-] (  -2.9 , 3 ) -- (  -2 , 2.7 );
\draw [line width=0.215mm,->,black!90] (  -3.1 , 3 ) -- (  -4 , 2.7 );
\draw [line width=0.215mm,<-] (  -3.1 , 2 ) -- (  -4 , 2.3 );

\draw [line width=0.215mm,<-] (  2.9 , 2 ) -- (  2 , 2.3 );
\draw [line width=0.215mm,->,black!90] (  2.9 , 3 ) -- (  2 , 2.7 );
\draw [line width=0.215mm,<-] (  3.1 , 3 ) -- (  4 , 2.7 );
\draw [line width=0.215mm,->,black!90] (  3.1 , 2 ) -- (  4 , 2.3 );

\draw (-3.75,-0.2) node[black!50] {$1$};
\draw (-2.25,-0.2) node[black!50] {$2$};
\draw (-0.75,-0.2) node[black!50] {$3$};
\draw (0.75,-0.2) node[black!50] {$4$};
\draw (2.25,-0.2) node[black!50] {$5$};

\draw (3.75,0.8) node[black!50] {$6$};
\draw (2.25,0.8) node[black!50] {$7$};
\draw (0.75,0.8) node[black!50] {$8$};
\draw (-0.75,0.8) node[black!50] {$9$};
\draw (-2.25,0.8) node[black!50] {$10$};

\draw (3.25,-0.5) node[black!50] {$11$};
\draw (3.25,0.5) node[black!50] {$12$};
\draw (3.25,1.5) node[black!50] {$13$};
\draw (3.25,2.5) node[black!50] {$14$};

\draw (2.75,3.5) node[black!50] {$15$};
\draw (2.75,2.5) node[black!50] {$16$};
\draw (2.75,1.5) node[black!50] {$17$};
\draw (2.75,0.5) node[black!50] {$18$};

\draw (-2.75,-0.5) node[black!50] {$19$};
\draw (-2.75,0.5) node[black!50] {$20$};
\draw (-2.75,1.5) node[black!50] {$21$};
\draw (-2.75,2.5) node[black!50] {$22$};

\draw (-3.25,3.5) node[black!50] {$23$};
\draw (-3.25,2.5) node[black!50] {$24$};
\draw (-3.25,1.5) node[black!50] {$25$};
\draw (-3.25,0.5) node[black!50] {$26$};

\draw (2.3,2.0) node[black!35] {$27$};
\draw (3.7,2.0) node[black!35] {$28$};
\draw (3.7,3.0) node[black!35] {$29$};
\draw (2.3,3) node[black!35] {$30$};

\draw (-3.7,2.0) node[black!35] {$31$};
\draw (-2.3,2.0) node[black!35] {$32$};
\draw (-2.3,3.0) node[black!35] {$33$};
\draw (-3.7,3) node[black!35] {$34$};

\draw (-7,2.25) node[black!65] {$35$};
\draw (-5,2.25) node[black!65] {$36$};
\draw (-3.3,2.25) node[black!65] {$37$};
\draw (-1,2.25) node[black!65] {$38$};
\draw (1,2.25) node[black!65] {$39$};
\draw (2.7,2.25) node[black!65] {$40$};
\draw (5,2.25) node[black!65] {$41$};

\draw (7,2.75) node[black!65] {$42$};
\draw (5,2.75) node[black!65] {$43$};
\draw (3.3,2.75) node[black!65] {$44$};
\draw (1,2.75) node[black!65] {$45$};
\draw (-1,2.75) node[black!65] {$46$};
\draw (-2.7,2.75) node[black!65] {$47$};
\draw (-5,2.75) node[black!65] {$48$};

\draw (-1.7,-0.5) node[black!50] {$49$};
\draw (-1.7,0.5) node[black!50] {$50$};
\draw (-0.2,0.5) node[black!50] {$51$};
\draw (1.3,-0.5) node[black!50] {$52$};
\draw (1.3,0.5) node[black!50] {$53$};

\draw (1.3,0.5) node[black!50] {$53$};

\draw (-2.7,-0.2) node[red!50!blue] {$a$};
\draw (-1.3,-0.2) node[red!50!blue] {$b$};
\draw (0.3,-0.2) node[red!50!blue] {$c$};
\draw (1.7,-0.2) node[red!50!blue] {$d$};
\draw (3.2,-0.2) node[red!50!blue] {$e$};
\draw (-2.7,0.8) node[red!50!blue] {$j$};
\draw (-1.3,0.8) node[red!50!blue] {$i$};
\draw (0.3,0.8) node[red!50!blue] {$h$};
\draw (1.7,0.8) node[red!50!blue] {$g$};
\draw (3.3,0.8) node[red!50!blue] {$f$};
\draw (-2.7,1.8) node[red!50!blue] {$k$};
\draw (-2.7,3.2) node[red!50!blue] {$l$};
\draw (3.3,1.8) node[red!50!blue] {$m$};
\draw (3.3,3.2) node[red!50!blue] {$n$};
\draw (-1.9,2.2) node[red!50!blue] {$q$};
\draw (4.1,2.2) node[red!50!blue] {$o$};
\draw (1.9,2.8) node[red!50!blue] {$p$};
\draw (-4.1,2.8) node[red!50!blue] {$r$};

\end{tikzpicture}
\caption{Traffic management case study: A network of freeways and urban roads. There are 14 intersections controlled by traffic lights and 4 ramp meters.}
\label{fig:network}
\end{figure}

\begin{table}[t]
\centering
\caption{Parameters of the network in Fig. \ref{fig:network}}
\resizebox{0.49\textwidth}{!}{%
\begin{tabular}{|c|c|}
\hline 
links & parameters  \\
\hline
$1-26,49-53$ & $\bar{q}_l=15,c_l=40$ \\ \hline
$27-34$ & $\bar{q}_l=15,c_l=30$ \\ \hline
$35-48$ & $\bar{q}_l=40,c_l=60$ \\
\hline \hline
Turning ratios & value \\
\hline
$\begin{array}{l} \beta_{ 2 : 50 }, \beta_{ 4 : 53 }, \beta_{ 8 : 51 }, \beta_{ 12 : 7 }, \beta_{ 13 : 28 }, \beta_{ 15 : 30 }, \beta_{ 16 : 28 }, \\ \beta_{ 21 : 32 }, \beta_{ 24 : 32 }, \beta_{ 26 : 2 }, \beta_{ 36 : 31 }, \beta_{ 36 : 33 }, \beta_{ 39 : 27 }, \beta_{ 43 : 29 }
\end{array} $ & 0.2 \\
\hline
$\begin{array}{l} \beta_{ 5 : 12 }, \beta_{ 6 : 13 }, \beta_{ 6 : 18 }, \beta_{ 10 : 21 }, \beta_{ 10 : 26 }
 \end{array} $ & 0.3 \\

\hline
$\beta_{ 1 : 20 }, \beta_{ 6 : 7 } $ & 0.4 \\

\hline

$ \begin{array}{l}
\beta_{ 1 : 2 }, \beta_{ 11 : 12 }, \beta_{ 14 : 30 }, \beta_{ 17 : 7 }, \beta_{ 17 : 18 }, \beta_{ 19 : 2 }, \beta_{ 19 : 20 },\\ \beta_{ 22 : 34 }, \beta_{ 23 : 24 }, \beta_{ 23 : 34 }, \beta_{ 27 : 14 }, \beta_{ 27 : 17 }, \beta_{ 29 : 16 }, \beta_{ 31 : 22 },\\ \beta_{ 31 : 25 }, \beta_{ 33 : 24 }, \beta_{ 49 : 3 }, \beta_{ 49 : 50 }, \beta_{ 51 : 4 }, \beta_{ 52 : 5 }, \beta_{ 52 : 53 },
    \end{array} $ & 0.5 \\
\hline

$\begin{array}{l} \beta_{ 2 : 3 }, \beta_{ 3 : 4 }, \beta_{ 4 : 5 }, \beta_{ 8 : 9 }, \beta_{ 12 : 13 }, \beta_{ 13 : 14 }, \beta_{ 15 : 16 }, \beta_{ 16 : 17 }, \beta_{ 20 : 21 }, \\ \beta_{ 21 : 22 }, \beta_{ 24 : 25 }, \beta_{ 25 : 26 },  \beta_{ 36 : 37 }, \beta_{ 39 : 40 }, \beta_{ 43 : 44 }, \beta_{ 46 : 47 },

 \end{array} $ & 0.8 \\
\hline \hline
Capacity ratios & value \\
\hline
$\begin{array}{l} \alpha_{19:2}, \alpha_{26:2}, \alpha_{17:7}, \alpha_{12:7}, \alpha_{13:28},  \alpha_{16:28} \\  \alpha_{14:30}, \alpha_{15:30}, \alpha_{21:32}, \alpha_{24:32}, \alpha_{22:34}, \alpha_{23:34}\end{array} $ & 0.5 \\
\hline
\hline
\multicolumn{2}{|c|}{Disturbances (arrival rates)} \\
\hline
\multicolumn{2}{|c|}{$\begin{array}{c} w_1^*=w_6^*=4.5, w_{11}^*=w_{15}^*=w_{19}^*=5, w_{23}^*=6 \\ w_{35}^*=w_{42}^*=20, w_{49}^*=w_{52}^*=2 \end{array}$}
\\

\hline
\end{tabular}
}
\vspace{1pt}
\label{table:parameters}
\end{table}

\subsubsection*{\userfinal{Specification}}
As mentioned earlier, the primary objective is keeping the state in the congestion-free set. In addition, since the demand for the north-south side roads (links 49-53) is smaller than the traffic in the east-west roads, we add a timed liveness requirement for the traffic flow on links 49-53:
\begin{equation*}
\psi= \bigwedge_{l=49,50,\cdots,53} (x_{[l]} \ge 5) \Rightarrow \bolds{F}_{[0,3]} (x_{[l]}\le 5),
\end{equation*}
which states that ``if the number of vehicles on any of the north-south side roads exceeds 5, their flow is eventually actuated within three time units ahead".  
The global specification is given as:
\begin{equation}
\phi=\bolds{G}_{[0,\infty]} \left ( (x\in \Pi) \wedge \psi \right).
\end{equation}
Note that $h^\varphi=3$, $\varphi=(x\in \Pi) \wedge \psi$. 

\vspace{2pt}

\subsubsection*{Open-loop Control Policy}
We use Theorem \ref{theorem:phi-sequence}. The shortest $\phi$-sequence that we found for this problem has $T=5, T_0=0$. The corresponding MILP had 2357 variables (of which 1061 were binary) and 4037 constraints \footnote{The scripts for this case study are available in \texttt{http://blogs.bu.edu/sadra/format-monotone}}, which is solved using the Gurobi MILP solver in less than 6 seconds on a dual core 3.0 GHz MacBook Pro. The cost is set to zero in order to just check for feasibility. Even though finding an optimal solution and checking for feasibility of a MILP have the same theoretical complexity, the latter is executed much faster in practice. For instance, finding a $\phi$-sequence, while \userfinal{minimizing or maximizing $\sum_{k=0}^{7} \|x^\phi_k\|_1$ both took more than 20 minutes}. Note that it is virtually intractable to attack a problem of this size (53 dimensional state) using any method that involves state-space discretization, such as the method in \cite{coogan2015efficient} (e.g., if each state-component is partitioned into 2 intervals, the finite-state problem size will be $2^{53}$).

Monotonicity implies that any demand set $\mathcal{W}$ for which there exists a solution to Problem \ref{prob:global} is a lower-set. The set corresponding to the values at the bottom of Table \ref{table:parameters} is one of them. Table \ref{table_scenario} shows \userfinal{results on existence of $\phi$-sequences} for some other demand scenarios. Computation times for solving a MILP do not demonstrate a generic behavior. For the rest of this section, the numerical examples are reported for the values in Table \ref{table:parameters}.    

\begin{table}[t]
\centering
\revtwo{
\caption{Existence of $\phi$-sequences for the network in Fig. \ref{fig:network}}
\resizebox{0.49\textwidth}{!}{%
\begin{tabular}{|c|c|c|c|}
\hline 
Demand Changes from Table \ref{table:parameters} & $T$ & Existence & Comp. Time (s) \\
\hline
- & 5 & yes & 6 \\ \hline
- & 6 & no & 4 \\ \hline
- & 7 & no & 10 \\ \hline
- & 8 & no & 75 \\ \hline
- & 9 & no & 11 \\ \hline
- & 10 & yes & 36 \\ \hline
$w^*_1=w^*_6=3,w^*_{11}=w^*_{15}=w^*_{19}=6$  & 5 & yes & 5 \\ \hline
$w^*_1=w^*_6=4,w^*_{11}=w^*_{15}=w^*_{19}=6$  & 5 & no & 0.5 \\ \hline
$w^*_1=w^*_6=1.5,w^*_{49}=w^*_{52}=3.5$  & 5 & yes & 16 \\ \hline
$w^*_1=w^*_6=7.5,w^*_{11}=w^*_{15}=w^*_{19}=w^*_{23}=2$  & 6 & yes & 9 \\ \hline
$w^*_1=w^*_6=9,w^*_{11}=w^*_{15}=w^*_{19}=w^*_{23}=1$  & 5 & yes & 4 \\ \hline
$w^*_1=w^*_6=10,w^*_{11}=w^*_{15}=w^*_{19}=w^*_{23}=0$  & 30 & no & 3.5 \\ \hline
$w^*_{15}=w^*_{23}=8,w^*_{35}=w^*_{42}=10$  & 6 & yes & 23 \\ \hline
$w^*_{15}=w^*_{23}=0,w^*_{35}=w^*_{42}=30$  & 5 & yes & 4 \\ \hline
\end{tabular}
}
}
\vspace{1pt}
\label{table_scenario}
\end{table}

The control values in the $\phi$-sequence are shown in Table \ref{table:controls}. As stated in Theorem \ref{theorem:phi-sequence}, starting from an initial condition in $L(x_0$), applying the open-loop control policy \eqref{eq:phi} guarantees satisfaction of the specification. 
In other words, after applying the initialization segment, the repetitive controls in Table \ref{table:controls} become a fixed time-table for the inputs of the traffic lights and the ramp meters. Starting from $x_0$, which is a $53$-dimensional vector, we apply \eqref{eq:phi} using the values in Table \ref{table:controls}. The trajectory of the maximal system is shown in Fig. \ref{fig:case_results} [Top]. The traffic signals are coordinated such that the traffic flows free of congestion. The black dashed lines represent the capacity of the links, and the dashed line in the fourth figure (from the left) represents the threshold for the liveness sub-specification ($\psi$). It is observed that all the state values for side road links (49-53) persistently fall below the threshold. The robustness values for $(x\in \Pi)$ and $\psi$ are shown in the fifth figure. As mentioned earlier, robustness corresponds to the minimum volume of vehicles that the system is away from congestion, or violating the specification. The robustness values are always positive, indicating satisfaction.

 As stated in Theorem \ref{theorem:orbit}, the trajectory of the maximal system converges to a periodic orbit. 
It is worth to \userfinal{note} that the number of vehicles on freeway links is significantly smaller than its capacity, which is attributed to the fact that the number designated for $\bar{q}$ (related to the maximum speed) of freeway links is relatively large (30, as opposed to 15 for roads). Therefore, freeway links are utilized in a way that there is enough space for high speed non-congested flow.   

\begin{table}[t]
\centering
\caption{$\phi$-sequence in the case study}
\resizebox{0.49\textwidth}{!}{%
\begin{tabular}{|c||c|c|c|c|c|c|c|c|}
\hline
- & \multicolumn{3}{|c|}{Initialization} & \multicolumn{5}{|c|}{Repetitive Controls} \\
\hline
node & $u_0^\phi$ & $u_1^\phi$ & $u_2^\phi$ & $u_3^\phi$ & $u_4^\phi$ & $u_5^\phi$ & $u_6^\phi$ & $u_7^\phi$ \\
\hline
\hline 
 {\color{red!50!blue} $a$} &  $EW$ &   $NS$ &   $NS$ &   $NS$ &   $EW$ &   $EW$ &   $NS$ &   $NS$    \\ 
\hline 
 {\color{red!50!blue} $b$} &  $NS$ &   $EW$ &   $EW$ &   $EW$ &   $NS$ &   $NS$ &   $EW$ &   $EW$    \\ 
\hline 
 {\color{red!50!blue} $c$} &  $EW$ &   $NS$ &   $NS$ &   $EW$ &   $EW$ &   $EW$ &   $NS$ &   $NS$    \\ 
\hline 
 {\color{red!50!blue} $d$} &  $EW$ &   $NS$ &   $EW$ &   $NS$ &   $EW$ &   $EW$ &   $NS$ &   $EW$    \\ 
\hline 
 {\color{red!50!blue} $e$} &  $EW$ &   $EW$ &   $NS$ &   $NS$ &   $NS$ &   $EW$ &   $EW$ &   $NS$    \\ 
\hline 
 {\color{red!50!blue} $f$} &  $NS$ &   $EW$ &   $NS$ &   $EW$ &   $NS$ &   $NS$ &   $EW$ &   $NS$    \\ 
\hline 
 {\color{red!50!blue} $g$} &  $NS$ &   $EW$ &   $NS$ &   $EW$ &   $EW$ &   $NS$ &   $EW$ &   $NS$    \\ 

\hline 
 {\color{red!50!blue} $h$} &  $EW$ &   $NS$ &   $EW$ &   $EW$ &   $EW$ &   $EW$ &   $NS$ &   $EW$   \\ 

\hline 
 {\color{red!50!blue} $i$} &  $EW$ &   $NS$ &   $EW$ &   $EW$ &   $NS$ &   $EW$ &   $NS$ &   $EW$   \\ 

\hline 
 {\color{red!50!blue} $j$} &  $EW$ &   $EW$ &   $NS$ &   $NS$ &   $EW$ &   $EW$ &   $EW$ &   $NS$   \\ 

\hline 
 {\color{red!50!blue} $k$} &  $EW$ &   $EW$ &   $NS$ &   $NS$ &   $NS$ &   $EW$ &   $EW$ &   $NS$   \\ 

\hline 
 {\color{red!50!blue} $l$} &  $NS$ &   $EW$ &   $EW$ &   $NS$ &   $NS$ &   $NS$ &   $EW$ &   $EW$   \\ 

\hline 
 {\color{red!50!blue} $m$} &  $EW$ &   $NS$ &   $NS$ &   $EW$ &   $NS$ &   $EW$ &   $NS$ &   $NS$   \\ 

\hline 
 {\color{red!50!blue} $n$} &  $NS$ &   $NS$ &   $EW$ &   $NS$ &   $EW$ &   $NS$ &   $NS$ &   $EW$   \\ 

\hline {\color{red!50!blue} $o$} &  $ 0.0 $ &   $ 0.0 $ &   $ 0.0 $ &   $ 12.8 $ &   $ 0.0 $ &   $ 0.0 $ &   $ 0.0 $ &   $ 0.0 $   \\ 

\hline {\color{red!50!blue} $p$} &  $ 4.0 $ &   $ 14.0 $ &   $ 0.0 $ &   $ 9.5 $ &   $ 0.0 $ &   $ 4.0 $ &   $ 11.5 $ &   $ 0.0 $   \\ 

\hline {\color{red!50!blue} $q$} &  $ 0.0 $ &   $ 0.0 $ &   $ 10.0 $ &   $ 0.0 $ &   $ 2.5 $ &   $ 0.0 $ &   $ 0.0 $ &   $ 10.0 $   \\ 
\hline {\color{red!50!blue} $r$} &  $ 5.5 $ &   $ 0.0 $ &   $ 4.0 $ &   $ 14.0 $ &   $ 11.5 $ &   $ 5.5 $ &   $ 0.0 $ &   $ 4.0 $   \\ \hline
\end{tabular}
}
\vspace{1pt}
\label{table:controls}
\end{table}

\subsubsection*{Robust MPC}
Here it is assumed that the controller has full state knowledge. We apply the techniques developed in Sec. \ref{sec:mpc}. Using the result from the previous section, the set $\Omega_{\mathcal{L}^\varphi}$ is constructed in \revthree{$\mathbb{R}_+^{212}$ ($=\mathbb{R}^{n(h^\varphi+1)},n=53, h^\varphi=3$)}. The cost criteria that we use in this case study is the total delay induced in the network over the planning horizon $H$. A vehicle is delayed by one time unit if it can not flow out of a link in one time step, which may be because of the actuation (e.g., red light) or waiting for the flow of other vehicles in the same link (i.e., we have $x_{[l]}\ge c_l$). 
We are also interested in maximizing the STL robustness score. The cost function is: 

\begin{equation}
\begin{array}{l}
\label{eq:traffic_cost}
J_{\text{traffic}}(x^H,u^H):= - \revtwo{\zeta} ~\rho(\bolds{x},\user{\bolds{G}_{[0,H-1]}}\varphi,t-h^\varphi+1) \\ +  \displaystyle \sum_{k=0}^{H-1} \gamma^{k} \displaystyle \sum_{l \in \mathcal{L}} (x_{[l],t+k}-\vec{q}_{[l],t+k}) ,
\end{array}
\end{equation} 
where $\vec{q}_{[l]}$, given by \eqref{eq:actuated}, is the amount of vehicles that flow out of link $l$, $\gamma$ is the discount factor for delays predicted in further future, and $\zeta$ is a positive weight for robustness. Notice the connection between the time window of STL robustness score in \eqref{eq:traffic_cost} and MPC constraint enforcement in \eqref{eq:recursive_optimization}.
It follows from Theorem \ref{theorem:monotonicty_of_traffic} and STL quantitative semantics \eqref{equ:quant} that the cost function above is non-decreasing with respect to the state in $\Pi$. Therefore, in order to minimize the worst case cost, the maximal system is considered in the MPC optimization problem.  

Starting from zero initial conditions, we implement the MPC algorithm \eqref{eq:MPC_recursive} with $H=3$ for 40 time steps. We set $\zeta=1000$, $\gamma=0.5$ in \eqref{eq:traffic_cost}. The disturbances at each time step were randomly drawn from $L(w^*)$ using a uniform distribution. The maximum computation time for each MPC step time step was less than 0.8 seconds (less than 0.5 seconds on average). The resulting trajectory is shown in Fig. \ref{fig:case_results} [Middle]. For the same sequence of disturbances, the trajectory resulted from applying the open-loop control policy \eqref{eq:phi} (using the values in Table \ref{table:controls}) is shown in Fig. \ref{fig:case_results} [Bottom]. Both trajectories satisfy the specification. However, robust MPC has obviously better performance when costs are considered. The total delay accumulated over 40 time steps is: 
$$J_{40}=\sum_{\tau=0}^{40} \sum_{l \in \mathcal{L}} (x_{[l],\tau}-\vec{q}_{[l],\tau}).$$
The cost above obtained from applying robust MPC was $J_{40}=1843$, while the one for the open-loop control policy was $J_{40}=2299$, which demonstrates the usefulness of the state knowledge in planning controls in a more optimal way. An optimal tuning of parameters $\eta$ and $\gamma$ requires an experimental study which is out of scope of this paper. We only remark that we usually obtained larger delays with non-zero $\eta$, which shows that including STL robustness score in the MPC cost function may be useful even though the ultimate goal is minimizing the total delay. 

It is worth to note that we also tried implementing the MPC algorithm (for the case $\bolds{w}=(w^*)^\omega$, or the maximal system) without the terminal constraints, as in \eqref{eq:optimization}. The MPC got infeasible at $t=8$. The violating constraints were those in $x \in \Pi$. This observation indicates that the myopic behavior of MPC in \eqref{eq:optimization}, when no additional constraints are considered, can lead to congestion in the network. 


\begin{figure*}[t]
\centering
{\includegraphics[width=0.199\textwidth]{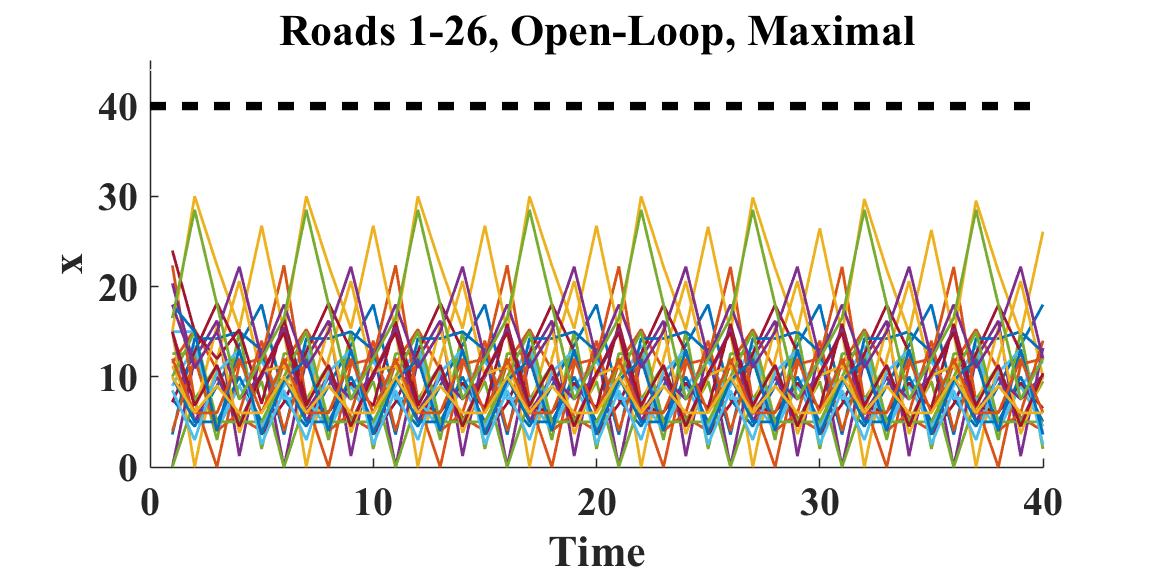}\includegraphics[width=0.199\textwidth]{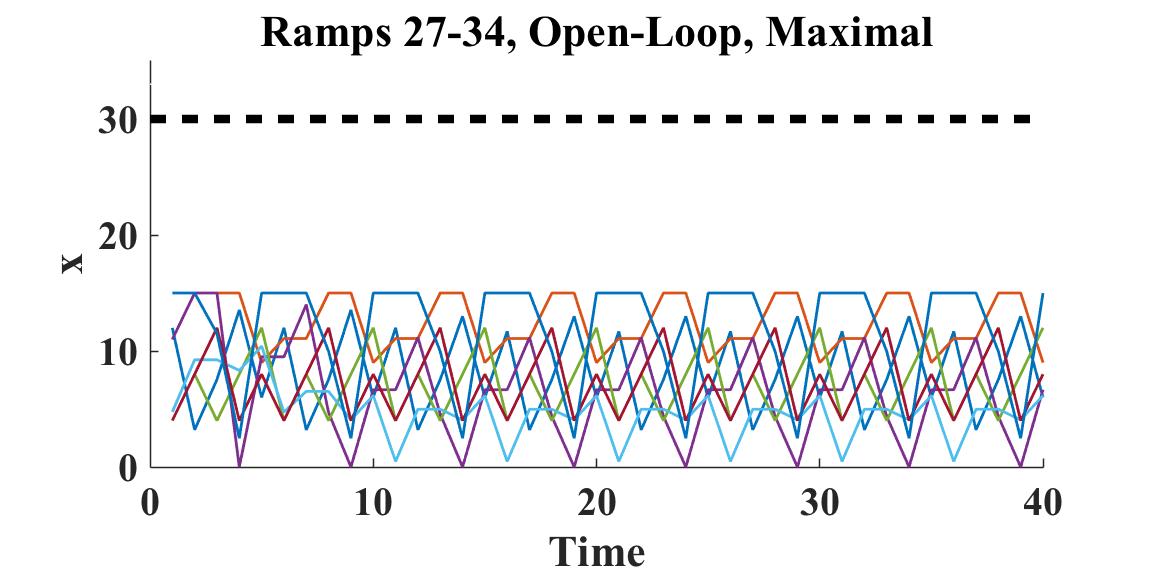}\includegraphics[width=0.199\textwidth]{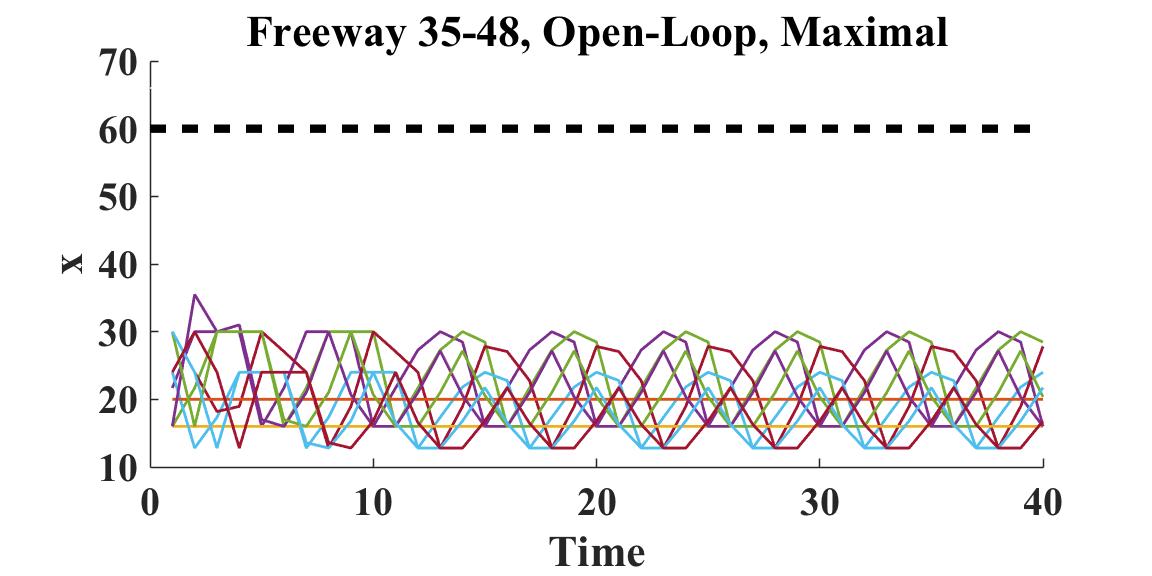}\includegraphics[width=0.199\textwidth]{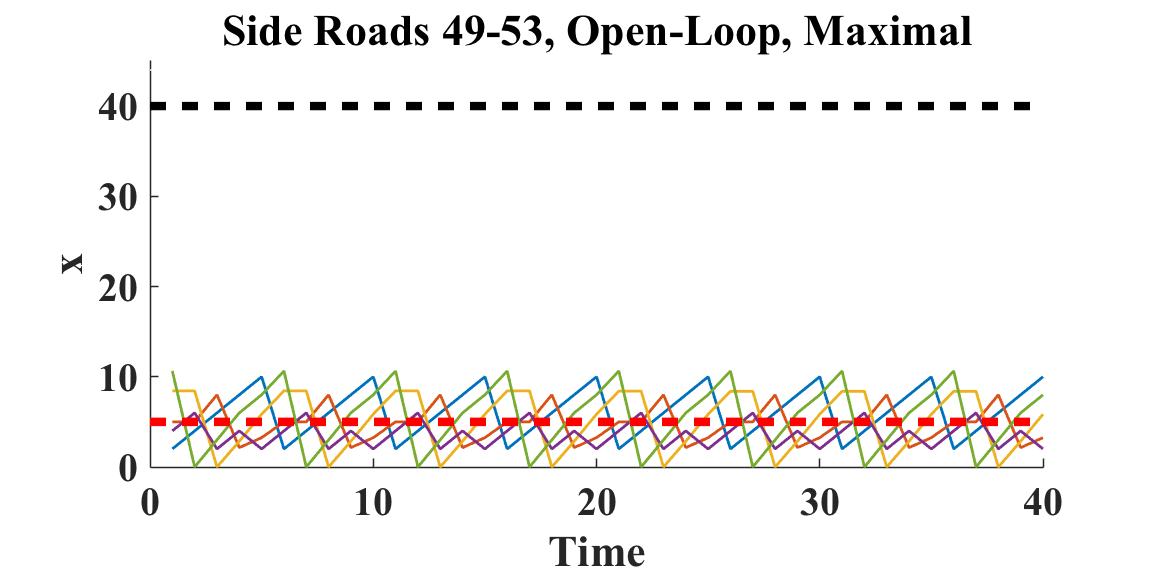}\includegraphics[width=0.199\textwidth]{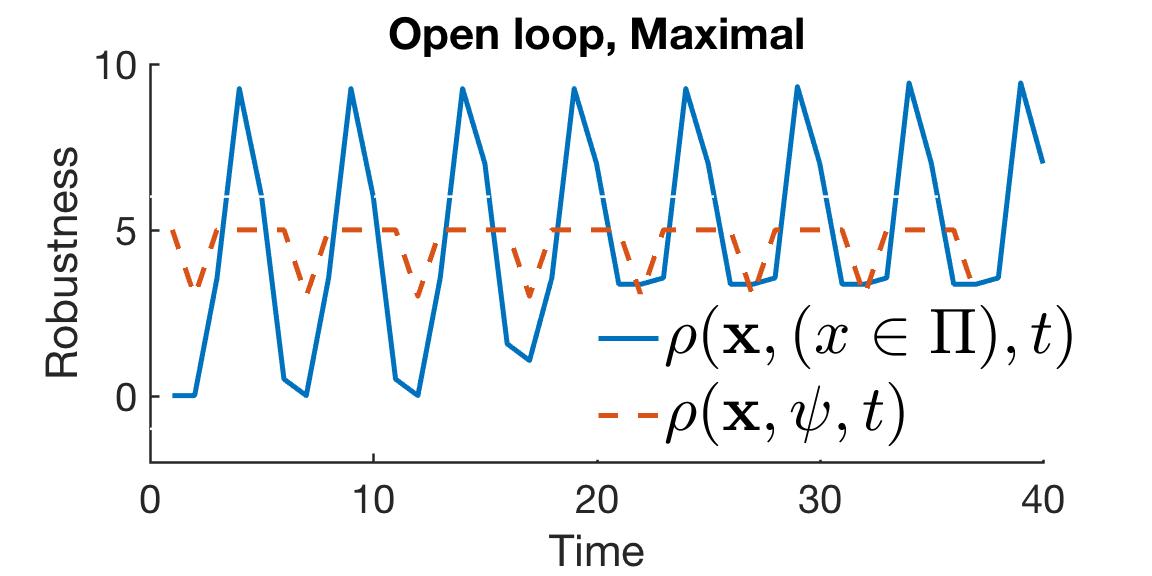}
}
{\includegraphics[width=0.199\textwidth]{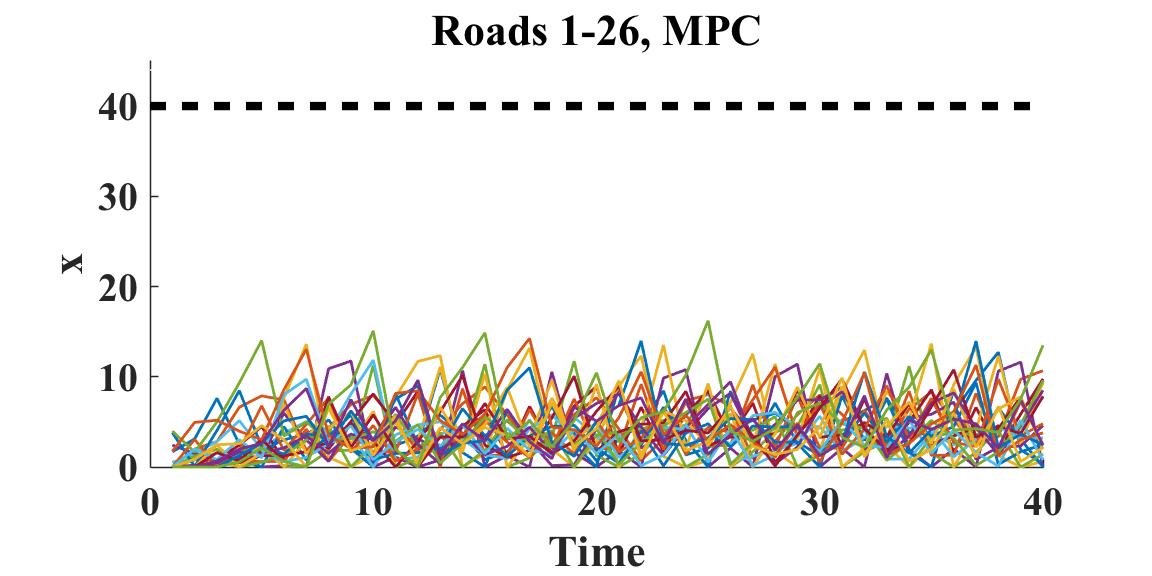}\includegraphics[width=0.199\textwidth]{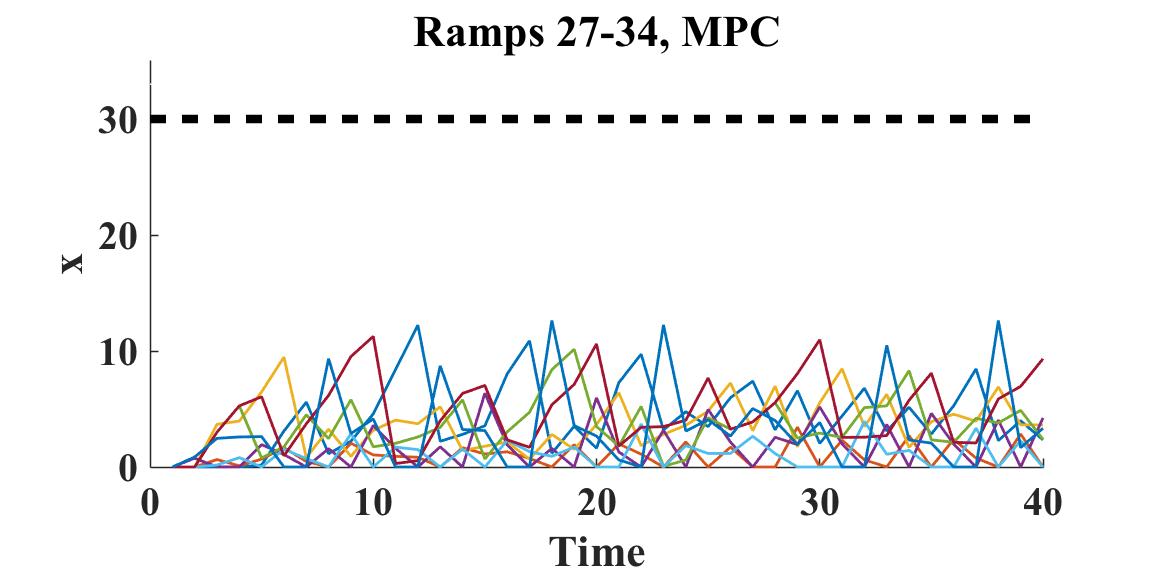}\includegraphics[width=0.199\textwidth]{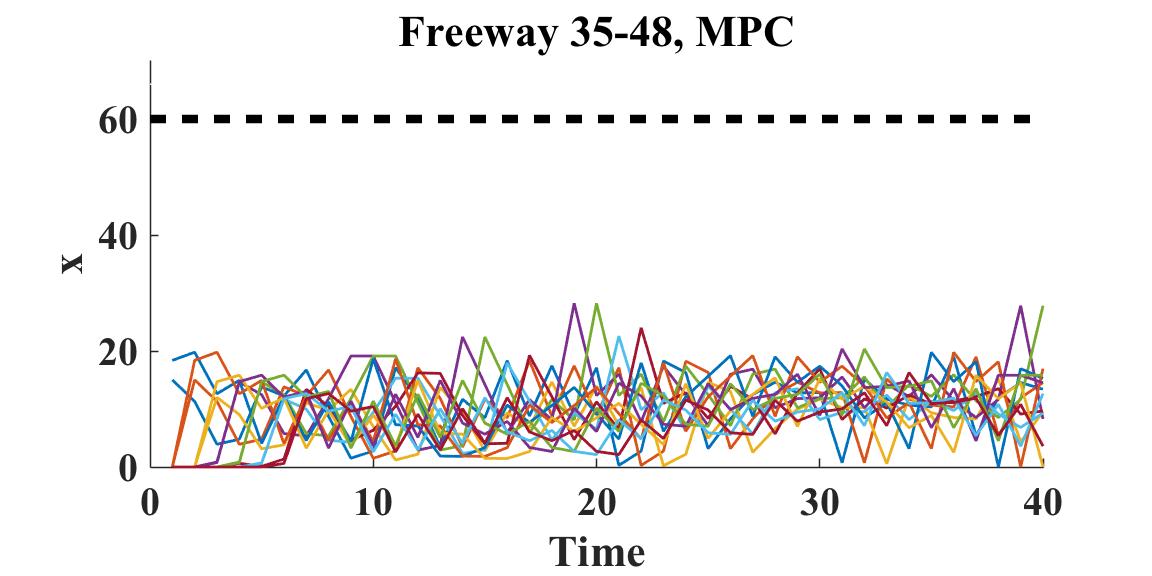}\includegraphics[width=0.199\textwidth]{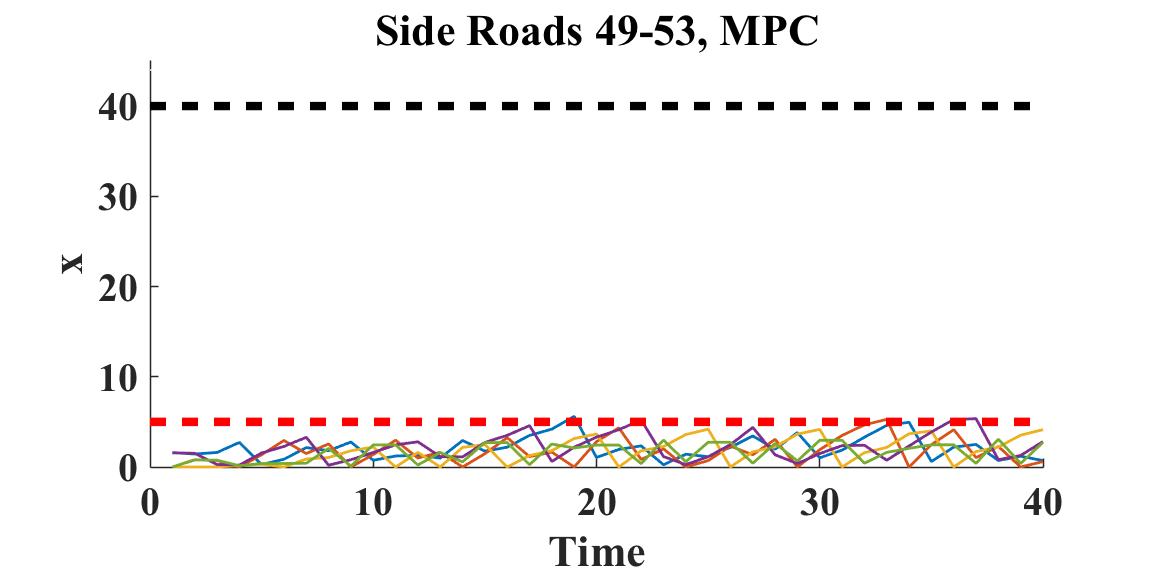}\includegraphics[width=0.199\textwidth]{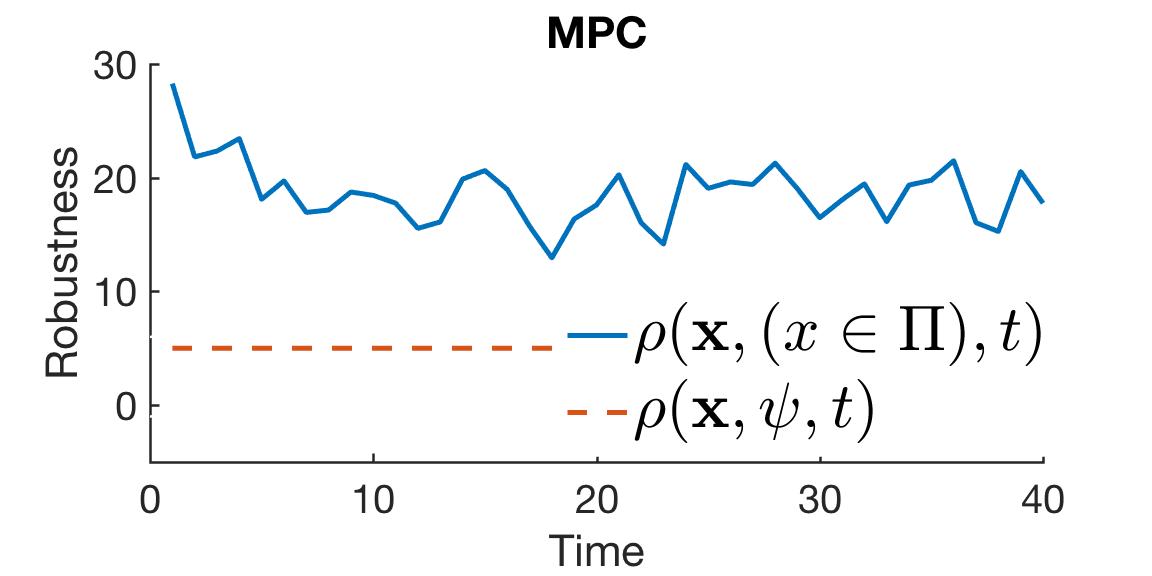}
}
{\includegraphics[width=0.199\textwidth]{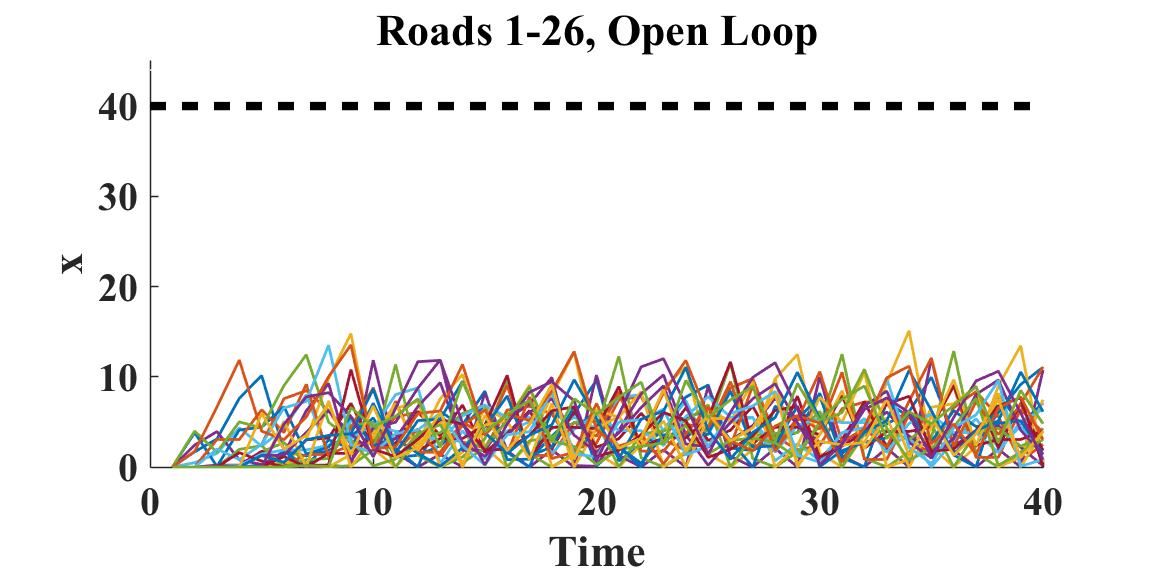}\includegraphics[width=0.199\textwidth]{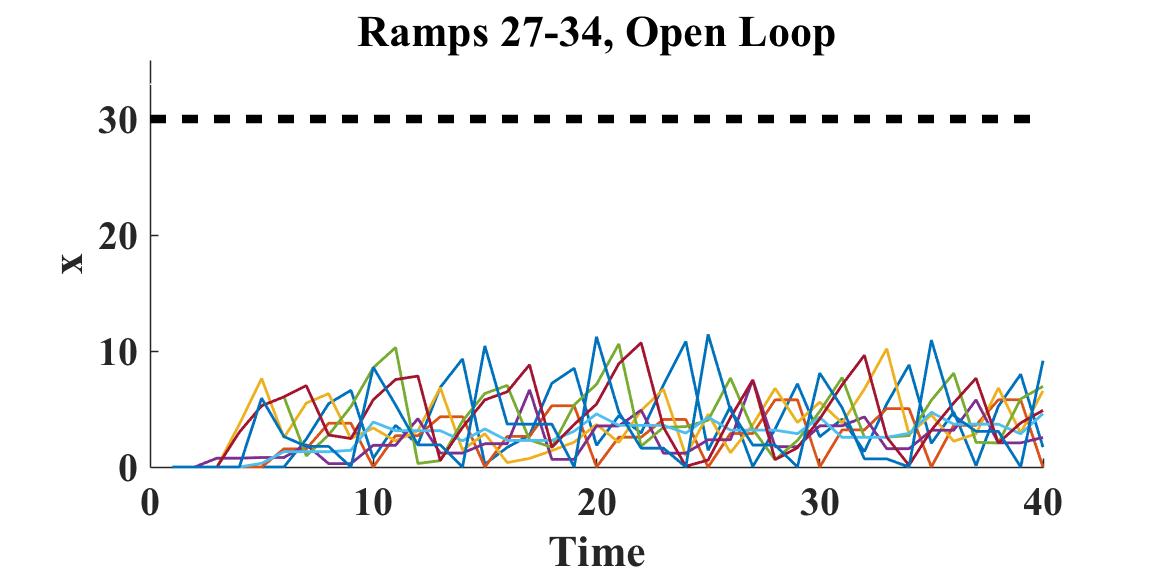}\includegraphics[width=0.199\textwidth]{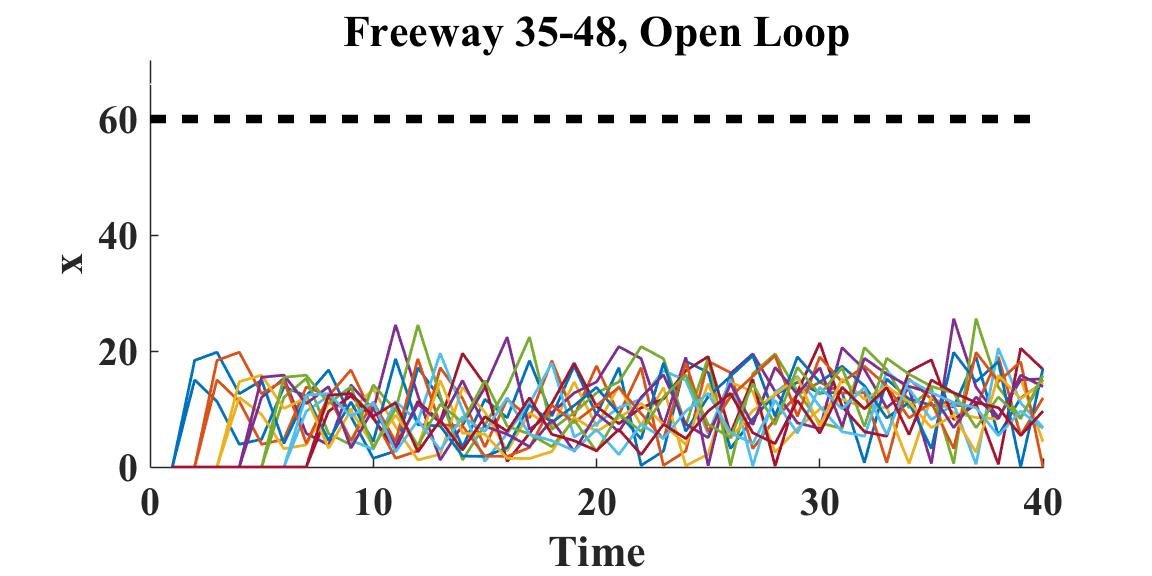}\includegraphics[width=0.199\textwidth]{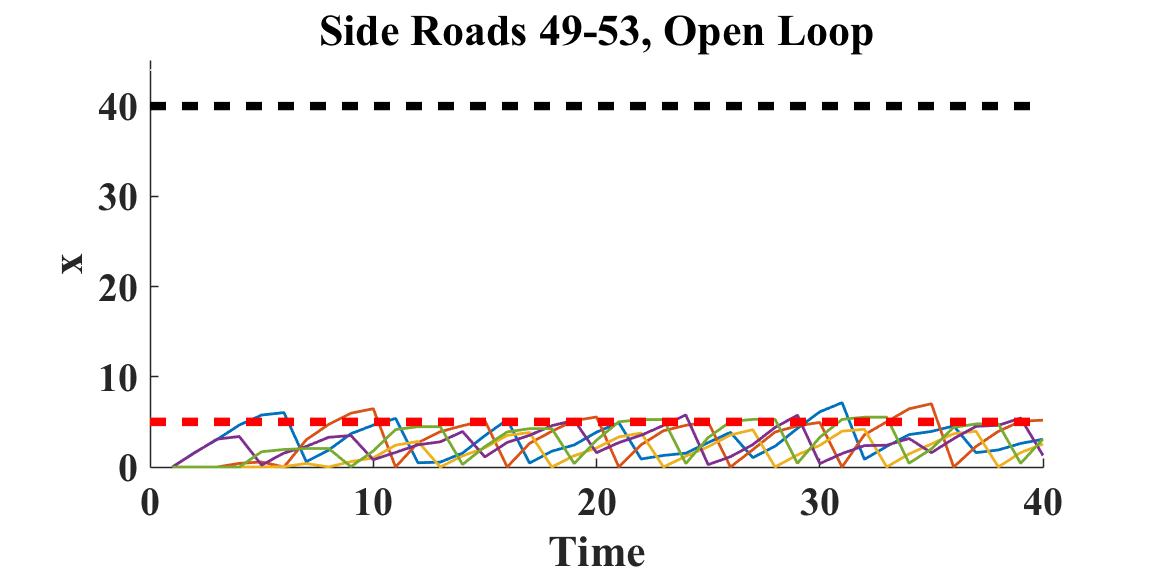}\includegraphics[width=0.199\textwidth]{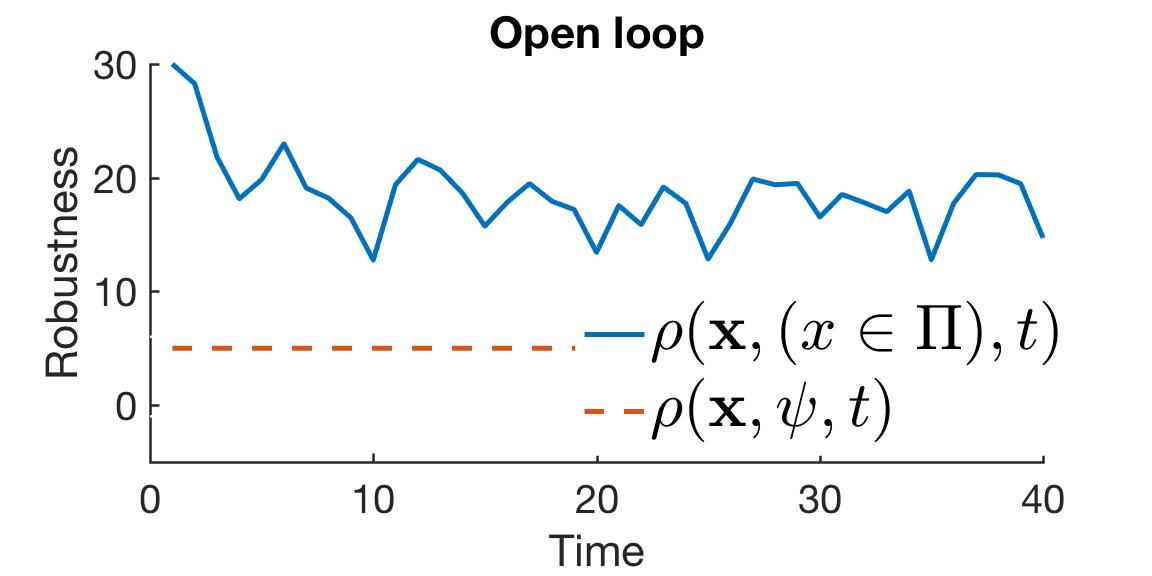}
}
\caption{Traffic management case study: [Top Row] the trajectory of the maximal system obtained from applying the open-loop control policy \eqref{eq:phi} with initial condition $x_0^\phi$ [Middle Row] robust MPC generated trajectory with zero initial condition with disturbances chosen uniformly $L(w^*)$ [Bottom Row] trajectory generated from applying the open-loop control policy \eqref{eq:phi} with zero initial conditions and the same disturbances as in [Middle Row].  
}
\label{fig:case_results}
\end{figure*}

\section{Conclusion and Future Work}
We developed methods to control positive monotone discrete-time systems from STL specifications. We showed that open-loop control sequences are sufficient and (almost) necessary for guaranteeing the correctness of STL specifications. A robust MPC method was introduced to plan controls optimally, while guaranteeing global STL specifications. We showed the usefulness of our results on traffic management.  

Future work will focus on non-monotone systems with parametric uncertainty whose state evolution can be over-approximated in an appropriate way using monotone systems. We will develop adaptive control schemes to tune parameters automatically using the data gathered from the evolution of the system. This will eventually lead to data-driven control techniques for transportation networks with formal guarantees.          

\bibliography{references2}

\revtwo{
\section*{Appendix}
\begin{theorem}
Let $\mathbb{S}$ be the set of all STL formulas that can be written in the form:
\begin{equation}
\label{eq_safety_formula}
\phi=\bigvee_{i=1}^{n_\phi} \varphi_{b,i} \wedge {\bf G}_{[\Delta_i,\infty]} \varphi_{g,i},
\end{equation}
where $\varphi_{b,i},\Delta_i \ge h^{\varphi_{b,i}},\varphi_{g,i}, i=1,\cdots,n_\phi,$ are bounded STL formulas. Then $\mathbb{S}$ is a subset of safety STL formulas that is closed under STL syntax with bounded temporal operators.   
\end{theorem}
\begin{IEEEproof}
First, a quick inspection of \eqref{eq_safety_formula} verifies that it is a safety STL formula. A predicate $\pi$ is a bounded formula (with zero horizon) and is a special case of \eqref{eq_safety_formula}, hence $\pi \in \mathbb{S}$.   

We also have the following property that relaxes the form in \eqref{eq_safety_formula}:
\emph{For all bounded STL formulas $\varphi_1,\varphi_2$,  we have $\varphi_1 \wedge \bolds{G}_{[\Gamma,\infty)} \varphi_2 \in \mathbb{S}$, $\forall \Gamma \in \mathbb{N}$.} Proof: The case for $\Gamma\ge h^{\varphi_{1}}$ is already in the form \eqref{eq_safety_formula} with $n_\phi=1$. If $\Gamma< h^{\varphi_{1}}$, we write $\bold{G}_{[\Gamma,\infty)}\varphi_2 = \bold{G}_{[\Gamma,h^{\varphi_1}]}\varphi_2 \wedge \bold{G}_{[h^{\varphi_1},\infty)}\varphi_2$. Now, define $\varphi_1 \wedge \bold{G}_{[\Gamma,h^{\varphi_1}]}\varphi_2$ as the new bounded formula and retain the form in \eqref{eq_safety_formula} with $n_\phi=1$. 

We show that $\mathbb{S}$ is closed under STL syntax with bounded operators. The distributivity properties of Boolean connectives and temporal operators (see, e.g., \cite{huth2004logic}) imply that:
 $\phi_1 \vee (\phi_2 \wedge \phi_3)= (\phi_1 \vee \phi_2) \wedge (\phi_2 \vee \phi_3)$,
$\phi_1 \wedge (\phi_2 \vee \phi_3)= (\phi_1 \wedge \phi_2) \vee (\phi_2 \wedge \phi_3)$,
$\bolds{F}_I (\phi_1 \vee \phi_2)=(\bolds{F}_I \phi_1) \vee (\bolds{F}_I \phi_2)$, and 
$\bolds{G}_I (\phi_1 \wedge \phi_2)=(\bolds{G}_I \phi_1) \wedge (\bolds{G}_I \phi_2)$,
where $\phi_1,\phi_2,\phi_3$ are temporal logic formulas and $I$ is an interval.
\begin{enumerate}
\item  $\phi_1, \phi_2 \in \mathbb{S} \Rightarrow \phi_1 \wedge \phi_2 \in \mathbb{S}, \phi_1 \vee \phi_2 \in \mathbb{S}$: this result easily follows from the distributivity properties of Boolean connectives mentioned above.

\item $\phi \in \mathbb{S} \Rightarrow \bolds{F}_{\{t\}} \phi  \in \mathbb{S}$: we use $\bolds{F}_{\{t\}}\bolds{G}_{[a,b]}=\bolds{G}_{[t+a,t+b]}$ and distributivity to have (note that $\bolds{F}_{\{t\}}=\bolds{G}_{\{t\}}$) 
\begin{equation*}
\begin{array}{rl}
& \bolds{F}_{\{t\}} (\bigvee_{i=1}^{n_\phi} (\varphi_{b,i} \wedge \bolds{G}_{[\Gamma_i,\infty]}\varphi_{g,i})) \\
= & \bigvee_{i=1}^{n_\phi} (\bolds{F}_{\{t\}} \varphi_{b,i} \wedge \bolds{G}_{[t+\Gamma_i,\infty]}\varphi_{g,i}).
\end{array}
\end{equation*}
Introducing $\bolds{F}_{\{t\}} \varphi_{b,i}, i=1,\cdots,n_\phi$, as new bounded STL formulas leads to the form in \eqref{eq_safety_formula}.

\item $\phi \in \mathbb{S} \Rightarrow \bolds{F}_{[a,b]} \phi  \in \mathbb{S}, \bolds{G}_{[a,b]} \phi  \in \mathbb{F}$: use $\bolds{F}_{[a,b]}=\bigvee_{t \in [a,b]}\bolds{F}_{\{t\}}$ and $\bolds{G}_{[a,b]}=\bigwedge_{t \in [a,b]}\bolds{F}_{\{t\}}$ to convert temporal operators to Boolean connectives.


\item  $\phi_1,\phi_2 \in \mathbb{S} \Rightarrow \phi_1\bolds{U}_{[a,b]} \phi_2  \in \mathbb{S}$:  use the STL semantics \eqref{equ:semantics} to substitute the bounded ``until" operator using bounded ``eventually" and bounded ``always" operators:
\begin{equation*}
\begin{array}{r}
\phi_1\bolds{U}_{[a,b]} \phi_2 = \bigvee_{t \in [a,b]} (\bolds{G}_{[a,t]\phi_1} \wedge \bolds{F}_{\{t\}} \phi_2).
\end{array}
\end{equation*}
\end{enumerate}
\begin{example}
The ``reach and stay" formula $\bolds{F}_I \bolds{G}_{[0,\infty)} \varphi $, where $\varphi$ is a bounded formula, is equivalent to $\bigvee_{t \in I} \bolds{G}_{[t,\infty)} \varphi$.    
\end{example}
\begin{remark}
What remains to show that $\mathbb{S}$ is equivalent to the set of all safety STL formulas is having that $\phi \in \mathbb{S} \Rightarrow \bolds{G}_{[\userfinal{\Gamma},\infty)} \phi  \in \mathbb{S}, \forall \Gamma \in \mathbb{N}$, which is not true by restricting $n_\phi$ in \eqref{eq_safety_formula} to be finite. Formulas that involve nested unbounded ``always" operator and can not be further simplified, such as $\bolds{G}_{[\Gamma',\infty)} (\varphi_1 \vee \bolds{G}_{[\Gamma,\infty)} \varphi_2)$, are rarely encountered in applications.   
\end{remark}

\end{IEEEproof}
}

\begin{IEEEbiography}[{\includegraphics[width=1in,height=1.25in,clip,keepaspectratio]{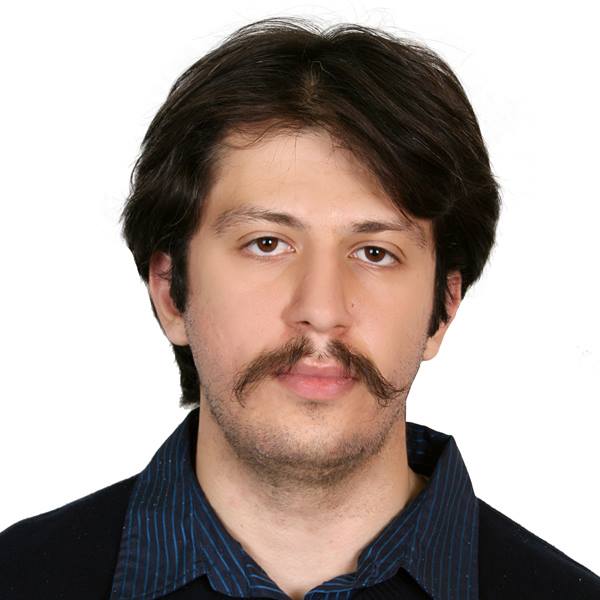}}]{Sadra Sadraddini} (S' 16) 
received the B.Sc. in Mechanical Engineering and the B.Sc. in Aerospace Engineering (dual majors) in 2013 from Sharif University of Technology, Tehran, Iran. 
He is currently pursuing a degree toward Ph.D. in Mechanical Engineering at Boston University, Boston, MA. 
His research focuses on formal methods to control theory with various applications in cyber-physical systems.      
\end{IEEEbiography}

\begin{IEEEbiography}[{\includegraphics[width=1in,height=1.25in,clip,keepaspectratio]{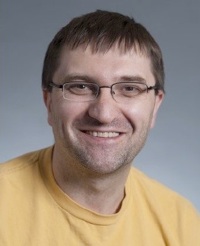}}]{Calin Belta} 
 (F' 17) is a Professor in the Department of Mechanical Engineering at Boston University, where he holds the Tegan family Distinguished Faculty Fellowship. He is the Director of the BU Robotics Lab, and is also affiliated with the Department of Electrical and Computer Engineering, the Division of Systems Engineering at Boston University, the Center for Information and Systems Engineering (CISE), and the Bioinformatics Program.  
His research focuses on dynamics and control theory, with particular emphasis on hybrid and cyber-physical systems, formal synthesis and verification, and applications in robotics and systems biology.  He received the Air Force Office of Scientific Research Young Investigator Award and the National Science Foundation CAREER Award. He is a fellow of IEEE.
\end{IEEEbiography}

\end{document}